\documentclass[
,aps
,pra
,amsmath
,eqsecnum
,amssymb
,amsfonts
,twocolumn
,letterpaper
,10pt
,tightenlines
,superscriptaddress
,nofootinbib
,eqsecnum
,longbibliography
,notitlepage
]{revtex4-2}
\nonstopmode

\usepackage{nag}
\usepackage{graphicx}
\usepackage{amsthm}
\usepackage{tabulary}
\usepackage{threeparttable}
\usepackage{qcircuit}
\usepackage{mathtools}
\usepackage{bm}
\usepackage[table,xcdraw]{xcolor}
\usepackage{braket}
\usepackage{datetime}
\usepackage{stackrel}
\usepackage{stmaryrd}
\usepackage{dsfont}
\usepackage{qcircuit}
\usepackage[english]{babel}
\usepackage{units}

\usepackage{tikz}
\usetikzlibrary{backgrounds,decorations.pathreplacing}

\usepackage{chngcntr}
\counterwithout{equation}{section}

\usepackage[normalem]{ulem}
\usepackage{slashed}
\usepackage{soul}

\usepackage{xcolor}
\usepackage[
    colorlinks=true,
    urlcolor=blue,
    linkcolor=blue,
    citecolor=blue,
    filecolor=blue,
]{hyperref}
\usepackage[caption=false]{subfig}
\usepackage{amssymb}
\usepackage{amsmath}
\usepackage{amsthm}
\usepackage{braket}

\usepackage[OT2,T1]{fontenc}
\DeclareSymbolFont{cyrletters}{OT2}{wncyr}{m}{n}
\DeclareMathSymbol{\Sha}{\mathalpha}{cyrletters}{"58}

\usepackage{makeidx}
\makeindex

\usepackage{soul}
\usepackage{xargs}
\newcommandx{\cmnote}[2][1=]{\linespread{1.0}\todo[linecolor=red,backgroundcolor=red!25,bordercolor=red,#1]{#2}}

\let\underline\ul

\allowdisplaybreaks[4]

\DeclareMathOperator{\diag}{diag}

\newtheorem{theorem}{Theorem}

\newtheorem{definition}{Definition}

\let\originalleft\left
\let\originalright\right
\renewcommand{\left}{\mathopen{}\mathclose\bgroup\originalleft}
\renewcommand{\right}{\aftergroup\egroup\originalright}

 \index{} \index{} \index{} \index{} \index{} \index{} \index{} \index{}%
\makeatletter

\newcommand{\ringplus}{\mathbin{\text{\@ringplus}}}

\newcommand{\@ringplus}{%
  \ooalign{\hidewidth\raise1.3ex\hbox{\tiny$\circ$}\hidewidth\cr$\m@th+$\cr}%
}

\newcommand{\ringminus}{\mathbin{\text{\@ringminus}}}

\newcommand{\@ringminus}{%
  \ooalign{\hidewidth\raise0.9ex\hbox{\tiny$\circ$}\hidewidth\cr$\m@th-$\cr}%
}
\makeatother
 \index{} \index{} \index{} \index{} \index{} \index{} \index{} \index{}%

\newcommand{\tp}[0]{\mathrm{T}}

\newcommand{\EPR}{\text{EPR}}

\newcommand{\GKP}{\text{GKP}}

\newcommand{\qunaught}{\varnothing}

\DeclareFontFamily{U}{wncy}{}
\DeclareFontShape{U}{wncy}{m}{n}{<->wncyr10}{}
\DeclareSymbolFont{mcy}{U}{wncy}{m}{n}
\DeclareMathSymbol{\Sh}{\mathord}{mcy}{"58}

\newcommand{\negspace}{\!}

\newcommand{\lsub}[2]{{\protect\vphantom{#1}}_{#2} \negspace {#1}}
\newcommand{\rsub}[2]{{#1} \negspace {\protect\vphantom{#1}}_{#2}}
\newcommand{\lrsub}[3]{{\protect\vphantom{#1}}_{#2} \negspace {#1} \negspace {\protect\vphantom{#1}}_{#3}}

\newcommand{\ketsub}[2]{\rsub {\ket{#1}} {#2}}
\newcommand{\brasub}[2]{\lsub {\bra{#1}} {#2}}

\newcommand{\pbra}[1]{\brasub{#1} p}

\newcommand{\qket}[1]{\ketsub{#1} q}

\newcommand{\inprod}[2]{\left\langle {#1} | {#2} \right\rangle}

\newcommand{\inprodsubsub}[4]{\lrsub {\inprod{#1}{#2}} {#3} {#4}}

\newcommand{\pinprod}[2]{\inprodsubsub{#1}{#2} p p}

\newcommand{\abs}[1]{\left\lvert{#1}\right\rvert}
\newcommand{\abss}[1]{\lvert{#1}\rvert}

\newcommand{\reals}[0]{\mathbb{R}}

\newcommand{\op}[1]{\hat{#1}}

\newcommand{\opvec}[1]{\op{\vec{#1}}}

\newcommand{\mat}[1]{\bm{\mathrm{#1}}}

\renewcommand{\vec}[1]{\bm{#1}}

\newcommand{\controlled}[1]{\op{\mathrm{C}}_{#1}}
\newcommand{\CZ}[0]{\controlled Z}
\newcommand{\CX}[0]{\controlled X}

 \index{} \index{} \index{} \index{} \index{} \index{} \index{} \index{}%
 \index{} \index{} \index{} \index{} \index{} \index{} \index{} \index{}%

\newcommand{\blu}{\color{blue!70!cyan}}
\newcommand{\blk}{\color{black}}

 \index{} \index{} \index{} \index{} \index{} \index{} \index{} \index{}%
 \index{} \index{} \index{} \index{} \index{} \index{} \index{} \index{}%
\usepackage{graphicx}

\newcommand{\bsop}{\op{B} }

\usepackage{graphicx}
\usepackage{multirow}
\usepackage{hhline}

\xyoption{color}
\xyoption{line}

\newcommandx*\bsbal[3][1=black, 3=->]{\ar @[#1]@{#3} [#2,0] \qw}

\newcommandx*\varbs[4][1=black, 3=\theta, 4=->]{\ar @[#1]@{#4}^{#3} [#2,0] \qw}
\newcommandx*\varbsleft[5][1=black, 3=, 4=->]{\ar @[#1]@{#4}^{#3}_{#5} [#2,0] \qw}

\newcommandx*\ctrlg[2]{\control \ar @{-}^{#1} [#2,0] \qw}
\newcommandx*\ctrlog[2]{\controlo \ar @{-}^{#1} [#2,0] \qw}

\newcommand{\EPRdr}{\ar @{-} [dr(0.5)]\qw }
\newcommand{\EPRur}{\ar @{-} [ur(0.5)]\qw }
\newcommand{\EPRul}{\ar @{-} [ul(0.5)] }
\newcommand{\EPRdl}{\ar @{-} [dl(0.5)] }

\newcommand{\nogate}[1]{*+<.6em>{#1} \POS ="i","i"+UR;"i"+UL **[white]\dir{-};"i"+DL **[white]\dir{-};"i"+DR **[white]\dir{-};"i"+UR **[white]\dir{-},"i" }

\makeatletter
 \newcommand{\xmapsfrom}[2][]{%
    \ext@arrow3095\leftarrowfill@{#1}{#2}\mapsfromchar
}
\makeatother

\newcommand{\QRL}{\text{QRL}}

\newcommand{\BSmat}{\mat{R}}
\newcommand{\FSmat}{\mat{R}}

\newcommand{\swap}{\text{SWAP}}

\usepackage{comment}
\usepackage{placeins}
\usepackage{flushend}
\begin{document}
\title{Equivalent noise properties of scalable continuous-variable cluster states}
\newpage
\setcounter{page}{1}
\pagenumbering{arabic}

\author{Blayney W. Walshe}
\email{blayneyw@gmail.com}
\affiliation{Centre for Quantum Computation and Communication Technology, School of Science, RMIT University, Melbourne, VIC 3000, Australia}
\author{Rafael N. Alexander}
\affiliation{Centre for Quantum Computation and Communication Technology, School of Science, RMIT University, Melbourne, VIC 3000, Australia}
\author{Takaya Matsuura}
\affiliation{Centre for Quantum Computation and Communication Technology, School of Science, RMIT University, Melbourne, VIC 3000, Australia}
\author{Ben Q. Baragiola}
\affiliation{Centre for Quantum Computation and Communication Technology, School of Science, RMIT University, Melbourne, VIC 3000, Australia}
\affiliation{Yukawa Institute for Theoretical Physics, Kyoto University, Kitashirakawa Oiwakecho, Sakyo-ku, Kyoto 606-8502, Japan}
\author{Nicolas C. Menicucci}
\affiliation{Centre for Quantum Computation and Communication Technology, School of Science, RMIT University, Melbourne, VIC 3000, Australia}
\begin{abstract}
   
    Optical continuous-variable cluster states (CVCSs) 
    in combination with Gottesman-Kitaev-Preskill~(GKP) qubits enable fault-tolerant quantum computation so long as these
    resources are of high enough quality. Previous studies concluded that a particular CVCS, the quad rail lattice~(QRL), 
    exhibits lower GKP gate-error rate than others do.
    We show in this work that many other experimentally accessible CVCSs also achieve this level of performance
    by identifying operational equivalences to the QRL. Under this equivalence, the GKP Clifford gate set for each CVCS maps straightforwardly from that of the QRL, inheriting 
    its noise properties.
    Furthermore, each cluster state has at its heart a balanced four-splitter---the four-mode extension to a balanced beam splitter. We classify all four-splitters, show they form a single equivalence class under SWAP and parity operators, and we give a construction of any four-splitter with linear optics, thus extending the toolbox for theoretical and experimental cluster-state design and analysis.
    
\end{abstract}

\maketitle

\section{Introduction}
\label{intro}

 Measurement-based quantum computing (QC) 
with continuous-variable (CV) systems relies on preparing 
large-scale entangled resource states, called CV cluster states~\cite{Menicucci2006,Gu2009} on which adaptive measurements are sufficient for universal QC.
Optics has shown itself to be an excellent platform for demonstrating these states on a large scale
using squeezed vacuum  and passive linear optical components~\cite{Flammia2009optical, Menicucci2011tempmodeCVCS, Wang2014hypercubic, Alexander2016oneway, Alexander2018BSL, Wu2020scalable,Yang2020spatiotemporalgraphs}.
These cluster states are accessible with current technology and highly scalable. Indeed, experiments have demonstrated large cluster states~\cite{Chen2014,Yokoyama2013ultra,Yoshikawa2016} including 2D CV cluster states whose connectivity makes them resources for universal quantum computing~\cite{Asavanant2019detCVCS,Larsen2019detCVCS}.

In measurement-based quantum computing (MBQC) protocols using CV cluster states, performing homodyne detection on local modes teleports quantum information throughout the cluster and applies deterministic Gaussian operations that depend on the measurement bases~\cite{Loock2007GaussgatesCVCS, Asavanant2020detgatesCVCS,Larsen2020detgatesCVCS}.
Hindering quantum computation is inherent
noise from the fact that physical states can only be finitely squeezed%
~\cite{Alexander2014noise, Walshe2019}. 
As measurements are performed, this noise accumulates and eventually overwhelms the intended calculation~\cite{ohliger2010limitations}. As a countermeasure, digital quantum information can be encoded into a Gottesman-Kitaev-Preskill~(GKP) error-correcting code~\cite{GKP}, a class of bosonic code that interfaces seamlessly with CVCSs~\cite{pantaleoni2021hidden} because GKP codes admit all-Gaussian Clifford-gate operations and magic-state production~\cite{Baragiola2019}.
Moreover, GKP codes are resistant to various types of CV errors, including errors from noisy teleportation~\cite{GKP,Noh2019channel,Wu2021GKPgengauss}. Thus, combining CV cluster states with GKP codes provides a pathway to fault tolerance~\cite{menicucci2014fault, fukui2018high, Larsen2021fault, Bourassa2021blueprint, Tzitrin2021passive}.

In order to use the GKP code with a CV cluster state, a specific set of Gaussian operations is required---those that realize GKP Clifford gates. Bespoke methods
have been proposed for some CV cluster states~\cite{Larsen2020architecture, walshe2021streamlined}, each involving a set of homodyne measurement angles tailored for each gate and for each cluster state. For some CV cluster states, this includes position-quadrature measurements to delete certain modes~\cite{Alexander2018BSL,Larsen2020architecture}, effectively wasting them while injecting their noise into the remaining modes. Furthermore, some proposals require optical switches to inject GKP states when error correction is needed~\cite{Larsen2020architecture,Asavanant2019detCVCS}.

A detailed study~\cite{walshe2021streamlined} of a specific CV cluster state---the quad-rail lattice (QRL)~\cite{Menicucci2011tempmodeCVCS, Alexander2016flexible}---showed that it serves as an efficient resource for computing with the square-lattice GKP code: the full set of GKP Cliffords can be performed in a single teleportation step without requiring deletions, thus minimizing inherent noise from the squeezed states in the cluster. Further, replacing the squeezed states with a particular type of GKP state called a \emph{qunaught} state~\cite{Walshe2020,walshe2021streamlined} (also called GKP sensor states~\cite{Duivenvoorden2017}) obviates the need for optical switches and automatically performs GKP error correction during gate execution~\cite{Walshe2020}. Together, these properties give the QRL the lowest logical GKP-Clifford gate noise of any known two-dimensional CV cluster state~\cite{walshe2021streamlined}.
Lower gate noise places lower demands on the quality of the GKP resources, making the QRL  appear to stand out as the CV cluster state of choice in the quest for fault tolerance. But this is not the whole story.

In this work, we show that 
several other two-dimensional CV cluster states can achieve the same performance 
as the QRL.
Previous analysis indicated that the bilayer square lattice (BSL)~\cite{Alexander2018BSL}, the double bilayer square lattice (DBSL)~\cite{Larsen2019detCVCS}, and the modified bilayer square lattice (MBSL)~\cite{Asavanant2019detCVCS} CV cluster states perform worse than than the QRL~\cite{Larsen2020architecture}.
We analyse these cluster states along with another one---a three-dimensional extension to the DBSL~\cite{Larsen2021fault}---and show that each, when modified with an additional beam splitter, is operationally equivalent to the QRL. 
Moreover, for the BSL and DBSL, the additional beam splitter is not even necessary---it can be added virtually by correlating certain specific homodyne measurement angles and post-processing the outcomes. 
This equivalence provides two benefits. 
First, it gives a recipe to map the measurement angles for GKP Clifford gates from the QRL to the other cluster states. Second, every equivalent CV cluster state inherits the QRL's low GKP-logical gate noise. Together, these make the BSL, DBSL, and MBSL in principle as useful as the QRL for achieving fault tolerance with the GKP code.

We further show that the CV cluster states above are members of a large class of scalable CV cluster states whose internal coupling uses balanced \emph{four-splitters}~\cite{Alexander2016flexible}---a generalization of a balanced beam splitter to four modes. 
We characterize all balanced four-splitters and show that they fall into an equivalence class defined by a small set of operations. We provide a procedure to engineer any four-splitter using only linear-optical components---a network of four balanced beam splitters along with SWAP gates and parity operators. 

Section~\ref{sec:bsNetworks} gives properties of beam-splitter networks and introduces quantum optical conventions used throughout this work.
Section~\ref{sec:clusterStates} shows the equivalence of the BSL, DBSL, MBSL, and extension of the DBSL to the QRL cluster state.
Section~\ref{sec:GKPcomputingwithCVCS} analyzes this equivalence when the CV cluster states are used in conjuction with the GKP code.
Finally, Sec.~\ref{sec:four-splitters} classifies the equivalence class of four-splitters and shows how to construct them from linear optical components.

\blu

\blk

\section{Beam-splitter networks}\label{sec:bsNetworks}

A balanced beam splitter mixing modes $j$ and~$k$ is described by the unitary operator
    \begin{align} \label{BSdef}
		\bsop_{jk}\coloneqq & e^{ \frac{\pi}{4} ( \op{a}_j   \op{a}^\dagger_k - \op{a}^\dagger_j   \op{a}_k   ) } 
			 =  e^{- i \frac{\pi}{4}  ( \op{q}_j   \op{p}_k  - \op{p}_j   \op{q}_k  ) },
    \end{align}
where $\op{a}_j$ and $\op{a}^\dagger_k$ are mode annihilation and creation operators satisfying $[\op{a}_j,\op{a}^\dagger_k] = \delta_{j,k}$, and $\op{q}_j = \frac{1}{\sqrt{2}}(\op{a}_j + \op{a}^\dagger_j)$ and $\op{p}_j = \frac{-i}{\sqrt{2}}(\op{a}_j - \op{a}^\dagger_j)$ are the position and momentum quadrature operators.
The beam splitter above is depicted in a right-to-left circuit diagram as
\begin{equation} 
\begin{split} \label{eq:beamsplitter_circuit}
	\raisebox{-1.2em}{$\op{B}_{jk} \, = \, $~~}
         \Qcircuit @C=1.25em @R=2.5em @! 
         {
         	& \bsbal{1} & \rstick{j}  \qw \\
         	& \qw       & \rstick{k} \qw
  		  } 
\end{split}		
\quad 
\end{equation}
with the arrow pointing from mode $j$ to mode $k$.

A beam splitter generates a linear transformation of the mode operators $\op{a}_j$ and $\op{a}_k$ that can be described by a unitary, $2 \times 2$ 
matrix.
For the beam splitter in Eq.~\eqref{BSdef}, this matrix is real and thus orthogonal, and we denote this relationship as
    \begin{align} \label{bsmat}
        \op B_{jk} \to \BSmat_{jk} = \frac{1}{\sqrt{2}} 
        \begin{bmatrix}
        1 & -1  \\
        1 & 1  
    \end{bmatrix},
    \end{align} 
where the arrow ($\to$) is to be interpreted as
\begin{align}
\label{eq:unitaryarrow}
    (\op U \to \mat U)
&\quad\Longrightarrow\quad
    (\op U^\dag \opvec a \op U =
    \mat U \opvec a),
\end{align}
where $\op U$ is a unitary operator representing passive linear-optical elements, $\opvec a = (\op a_1, \dotsc, \op a_N)^\tp$ is a column vector of $N$~annihilation operators, and $\mat U \in \mathrm{U}(N)$ is the unitary matrix acting to linearly combine the annihilation operators, representing this Heisenberg action. (This relation is one directional because the converse is only true up to an overall phase.) 
Note that $\BSmat_{jk}$ may be a submatrix within a larger matrix over more modes.

The complex conjugate of the beam-splitter unitary swaps the indices and results in a transpose of its associated matrix: $(\op{B}_{jk}^\dagger = \op{B}_{kj}) \to (\BSmat_{jk}^\tp = \BSmat_{kj})$. In a circuit diagram, Eq.~\eqref{eq:beamsplitter_circuit}, the complex conjugate changes the direction of the arrow between the wires.

We refer to a series of beam splitters as a \emph{beam-splitter network}. This term is general and allows any beam splitters and any number of modes. In this work, however, we will be concerned with four-mode beam-splitter networks consisting of identical, balanced beam splitters as defined in Eq.~\eqref{BSdef}. 
For example, consider a beam-splitter network given by the unitary operator
    \begin{align}\label{eq:op_QRL}
        \bsop_\text{network} = \bsop_\text{24} \bsop_\text{13} \bsop_\text{34} \bsop_\text{12} ,
    \end{align}
which is described by the circuit\footnote{The fact that the beam splitters in Eqs.~\eqref{eq:op_QRL} and~\eqref{eq:U_QRL} are ordered the same way is a feature of our right-to-left circuit notation. When using left-to-right circuits, the ordering in the circuit description is opposite that in the equation. See Ref.~\cite{walshe2021streamlined} for further motivation for choosing this convention.}
\begin{equation}\label{eq:U_QRL}
    \begin{split}
        \raisebox{-0.94cm}{$\op B_{\text{network}} = $ \;}
\Qcircuit @C=0.5cm @R=0.60cm {
		& \qw & \bsbal{2} & \bsbal{1} & \qw  \\
        & \bsbal{2} & \qw & \qw & \qw  \\
        & \qw & \qw  & \bsbal{1} & \qw \\
        & \qw & \qw & \qw & \qw 
}
        \raisebox{-0.94cm}{\; .}
\end{split}
\end{equation}
Note that we use right-to-left circuits in this work, which benefits from a direct mapping to operator notation---compare the circuit to Eq.~\eqref{eq:op_QRL}.
This beam-splitter network produces a unitary transformation of the mode operators $\op{a}_j$ described by an orthogonal matrix $\BSmat_\text{network}$,
    \begin{align} \label{bsnetworkmatrices}
        \bsop_\text{network} \to
        \BSmat_\text{network} = \BSmat_\text{24} \BSmat_\text{13} \BSmat_\text{34} \BSmat_\text{12},
    \end{align}
where each beam-splitter matrix $\BSmat_{jk}$ from Eq.~\eqref{bsmat} has been expanded into the appropriate $4 \times 4$ matrix over all four modes. Note that non-identical beam splitters that share a mode do not commute, e.g.,~$[\bsop_{jk}, \bsop_{k\ell}] \neq 0$, which is reflected in the fact that their expanded rotation matrices also do not commute, $[\BSmat_{jk}, \BSmat_{k\ell}] \neq 0$.

\subsection{CV SWAP gates and parity operators}

Modes can be exchanged using a CV SWAP gate~\cite{Volkoff2022}
\begin{equation}
\begin{split}
	\raisebox{-1em}{$\swap_{jk} = \swap_{kj}\, = \, $~~}
         \Qcircuit @C=1.15em @R=2em @! 
         {
         	& \qswap      & \rstick{j}  \qw \\
         	& \qswap \qwx & \rstick{k} \qw
  		  } 
\end{split}
\qquad .
\end{equation}   
These satisfy several useful gate identities; see Eq.~\eqref{swapcircuitidentity}.
The action of SWAP gates on a beam-splitter network can also be represented by matrix multiplication of network matrices $\BSmat_\text{network}$,
    \begin{align} \label{SWAPmatrix}
        \swap_{jk} \bsop_\text{network} \rightarrow \mat{P}_{jk} \BSmat_\text{network},
    \end{align}
where $\mat{P}_{jk}$ is a permutation matrix between indices $j$ and $k$. Left multiplication by $\mat{P}_{jk}$ swaps rows $j$ and $k$. Similarly, the right action of a SWAP gate is represented as matrix multiplication from the right by $\mat{P}_{jk}$, which swaps columns $j$ and $k$.

The single-mode parity operator is defined as
    \begin{align} \label{eq:parity}
        \op{F}^2 \coloneqq e^{i \pi \op a^\dag \op a} 
        ,
    \end{align}
expressed here as two applications of the Fourier transform $\op{F}$. Parity transformations flip the sign of a quadrature operator, i.e.,~$\op{F}^2 \op{q} \op{F}^2 = -\op{q}$.
The left action of a parity operator on a beam splitter gives
\begin{equation} \label{paritymatrix}
    \op{F}_j^2 \op{B}_\text{network} \rightarrow \mat{M}_j\BSmat_\text{network}
    ,
\end{equation}
where $\mat{M}_j$ is a diagonal matrix with elements equal to 1 except entry $(j,j)$, which is $-1$. This flips the sign of the $j$th row in $\BSmat_\text{network}$. Similarly, the right action of a parity operator gives matrix multiplication from the right by $\mat{M}_{j}$, which flips the sign of the $j$th column.

Both SWAP gates and parity operators can reverse the direction of a beam splitter; see Eq.~\eqref{eq:bsdirectionswap}.

\subsection{Quadrature measurements} \label{sec:quadmeasurements}

In CV cluster-state quantum computing, modes are usually measured via homodyne detection. Adjusting the phase of the local oscillator allows measurements of any quadrature,
    \begin{align} \label{rotatedquadrature}
        \op{p}_\theta \coloneqq \op{R}^\dagger (\theta)\op{p} \op{R} (\theta) = \op{p} \cos \theta + \op{q} \sin \theta ,
    \end{align}
where
    \begin{align}
        \op{R}(\theta) \coloneqq e^{i \theta \op{a}^\dagger \op{a}}
    \end{align}
is the phase-delay operator (also called the rotation operator). 
The eigenstates $\ketsub{t}{\theta} = \op{R}^\dagger (\theta) \ket{t}_p$ satisfy $\op{p}_\theta \ketsub{t}{\theta} = t \ketsub{t}{\theta}$. 
In a right-to-left circuit, measurement of a quadrature $\op{p}_\theta$ with outcome $m$ is indicated via a projection onto the bra $\brasub{m}{p_\theta} = \brasub{m}{p} \op{R}(\theta)$, or in circuit form
\begin{equation} \label{rotatedmeasurement}
\begin{split}
         \Qcircuit @C=1.25em @R=2.5em @! 
         {
         	\lstick{\brasub{m}{p_{\theta}}} &  \qw 
  		  } 
  	\quad = \quad\quad\quad 
         \Qcircuit @C=0.1em @R=0.9em @! 
         {
         	\lstick{\brasub{m}{p}}  & \gate{R(\theta)}       & \qw     
  		  } 	  
\end{split}
\; .
\end{equation}
Momentum measurements are given by $\theta = 0$, and position measurements by $\theta = \frac{\pi}{2}$. Also, a parity operator $\op{F}^2 = \op{R}(\pi)$ in Eq.~\eqref{eq:parity} serves only to change the sign of the measurement outcome:
\begin{align}
    \brasub{m}{p_{\theta}} \op{F}^2
=
    \brasub{m}{p_{\theta+\pi}}
=
    \brasub{m}{-p_{\theta}}
=
    \brasub{-m}{p_{\theta}}
    ,
\end{align}
where the subscript ``$-p_\theta$'' represents a measurement of the operator~$-\op p_\theta$.

\subsection{Adding and removing beam splitters virtually by 
restricting homodyne angles%
} \label{subsection:virtualbeam splitters}

The CV cluster states in Sec.~\ref{sec:clusterStates} have at their heart beam-splitter networks over four modes. However, only the QRL uses four beam splitters; the others use only three. The equivalences to the QRL involve inserting the missing fourth beam splitter into other cluster state's beam-splitter network.

We show here that in some cases, the missing beam splitter can be inserted virtually simply by correlating homodyne measurement angles. Moreover, by the same process, an existing beam splitter can be effectively deleted.
The action of a beam splitter on two eigenstates of the same quadrature is
\begin{align}\label{eq:virtualBSproof}
    \op B_{21} \ketsub{t}{p_\theta}
    \otimes
    \ketsub{s}{p_\theta} 
    &=
    \ketsub{\tfrac{1}{\sqrt{2}}(s+t)}{p_\theta}
    \otimes 
    \ketsub{\tfrac{1}{\sqrt{2}}(s-t)}{p_\theta} .
\end{align}
Taking the Hermitian conjugate of this expression (%
recall $\op{B}_{21}^\dag = \op{B}_{12}$) gives
\begin{align}\label{eq:virtualBSproof}
    \brasub{t}{p_\theta} \otimes \brasub{s}{p_\theta} \op B_{12}  
    &= \brasub{\tfrac{1}{\sqrt{2}}(s+t)}{p_\theta} \otimes \brasub{\tfrac{1}{\sqrt{2}}(s-t)}{p_\theta}.
\end{align}
 We can interpret this expression in terms of quadrature measurements, as described in Sec.~\ref{sec:quadmeasurements}, on both modes.
This allows us to account for the effects of a beam splitter before measurements of the same quadrature by classically using the proper linear combinations of the measurement outcomes. Doing so effectively deletes the beam splitter from the circuit. Similarly, given measurements of two modes in the same quadrature with no beam splitter preceding them, a ``virtual beam splitter'' can be added to a circuit if the outcomes are processed similarly. Note that inserting or deleting beam splitters \emph{does not} depend on the measurement outcomes, only on the fact that the measured quadratures are the same.

The circuit to remove a beam splitter before identical quadrature measurements is
\begin{align}\label{eq:removeBS}
    \Qcircuit @C=0.6em @R=0.6em
    {
	 &\nogate{\brasub{m_1}{p_\theta}} & \qw & \bsbal{1} &  \qw\\
     &\nogate{\brasub{m_2}{p_\theta}} & \qw & \qw & \qw
	 }
     \raisebox{-1em}{\quad = }
     \Qcircuit @C=0.6em @R=0.6em
     {
      &\nogate{\brasub{\frac{1}{\sqrt{2}}(m_1 + m_2)}{p_\theta}} & \qw & \qw &  \qw\\
      &\nogate{\brasub{\frac{1}{\sqrt{2}}(m_2-m_1)}{p_\theta}} & \qw & \qw & \qw
      }
      \raisebox{-1em}{\; , }
\end{align}
and to add a virtual beam splitter is
\begin{align}\label{eq:virtualBS}
    \Qcircuit @C=0.6em @R=0.6em
    {
	 &\nogate{\brasub{m_1}{p_\theta}} & \qw & \qw &  \qw\\
     &\nogate{\brasub{m_2}{p_\theta}} & \qw & \qw & \qw
	 }
     \raisebox{-1em}{\quad = }
     \Qcircuit @C=0.6em @R=0.6em
     {
      &\nogate{\brasub{\frac{1}{\sqrt{2}}(m_1 - m_2)}{p_\theta}} & \qw & \bsbal{1}
      &  \qw\\
      &\nogate{\brasub{\frac{1}{\sqrt{2}}(m_2+m_1)}{p_\theta}} & \qw & \qw & \qw
     }
     \raisebox{-1em}{\; .}
\end{align}
In both cases above, the physical operation is on the left of the equal sign, and the circuit to which it is equivalent is on the right.

Note that the insertion and deletion of a beam splitter is accompanied by an orthogonal transformation on the measurement outcomes. Collecting the outcomes into a vector $\vec{m} = (m_1, m_2)^\tp$, removing a beam splitter is accompanied by the transformation $\vec{m} \mapsto \mat{R}^\tp \vec{m}$ and inserting one, by $\vec{m} \mapsto \mat{R} \vec{m}$.

\section{Equivalence of beam-splitter constructed cluster states}\label{sec:clusterStates}

The quad-rail lattice (QRL) cluster state, introduced in Ref.~\cite{Menicucci2011tempmodeCVCS}, is a resource for universal, fault-tolerant quantum computing. 
In this section, we briefly describe the construction of 
several types of scalable CV cluster states 
and show that they are equivalent to the QRL given
certain restrictions:
the bilayer square lattice (BSL)~\cite{Alexander2018BSL}, the double bilayer square lattice (DBSL)~\cite{Larsen2019detCVCS}, a 3D extension to the DBSL~\cite{Larsen2021fault}, and the modified bilayer square lattice (MBSL)~\cite{Asavanant2019detCVCS}. The work by 
Larsen \emph{et al.}~\cite{Larsen2020architecture} provides further details for each of the cluster states discussed here. The equivalences we show below are the link that let us apply the detailed analysis of the QRL in previous work~\cite{walshe2021streamlined} to the other cluster states under consideration.

The construction the cluster states considered in this work begins with the preparation of two-mode squeezed states~(TMSSs), generated by sending 
pairs of squeezed vacuum states
through beam splitters.
Differences arise from different ways of stitching the TMSSs together with additional beam splitters. 
Stitching involves four-mode entangling operations, in which one half of four TMSSs are coupled together with beam splitters. For each cluster state, equivalence to the QRL involves inserting an additional beam splitter in this process; doing so effectively gives a new, distinct cluster state with more connectivity between modes---we call these the \emph{completed} version of the original cluster state; \emph{i.e.} the completed BSL (cBSL). When used for CV quantum computing, all of the modes in the cluster state are measured via homodyne detection. We use this fact and show that the needed beam splitter can be inserted \emph{virtually} by restricting specific measurement angles---which allows the unmodified cluster states to be used as if they were a different cluster state with more connectivity. Under this restriction, we refer to a cluster state as  \emph{virtually completed}, \emph{i.e.} the virtually completed BSL (vcBSL). The exception to the cluster states considered here is the MSBL, for which the missing beam splitter cannot be inserted virtually.

\subsection{Representations of CV cluster states}

Large-scale CV cluster states have a complicated entanglement structure spanning many modes, so representing them is a challenge in itself. The graphical calculus for Gaussian pure states~\cite{Menicucci2011} offers a visual representation of these states as graphs that faithfully encode their precise wave functions. A simplified version of this graphical calculus is developed in Ref.~\cite{Menicucci2011tempmodeCVCS} that is applicable for a large class of scalable CV cluster states, which includes all of the ones examined in this work. It streamlines the original formalism while maintaining the faithful representation of the full state. It also introduces the notation of an arrow between nodes for a beam splitter between the corresponding modes, as well as the simultaneous use of both graphs and arrows to illustrate the steps in the construction of a cluster state.

\begin{figure}[b!]
    \includegraphics[width=0.85\columnwidth]{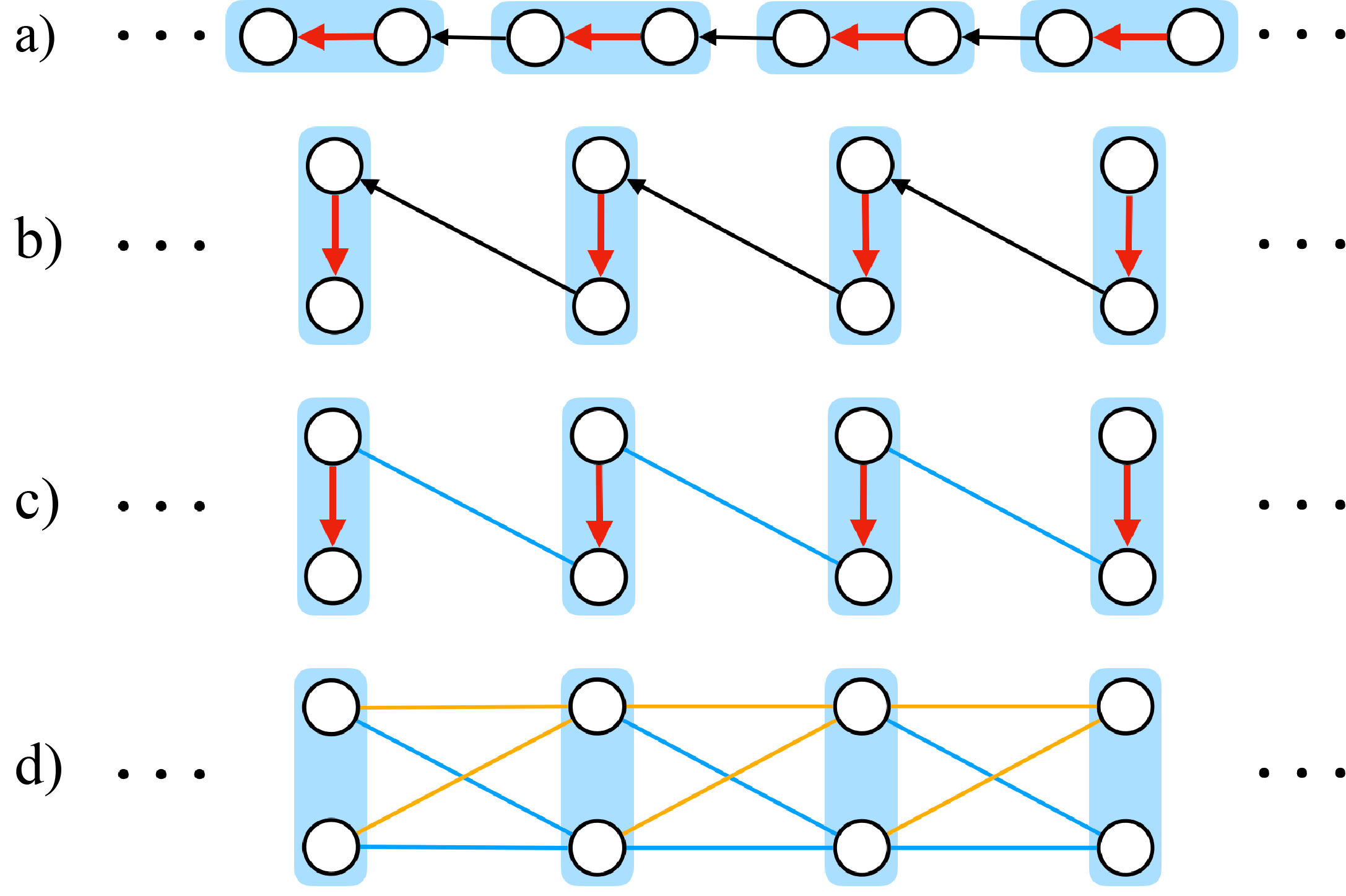}
    \caption{
    Four representations of the same macronode wire. Black and red beam splitters, represented by arrows, couple modes prepared in squeezed states. The black beam splitters are applied first to generate TMSSs, then the red beam splitters couple one half of neighboring TMSSs into a macronode wire. a) Representation showing the modes coupled by beam splitters. b) Folding the macronode wire to emphasize two-mode macronode groupings, indicated in blue. c) Black beam splitters, acting first, generate TMSSs between modes, each represented as two modes with a blue edge between them. Pairs of TMSSs are coupled by red beam splitters. 
    d) The action of both black and red beam splitters in graphical form, revealing the entanglement structure of the macronode wire~\cite{Menicucci2011tempmodeCVCS}. 
    }\label{eq:macronodeWireToBS}
\end{figure}

The \emph{macronode wire}~\cite{Alexander2014noise} is the primitive unit we start with when representing each of these cluster states. It is the basic highway down which quantum information is channeled, and our work within this paper builds upon the details of teleportation that we established in Ref.~\cite{Walshe2020} and extended to two modes for the QRL in Ref.~\cite{walshe2021streamlined}. Figure~\ref{eq:macronodeWireToBS} shows several representations of the macronode wire: a) and b) show a series of modes coupled by beam splitters (black beam splitters apply before the red ones), c) shows a hybrid representation where two-mode squeezed states are coupled by beam splitters, and d) a simplified graph.

For the mathematical details of this representation, the reader may consult Ref.~\cite{Menicucci2011tempmodeCVCS}, but many readers may be satisfied thinking of the graphs (in blue and yellow) as representing the entanglement structure of a multimode Guassian state, upon which beam splitters act in the locations and order prescribed. If the node being acted upon is isolated (not a graph, just a node), one may think of this node as representing a squeezed state. Thus, the macronode wire above is a collection of squeezed states, upon which the black beam splitters act to create a collection of TMSSs, followed by another set of beam splitters (the red ones) that link these TMSSs into a macronode wire whose full graph is shown in d).

Figure~\ref{fig:allCLusterStates} employs the three representations in b), c), and d) of Fig.~\ref{eq:macronodeWireToBS} to describe the cluster states of interest in this work. 
The figure caption provides additional information about the different representations and their uses. 
In each type of cluster state, adjacent macronode wires are coupled together by dashed black beam splitters, enabling two-mode gates.
It's important to recognise, however, that there are many ways to define the macronode wire. Our representation biases a single one. This is made clear in Ref.~\cite{Alexander2016flexible} because a macronode wire is the path along which information is routed, and there are many ways to route that information.

With this perspective, we introduce a unifying description that treats the core of each cluster state as a four-mode beam-splitter network~$\op{B}_\text{network}$, with each one described in its own subsection below.

\begin{figure*}[hp!]
    \vspace{3em}
    \includegraphics[width=0.95\textwidth]{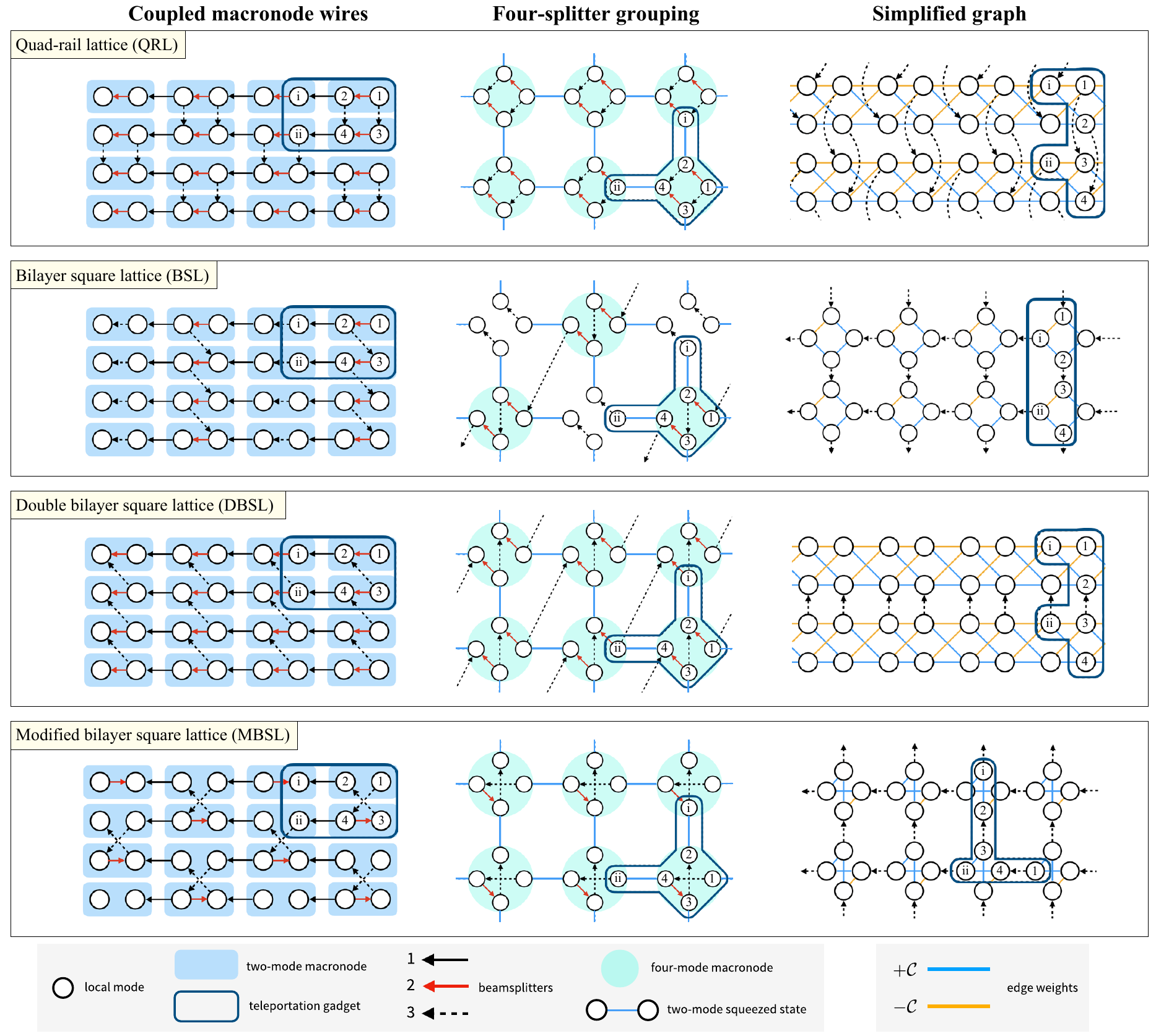}
    \caption{Representations of the CV cluster states of interest in this work. 
    Column 1 represents each cluster state as coupled macronode wires---a representation introduced in Ref.~\cite{walshe2021streamlined} for the QRL. Dashed black beam splitters provide the inter-wire coupling. The box highlights a 6-mode teleportation gadget that provides the computing power for each cluster state. The QRL, BSL, and DBSL all comprise macronode wires that are coupled together in different ways. The macronode wires in the MBSL, however, are periodically missing beam splitters (missing horizontal red arrows compared to the others). 
    Column 2 highlights the two-mode squeezed states that are generated initially (two nodes connected by a blue line) and how these states are coupled together at regular intervals by a four-splitter that is either complete (QRL) or incomplete (others)---see Sec.~\ref{sec:incompletefour-splitters}. This representation shows groupings of four modes (different to the grouping in the first column), where each mode is half of an entangled pair (a TMSS). The QRL representation here differs slightly from that presented in Ref.~\cite{Larsen2020architecture}. We do not address this specifically here because the four-splitter structure (the internal structure of the macronodes) is the same in that work as it is here, so the argument we make here about the gate noise holds regardless.
    Column 3 shows a representation that employs the simplified graphical calculus~\cite{Menicucci2011tempmodeCVCS} along with additional beam splitters (see main text). Similar representations for some of these states can be found in other works~\cite{Menicucci2011tempmodeCVCS, Alexander2016flexible, Larsen2020architecture, Alexander2018BSL}. This column highlights the different types of construction: For the QRL and DBSL, macronode wires are constructed first and then connected to one another, and for the BSL and MBSL, non-local modes are coupled in groups of four, and then these groupings are connected. (Comparison with the other columns shows that the BSL can also be represented as coupled macronode wires; we choose here to follow the description in Ref~\cite{Alexander2018BSL}.) Some works show the full graph-state structure~\cite{Alexander2018BSL, Larsen2020architecture, Alexander2016flexible} by implementing the graph rules of the final beam splitter~\cite{Menicucci2011tempmodeCVCS,Menicucci2011}, but we omit this description as it is graphically unwieldy and does not aid our attempts to compare them.
    }
    \vspace{3em}
    \label{fig:allCLusterStates}
\end{figure*}

\subsection{The teleportation gadget}
As the modes in the cluster states are measured via homodyne detection, states are teleported down the macronode wires that comprise them and across couplings. Each of the cluster-state constructions we consider has at its core a teleportation gadget:%
\footnote{Be mindful that this circuit uses the right-to-left convention discussed in Sec.~\ref{sec:bsNetworks}, with input states on the right and homodyne measurements and output states on the left.}
\begin{equation}\label{eq:2modegadget}
    \begin{split}
\Qcircuit @C=1em @R=1em {
		&&& \lstick{\brasub{m_1}{p_{\theta_1}}}  &  \qw & \multigate{4}{B_\text{network}}  & \qw & \qw &\qw & \nogate{\text{(in)}} & \rstick{1} \\
		&&& \lstick{\brasub{m_2}{p_{\theta_2}}}  & \qw  &\ghost{B_\text{network}}  &  \qw & \bsbal{1} & \qw & \nogate{\ket{\psi}}& \rstick{2} \\
	 & &  & &   &  & \nogate{\text{(out)}} & \qw & \qw& \nogate{\ket{\phi}}& \rstick{\text{i}} \\
		&&& \lstick{\brasub{m_3}{p_{\theta_3}}} & \qw & \ghost{B_\text{network}}  & \qw & \qw & \qw&\nogate{\text{(in)}}& \rstick{3} \\
		&&& \lstick{\brasub{m_4}{p_{\theta_4}}}  &\qw& \ghost{B_\text{network}}  & \qw & \bsbal{1}&\qw & \nogate{\ket{\psi}}& \rstick{4} \\
		 &&&&&&\nogate{\text{(out)}}  & \qw&\qw & \nogate{\ket{\phi}}& \rstick{\text{ii}}
}
\end{split}
\quad \, ,
\end{equation}
where the four-mode beam-splitter network is specific to the cluster state at hand. The Arabic-numbered modes are measured, while the Roman-numbered modes remain afterward.

For each cluster state in Fig.~\ref{fig:allCLusterStates}, this six-mode gadget is indicated inside a dashed box, with matching mode labels between Fig.~\ref{fig:allCLusterStates} and Circuit~\eqref{eq:2modegadget}.
For all cluster states%
, the two input modes are teleported to the two output modes with a two-mode Gaussian unitary 
    \begin{align} \label{eq:twomodeGaussiangate}
        \op{V}_\text{network}^{(2)}(\vec \theta)
    \end{align}
applied that depends on which quadratures are measured [Eq.~\eqref{rotatedquadrature}], as specified by the measurement angles $\vec \theta=\{ \theta_1, \theta_2, \theta_3, \theta_4 \}$. (The displacements on the two modes that arise from the measurement outcomes  will be ignored, which is possible because the gate to be implemented is Gaussian.) 
In addition, this circuit applies a potentially non-unitary operation depending on the states stitched together by the beam-splitter network~\cite{Walshe2020}, as indicated by $\ket{\psi}$ and $\ket{\phi}$ in the circuit above. 
In standard cluster-state constructions, these states are position- and momentum-squeezed states that become TMSSs after a beam splitter. In Section~\ref{subsec:GKPBellpairs}, we replace them with GKP qunaught states in order to perform GKP error correction.
For the moment, though, we leave these states unspecified in order to focus on the beam-splitter network itself and the Gaussian gates implemented by the measurements.

\subsection{Quad-rail lattice (QRL)}\label{sec:QRL}

The quad-rail lattice (QRL) cluster state~\cite{Menicucci2011tempmodeCVCS} has been studied 
as a platform for quantum computing~\cite{Menicucci2011tempmodeCVCS,Alexander2016flexible,Larsen2020architecture,walshe2021streamlined}. We compile and summarize important details relevant to comparison to the other cluster states considered in this work.

The QRL, Fig.~\ref{fig:allCLusterStates}(a), begins with a collection of TMSSs. Groups of four modes (each belonging to a different TMSS) are then combined using additional beam splitters (red and dashed black in the figure). A useful perspective considers the QRL as coupled four-mode beam-splitter networks, each of the form
 \begin{equation}\label{eq:U_QRL_styled}
    \begin{split}
        \raisebox{-3em}{$\op B_{\text{QRL}} = $ \;}
\Qcircuit @C=1.5em @R=0.8em {
		& \qw & \bsbal{2}[.>] & \bsbal[red]{1} & \qw  & \nogate{1} \\
        & \bsbal{2}[.>] & \qw & \qw & \qw   & \nogate{2}\\
        & \qw & \qw  & \bsbal[red]{1} & \qw  & \nogate{3}\\
        & \qw & \qw & \qw & \qw  & \nogate{4}
}
        \raisebox{-3em}{\; ,}
\end{split}
\end{equation}
contained within the first four modes of the two-mode macronode gadget in column 1 of Fig.~\ref{fig:allCLusterStates}(a), with corresponding orthogonal matrix 
\begin{align}
    \mat \BSmat_{\text{QRL}} = \BSmat_\text{24} \BSmat_\text{13} \BSmat_\text{34} \BSmat_\text{12}
    &=\frac{1}{2} \begin{bmatrix}
        1 & -1 & -1 & 1 \\
        1 & 1 & -1 & -1 \\
        1 & -1 & 1 & -1 \\
        1 & 1 & 1 & 1 
    \end{bmatrix}.
\end{align}
This matrix shows that the QRL beam-splitter network produces equal superpositions of quadrature operators on all output modes, with only the phases differing. This is an example of an object we refer to as a \emph{balanced four-splitter}, and it drives the equivalence between each of these cluster states. We discuss more about four-splitters in Sec.~\ref{sec:four-splitters}.

Aside from the internal structure of the QRL's two-mode macronode gadget, it is worth appreciating the brickwork structure of the QRL cluster state when viewed in terms of these macronode wires. The coupling with neighboring wires happens at alternating four-mode sites in the cluster state. This feature is useful because there are no physical operations interfering with the two-mode gate at each site. In contrast, other cluster states in column 1 of Fig.~\ref{fig:allCLusterStates} have columns of beam splitters coupling the macronode wires together, and these have to be accounted for in the measurement outcomes. This fact does not, however, influence the noise equivalence that we are presenting in this paper.

An important insight is that the description of the QRL in terms of local modes and beam-splitter networks does not rely on the input states being squeezed momentum states at all. We previously found%
~\cite{Walshe2020, walshe2021streamlined} 
that preparing some or all of the modes in GKP qunaught states can be used to implement GKP error correction, a fact that we will return to later.

Measuring the four modes coupled by the beam-splitter network in the context of the teleportation gadget, Circuit~\eqref{eq:2modegadget}, teleports the inputs (modes 1 and 3) in Fig.~\ref{fig:allCLusterStates} to the outputs (modes i and ii) and applies a two-mode Gaussian unitary gate $\op{V}_\QRL^{(2)}(\vec \theta)$, Eq.~\eqref{eq:twomodeGaussiangate} to them. This gate, calculated in Ref.~\cite{walshe2021streamlined}, can be decomposed into two single-mode Gaussian unitaries, $\op V(\theta,\theta')$, surrounded by beam splitters:
\begin{equation} \label{VtwoMode_QRL}
    \op{V}_\QRL^{(2)}(\vec \theta) = \bsop_{21} [\op{V}_1(\theta_1, \theta_2) \otimes \op{V}_2(\theta_3, \theta_4)] \bsop_{12}
    \, ,
\end{equation}
or in circuit form
\begin{equation}\label{VtwoMode_QRL_circuit}
    \begin{split}
\Qcircuit @C=1.1em @R=1.25em {
		& \qw&\gate{V(\theta_1,\theta_2)}  &  \bsbal{1} &\qw& \\
		&\bsbal{-1}& \gate{V(\theta_3,\theta_4)} & \qw&\qw&\\
}
\end{split}
\; .
\end{equation}
Explicit forms for $\op V(\theta,\theta')$ can be found in various places~\cite{walshe2021streamlined}. We present two new forms, derived in Appendix~\ref{shearVderivation},
    \begin{align} \label{eq:VgateShear}
        \op V(\theta,\theta') 
        &= \op R(\theta_1-\tfrac{\pi}{2}) \op P[2\cot (2 \theta_-)] \op R(\theta_1-\tfrac{\pi}{2}) 
        \\
        &= \op R(\theta -\pi) \op P_{p}[2\cot (\theta - \theta')] \op R(\theta)
    \, ,
    \end{align}
where the position- and momentum-shear gates are
  \begin{align} 
        \label{eq:posshear}
        \op P(\sigma)      &\coloneqq e^{i \frac{\sigma}{2} \op{q}^2}  
        \quad \text{and} \quad
        \op P_{p}(\sigma) \coloneqq e^{i \frac{\sigma}{2} \op{p}^2}.
    \end{align}

\subsection{Bi-layer square lattice (BSL)}\label{sec:BSL}

The bi-layer square lattice (BSL) cluster state, introduced in Ref.~\cite{Alexander2016flexible} and studied elsewhere~\cite{Alexander2018BSL, Larsen2020architecture}, is depicted in Fig.~\ref{fig:allCLusterStates}(b). 
The BSL begins with a collection of TMSSs. 
Groups of two modes, one from each TMSS, are chained together with beam splitters (red). These chains are periodically connected using a final set of beam splitters (dashed black). 
In some works~\cite{Larsen2020architecture, Alexander2016flexible, Alexander2018BSL}, this procedure is represented in terms of the four-mode groupings created by the first two sets of beam splitters (black and red)~\cite{Menicucci2011tempmodeCVCS},
\begin{align} \label{BSLfourmodes}
    \includegraphics[width=0.4\columnwidth]{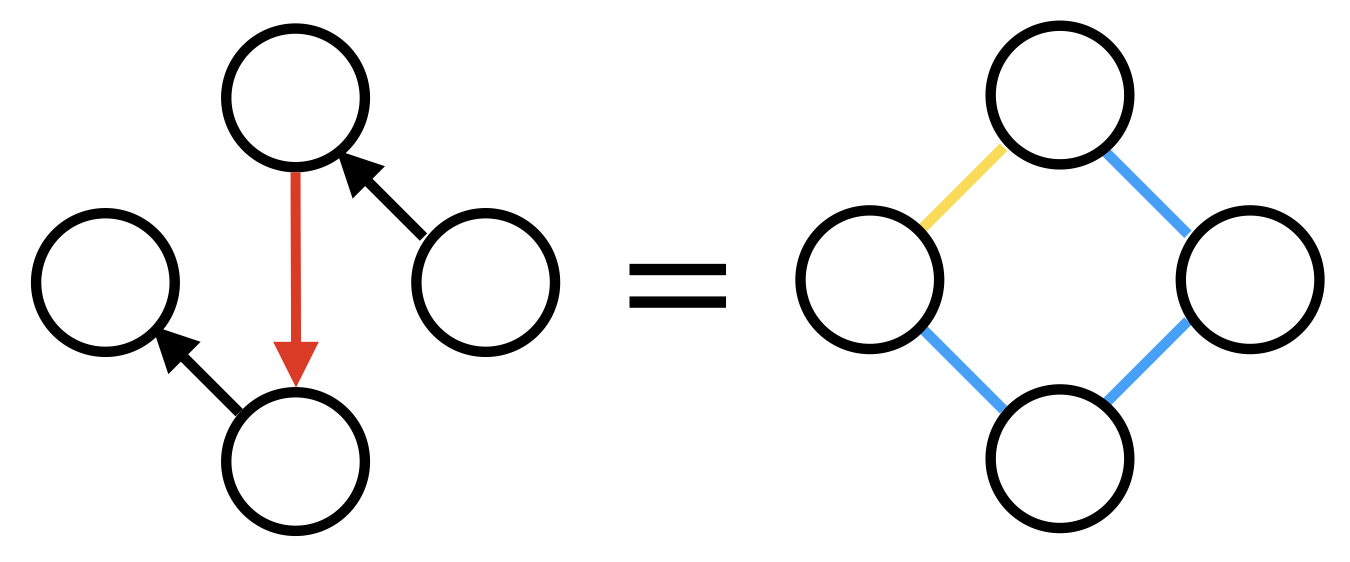}
    \raisebox{0.8cm}{\quad .}
\end{align}
The final set of beam splitters (dashed black) stitches together each group of four modes,
\begin{align} \label{BSLfourmodes2}
    \includegraphics[width=0.4\columnwidth]{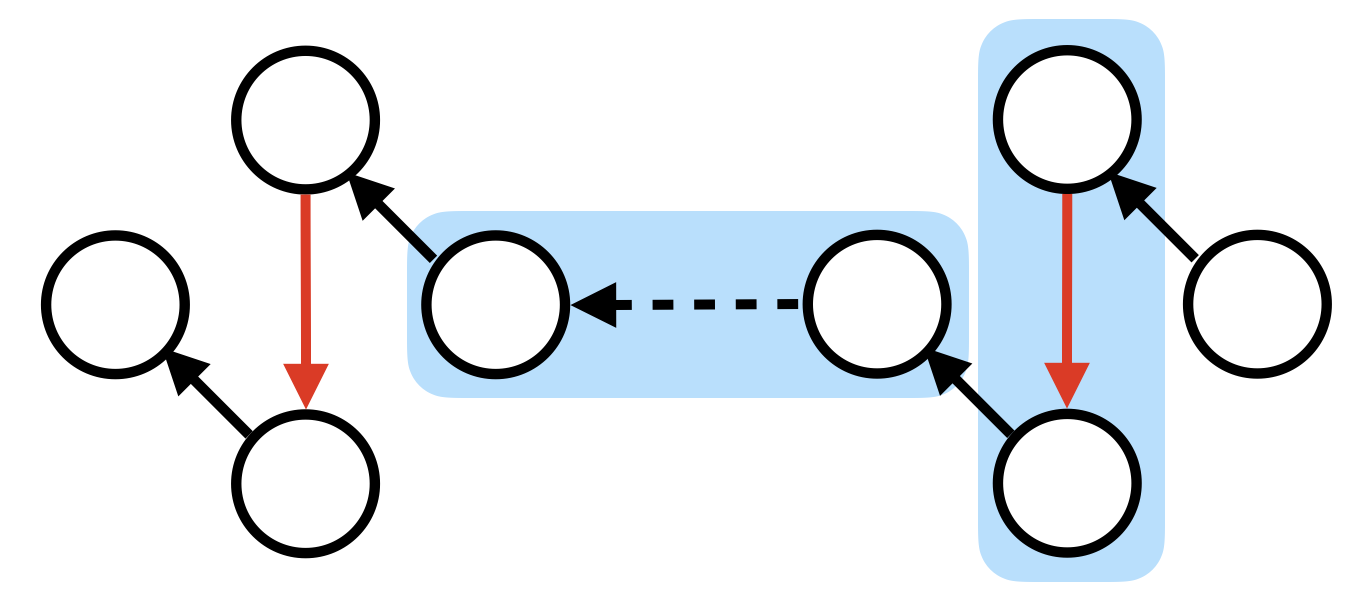}
    \raisebox{0.8cm}{\quad ,}
\end{align}
where we have included blue shading to identify the connection to the macronode wire, Fig.~\ref{eq:macronodeWireToBS}.
With the macronode wire strecthed out straight, we can identify the beam splitters that couple two macronode wires together, see column 1 of Fig.~\ref{fig:allCLusterStates}. In circuit form they are 
\begin{equation}\label{eq:U_BSL}
    \begin{split}
        \raisebox{-3em}{$\op B_{\text{BSL}} = $ \;}
\Qcircuit @C=1.5em @R=0.8em {
	& \qw & \qw & \bsbal[red]{1} & \qw & \nogate{1} \\
        & \qw & \bsbal{1}[.>] & \qw & \qw & \nogate{2} \\
        & \qw & \qw & \bsbal[red]{1} & \qw & \nogate{3} \\
        & \qw & \qw & \qw & \qw & \nogate{4}
}
\raisebox{-3em}{\; ,}
\end{split}
\end{equation}
corresponding to the orthogonal matrix
\begin{align} \label{eq:BSLmat}
    \BSmat_{\text{BSL}} 
    = \BSmat_\text{23} \BSmat_\text{34} \BSmat_\text{12}
    &=\frac{1}{2} \begin{bmatrix}
        \sqrt{2} & -\sqrt{2} & 0 & 0 \\
        1 & 1 & -1 & 1 \\
        1 & 1 & 1 & -1 \\
        0 & 0 & \sqrt{2} & \sqrt{2} 
    \end{bmatrix}.
\end{align}
Note that these are \emph{not} the same four modes as in Eqs.~\eqref{BSLfourmodes} and~\eqref{BSLfourmodes2}, which is evident from the beam-splitter styles. Rather, this four-mode grouping is best seen in Fig.~\ref{fig:allCLusterStates}(b).

\subsubsection*{Completing the BSL and connecting to the QRL}
The BSL's beam-splitter network, Eq.~\eqref{eq:U_BSL}, is \emph{incomplete} because it contains three beam splitters, not four, as described in Sec.~\ref{subsection:virtualbeam splitters}. The incompleteness is reflected in the fact that the BSL's matrix, Eq.~\eqref{eq:BSLmat}, contains elements other than $\pm \frac{1}{2}$, indicating that it does not create equal superpositions of position and momentum operators across all modes. 

The BSL can be completed by inserting a beam splitter $\bsop_{14}$ between modes 1 and 4 to create new CV cluster state that we call the \emph{completed} BSL (cBSL)\footnote{The choice $\bsop_{41}$ also works. That choice would give a different orthogonal matrix, connection to the QRL, and two-mode gate.} 
whose beam-splitter network, $\op B_{\text{cBSL}} \coloneqq \op{B}_{14} \op B_{\text{BSL}}$, is given in circuit form by 
\begin{equation}\label{eq:U_cBSL}
    \begin{split}
        \raisebox{-3em}{$\op B_{\text{cBSL}} = $ \;}
\Qcircuit @C=1.5em @R=0.8em {
	& \bsbal{3}[.>] & \qw & \bsbal[red]{1} & \qw & \nogate{1} \\
        & \qw & \bsbal{1}[.>] & \qw & \qw & \nogate{2} \\
        & \qw & \qw & \bsbal[red]{1} & \qw & \nogate{3} \\
        & \qw & \qw & \qw & \qw & \nogate{4}
}
        \raisebox{-3em}{\; ,}
\end{split}
\end{equation}
with associated orthogonal matrix, 
\begin{align}
    \BSmat_{\text{cBSL}} 
    = \BSmat_\text{14} \BSmat_\text{23} \BSmat_\text{34} \BSmat_\text{12}
    &=\frac{1}{2} \begin{bmatrix}
        1 & -1 & -1 & -1 \\
        1 & 1 & -1 & 1 \\
        1 & 1 & 1 & -1 \\
        1 & -1 & 1 & 1 
    \end{bmatrix}.
\end{align}

\begin{figure}[t!]
    \includegraphics[width=0.9\columnwidth]{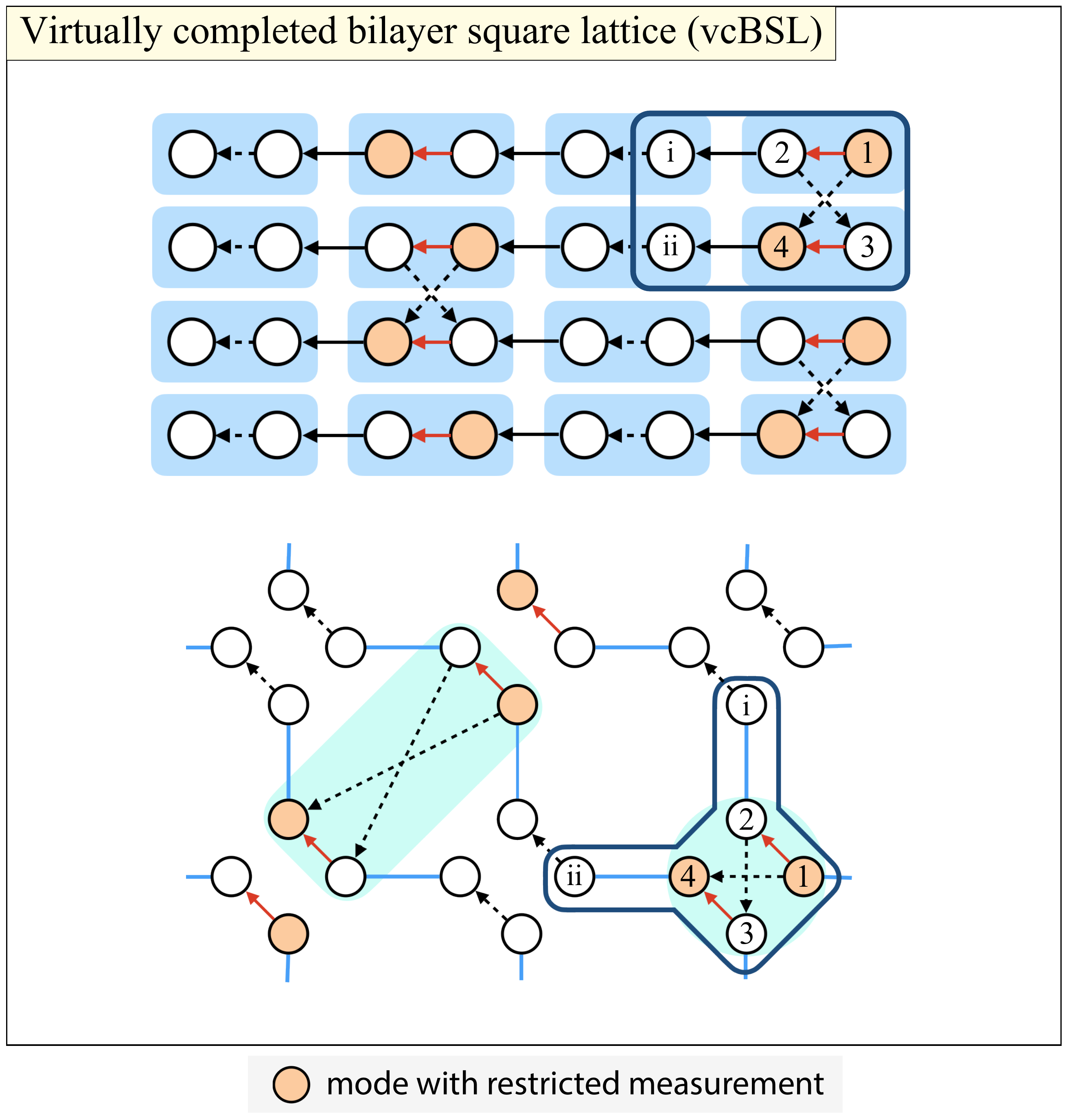}
    \caption{The BSL with chosen restricted measurement angles, allowing it to behave as the cBSL. Local modes in orange are measured in the same basis (same measurement angle), allowing insertion of virtual beam splitters, or virtual deletion of existing ones (see Sec.~\ref{subsection:virtualbeam splitters}). The virtual beam splitters are those between orange modes. Compare with the unrestricted BSL in Fig.~\ref{fig:allCLusterStates}. The choice here provides isolated two-mode teleportation gadgets in a sparse brickwork structure. 
    }\label{fig:cBSL}
\end{figure}

The cBSL's beam-splitter network can be written in terms of the QRL's under a single SWAP on the left and a parity operator, Eq.~\eqref{eq:parity}, on the right,
\begin{equation} \label{eq:U_cBSLTransformed}
    \op B_{\text{cBSL}} 
    = \text{SWAP}_{34} \op B_{\text{QRL}} \op F_4^2 .
\end{equation}
 which has a circuit  
 \begin{equation}\label{eq:U_cBSLTransformed_circuit}
     \begin{split}
         \raisebox{-3em}{$\op B_\text{cBSL} = $ \;}
 \Qcircuit @C=1.5em @R=0.8em {
 	& \qw & \qw & \bsbal{2}[.>] & \bsbal[red]{1} & \qw & \qw & \nogate{1}  \\
         & \qw & \bsbal{2}[.>] & \qw & \qw & \qw & \qw  & \nogate{2} \\
         & \qswap & \qw & \qw & \bsbal[red]{1} & \qw & \qw & \nogate{3} \\
         & \qswap \qwx & \qw & \qw & \qw & \gate{F^2} & \qw & \nogate{4}
 }
         \raisebox{-3em}{\; .}
 \end{split}
\end{equation}
This can be verified using the circuit identity in Eq.~\eqref{swapcircuitidentity}. 
Now consider the cBSL beam-splitter network in the two-mode teleportation gadget, Circuit~\eqref{eq:2modegadget}. 
Replacing $\op B_\text{network}$ with Eq.~\eqref{eq:U_cBSLTransformed}, pushing the SWAP into the measurements, and bouncing~\cite{Walshe2020} the $\op F^2$ to the second output mode, we get the two-mode gate, Eq.~\eqref{eq:twomodeGaussiangate}, for the cBSL,
\begin{align} \label{VtwoMode_cBSL}
    \op{V}_{\text{cBSL}}^{(2)}(\vec \theta)
    &= \op F_2^2 \bsop_{21} [\op{V}_1(\theta_1, \theta_2) \otimes \op{V}_2(\theta_4, \theta_3)] \bsop_{12} ,
\end{align}
where the vector $\vec{\theta}$ specifies the measurement angles in the initial gadget, Circuit~\eqref{eq:2modegadget}.
In circuit form, the two-mode gate is
\begin{equation}\label{twomodeGate_cBSL}
    \begin{split}
\Qcircuit @C=1.1em @R=1.25em {
		& \qw & \qw&\gate{V(\theta_1,\theta_2)}  &  \bsbal{1} &\qw& \\
		& \gate{F^2} &\bsbal{-1}& \gate{V(\theta_4,\theta_3)} & \qw&\qw& \\
}
\end{split}
\; .
\end{equation}
Up to the parity operator, this gate is nearly the same as that for the QRL, Circuit~\eqref{VtwoMode_QRL_circuit}, except that the measurement angles $\theta_3$ and $\theta_4$ have swapped positions.
The equivalence has revealed that the cBSL can implement any teleported gate that the QRL can by properly choosing the measurement angles.

\subsubsection*{Virtually completing the BSL}

In the teleportation gadget, Eq.~\eqref{eq:2modegadget}, a beam splitter network is followed by homodyne detection. 
This gives us the freedom to insert beam splitters into the measurement side of the circuit virtually by restricting measurement angles, see Sec.~\ref{subsection:virtualbeam splitters}.
Thus, we do not need to generate the cBSL physically---instead, 
we can complete the BSL virtually, allowing it to behave as if it were the cBSL. 
To do so, restrict the measurement
bases on modes 1 and 4 to be the same (indicated by the orange-shaded modes in Fig.~\ref{fig:cBSL}) to insert a virtual beam splitter between those modes,
\begin{equation}\label{eqform:U_cBSL}
    \brasub{m_1}{p_\theta,1} \otimes \! \brasub{m_4}{p_\theta,4} \op B_{\text{BSL}} 
    =  \! \! \! \brasub{\tfrac{m_1 - m_4}{\sqrt{2}} }{p_\theta,1} \otimes \! \! \brasub{\tfrac{m_1 + m_4}{\sqrt{2}}}{p_\theta,4} \op B_{\text{cBSL}}.
\end{equation}

In the teleportation gadget, this gives the circuit for the \emph{virtually completed BSL} (vcBSL),
\begin{equation}\label{eq:U_cBSL}
    \begin{split}
        \qquad
\Qcircuit @C=1.5em @R=0.8em {
	 \lstick{\brasub{\tfrac{m_1 - m_4}{\sqrt{2}}}{p_{\theta_1}}} & \bsbal{3}[-->] & \qw & \bsbal[red]{1} & \qw & \nogate{1} 
	 \\
    \lstick{\brasub{m_2}{p_{\theta_2}}}& \qw           & \bsbal{1}[.>] & \qw  & \qw & \nogate{2} 
    \\
    \lstick{\brasub{m_3}{p_{\theta_3}}}& \qw           & \qw & \bsbal[red]{1} & \qw & \nogate{3} 
    \\
    \lstick{\brasub{\tfrac{m_1 + m_4}{\sqrt{2}}}{p_{\theta_4 = \theta_1}}}&  \qw           & \qw & \qw & \qw & \nogate{4}
}
        \raisebox{-3em}{\; ,}
\end{split}
\end{equation}
with the virtual beam splitter indicated as the dashed arrow.

As long as this measurement angle restriction is respected, the vcBSL can be used in the teleportation gadget as if it were the cBSL. (Feedforward operations should also respect the linear transformations on the outcomes.) Thus, the vcBSL has the same two-mode gate as the cBSL under that restriction:
\begin{align} \label{VtwoMode_BSL}
    \op{V}_{\text{vcBSL}}^{(2)}(\vec \theta) & = \op{V}_{\text{cBSL}}^{(2)}(\vec \theta) |_{\theta_1=\theta_4} 
\end{align}
given in circuit form as
\begin{equation}\label{twomodeGate_BSL}
    \begin{split}
\Qcircuit @C=1.1em @R=1.25em {
		& \qw & \qw&\gate{V(\theta_1,\theta_2)}  &  \bsbal{1} &\qw& \\
		& \gate{F^2} &\bsbal{-1}& \gate{V(\theta_4 = \theta_1,\theta_3)} & \qw&\qw& \\
}
\end{split}
\; .
\end{equation}
Although any teleported gate that the QRL can implement is also available to the cMBSL via the relation in Eq.~\eqref{eq:U_cMBSLTransformed}, 
this is not the case for the vcBSL---the measurement restrictions required to insert the missing beam splitter constrain the available gates.

In the larger structure of the cluster state, freedom to restrict measurement angles in order to virtually add or remove beam splitters allows different connections between modes and between two-mode teleportation gadgets.  
In practice, restricted measurement angles can be made on-the-fly to tailor the connectivity to the problem at hand. 
Figure~\ref{fig:cBSL} illustrates a choice that produces a brickwork-like structure, similar to that of the QRL, although here the brickwork structure is more sparse---this choice allows for two-mode gates followed by single-mode gates.
\blk

\subsection{Double bi-layer square lattice (DBSL)}\label{sec:DBSL}
The DBSL, introduced in Ref.~\cite{Larsen2019detCVCS}, is a modification of the BSL with a construction that does not begin with the unit described in Eq.~\eqref{BSLfourmodes}. 
Rather, constructing the DBSL begins by producing a macronode wire, Eq.~\eqref{eq:macronodeWireToBS},
which is then wound up into a cylinder by introducing non-local beam-splitter connections~\cite{Larsen2020architecture}. 
Position measurements can then be used to cut out two-dimensional cluster-state sheets with the structure of the DBSL. (Further position measurements can reduce this DBSL to a BSL~\cite{Larsen2019detCVCS}).

Larsen \emph{et al.}~\cite{Larsen2019detCVCS} experimentally generated a DBSL cluster state and in a follow-up work~\cite{Larsen2020detgatesCVCS} used the state to implement a variety of measurement-based Gaussian operations on one and two modes. To do so, they used an alternate perspective on the DBSL cluster state that identifies macronode wires running perpendicular to those in 
column 3 of Figure~\ref{fig:allCLusterStates}(c),
again using position measurements to pare the DBSL down to a simpler cluster state.
That subfigure
depicts a portion of one of these two-dimensional sheets, revealing that the DBSL can be locally viewed as separate macronode wires connected periodically by beam splitters.

Column 1 of Figure~\ref{fig:allCLusterStates}(c,) depicts four of these coupled macronode wires stretched out into the rows of a two-mode cluster state. We will use this representation to propose a method of using the DBSL that doesn't require any such deletions.

The box indicates the teleportation gadget containing the four-mode beam-splitter network $\op B_\text{DBSL}$,
\begin{equation}\label{eq:U_DBSL}
    \begin{split}
        \raisebox{-3em}{$\op B_{\text{DBSL}} = $ \;}
\Qcircuit @C=1.5em @R=0.8em {
        & \qw & \qw & \bsbal[red]{1} & \qw & \nogate{1} \\
        & \qw & \qw & \qw & \qw & \nogate{2} \\
        & \qw & \bsbal{-1}[.>] & \bsbal[red]{1} & \qw & \nogate{3} \\
        & \qw & \qw & \qw & \qw & \nogate{4}
}
\raisebox{-3em}{\; ,}
\end{split}
\end{equation}
which has unitary matrix
\begin{align}
    \BSmat_{\text{DBSL}} 
    =
    \BSmat_\text{32} 
    \BSmat_\text{34} \BSmat_\text{12}
    &=\frac{1}{2} \begin{bmatrix}
        \sqrt{2} & -\sqrt{2} & 0 & 0 \\
        1 & 1 & 1 & -1 \\
        -1 & -1 & 1 & -1 \\
        0 & 0 & \sqrt{2} & \sqrt{2} 
    \end{bmatrix}.
\end{align}

While the beam-splitter networks for the BSL and DBSL differ only in the 
direction 
of a single beam splitter [compare Eq.~\eqref{eq:U_BSL} to Eq.~\eqref{eq:U_DBSL}], the cluster states at large have a more significant difference---the DBSL has double the connectivity between macronode wires; compare the DBSL and BSL in column 1 of Fig.~\ref{fig:allCLusterStates}. More connectivity provides more access to two-mode gates, as they can only be done between connected wires.

\subsubsection*{Completing the DBSL and connecting to the QRL}
\begin{figure}
    \includegraphics[width=0.9\columnwidth]{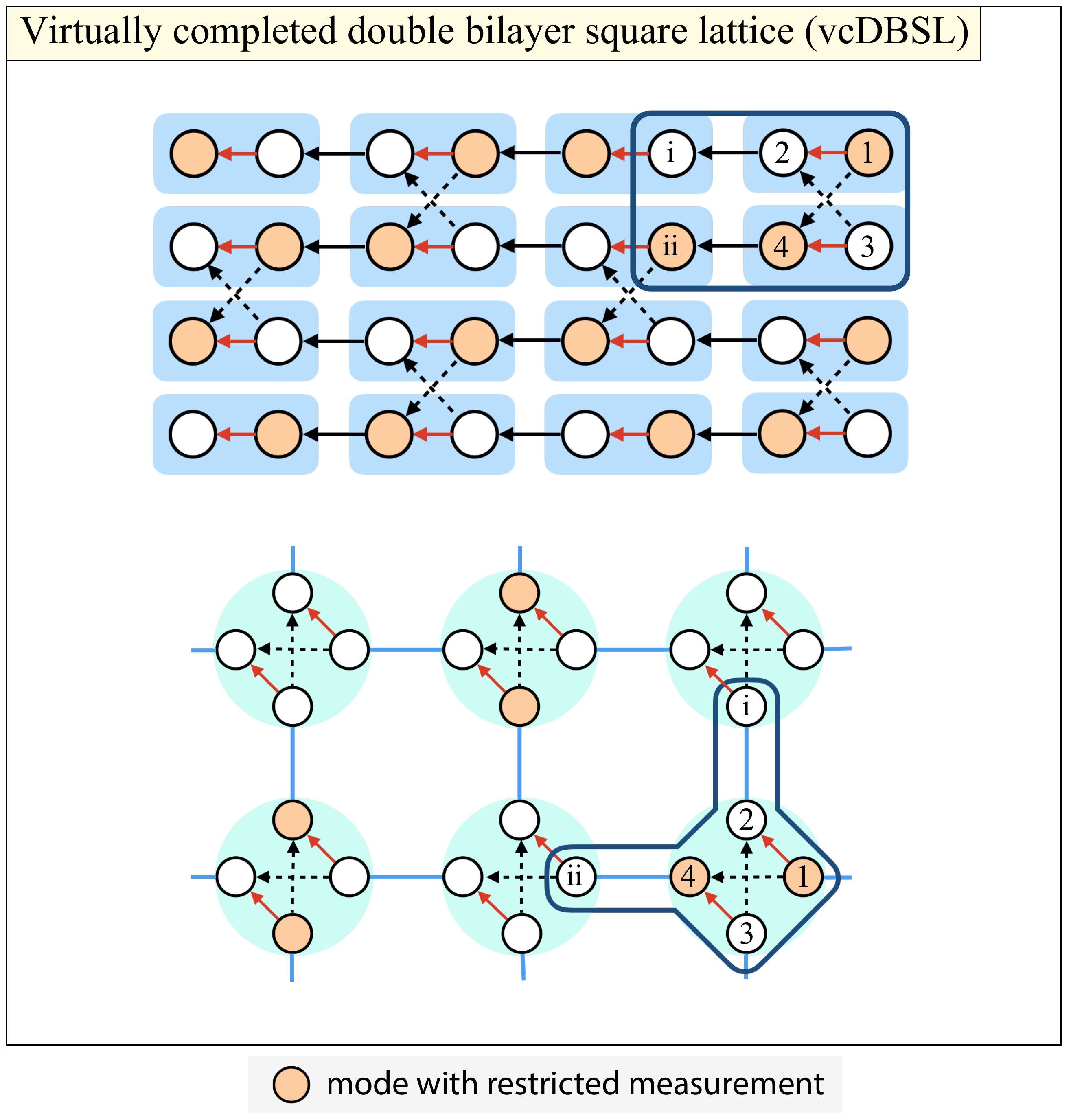}
    \caption{The DBSL with restricted measurement angles, allowing it to function as the cDBSL. Local modes in orange are measured in the same basis to virtually insert or delete beam splitters, see Sec.~\ref{subsection:virtualbeam splitters}. Compare with the unrestricted BSL in Fig.~\ref{fig:allCLusterStates}. The virtual beam splitters are those between orange modes. The choice here provides isolated teleportation gadgets in a dense brickwork structure like that of the QRL.}\label{fig:cDBSL}
\end{figure}

The DBSL's beam-splitter network, Eq.~\eqref{eq:U_DBSL}, is incomplete in the same way the BSL's is. 
In the same way, inserting a beam-splitter between modes 1 and 4
gives the \emph{completed DBSL} (cDBSL) with beam-splitter network $\op B_{\text{cDBSL}} \coloneqq \op{B}_{14} \op B_{\text{DBSL}}$. In circuit form,
\begin{equation}\label{eq:U_cDBSL}
    \begin{split}
        \raisebox{-3em}{$\op B_{\text{cDBSL}} = $ \;}
\Qcircuit @C=1.5em @R=0.8em {
	& \bsbal{3}[.>] & \qw & \bsbal[red]{1} & \qw & \nogate{1} \\
        & \qw & \qw & \qw & \qw & \nogate{2} \\
        & \qw & \bsbal{-1}[.>] & \bsbal[red]{1} & \qw & \nogate{3} \\
        & \qw & \qw & \qw & \qw & \nogate{4}
\raisebox{-3em}{\; ,}
}
\end{split}
\end{equation}
and with associated unitary matrix
\begin{align}
    \BSmat_{\text{cDBSL}} 
    =
     \BSmat_\text{14} \BSmat_\text{32}
    \BSmat_\text{34} \BSmat_\text{12}
    &=\frac{1}{2} \begin{bmatrix}
       1  & -1 & -1 & -1 \\
        1 & 1 & 1 & -1 \\
        -1 & -1 & 1 & -1 \\
        1 & -1 & 1 & 1 
    \end{bmatrix}.
\end{align}

The completed DBSL (cDBSL) beam-splitter network can be written in terms of the QRL under a parity operation and two SWAPs on the left and a parity operation on the right:
\begin{equation}
    \op B_{\text{cDBSL}}
    = \op F^2_3 \text{SWAP}_{23} \text{SWAP}_{34} \op B_{\text{QRL}} \op F_4^2
    \, ,
\end{equation}
which has a circuit  
\begin{equation}\label{eq:U_cDBSLTransformed}
    \begin{split}
        \raisebox{-3em}{$\op B_\text{cDBSL} = $ \;}
\Qcircuit @C=1.3em @R=0.8em {
	& \qw & \qw& \qw & \qw & \bsbal{2}[.>] & \bsbal[red]{1} & \qw & \qw  & \nogate{1} \\
        & \qw & \qswap   & \qw & \bsbal{2}[.>] & \qw & \qw & \qw & \qw  & \nogate{2} \\
        & \gate{F^2} & \qswap \qwx  & \qswap & \qw & \qw & \bsbal[red]{1} & \qw & \qw & \nogate{3} \\
        & \qw & \qw  & \qswap \qwx & \qw & \qw & \qw & \gate{F^2} & \qw & \nogate{4}
}
        \raisebox{-3em}{\; .}
\end{split}
\end{equation}
Using the cBSL beam-splitter network in the two-mode teleportation gadget, Circuit~\eqref{eq:2modegadget}, gives a two-mode gate
\begin{equation}
    \op V_\text{cDBSL}^{(2)}(\mat \theta) 
    = 
    \op F^2_2 \op B_{21}[\op V_1(\theta_1, \theta_4) \otimes \op V_2(\theta_2,\theta_3)]\op B_{12} .
\end{equation}
where the vector $\vec{\theta}$ specifies the measurement angles in the initial gadget, Circuit~\eqref{eq:2modegadget}. In circuit form, the two-mode gate is
\begin{equation}\label{twomodeGate_BSL}
    \begin{split}
\Qcircuit @C=1.1em @R=1.25em {
		& \qw & \qw&\gate{V(\theta_1,\theta_4)}  &  \bsbal{1} &\qw& \\
		& \gate{F^2} &\bsbal{-1}& \gate{V(\theta_2,\theta_3)} & \qw&\qw& \\
}
\end{split}
\; .
\end{equation}
Up to the parity operator, this gate is nearly the same as that for the QRL, Circuit~\eqref{VtwoMode_QRL_circuit}, except that some measurement angles have swapped positions.

\subsubsection*{Virtually completing the DBSL}

Now consider the DBSL in the two-mode teleportation gadget, Circuit~\eqref{eq:2modegadget}.
We complete the DBSL's beam-splitter network in the same way as we did for the BSL's, Eq.~\eqref{eqform:U_cBSL}: restricting the measurement bases on modes 1 and~4 and inserts a beam splitter between these modes. Doing so gives the circuit for the \emph{virtually completed DBSL} (vcDBSL),
\begin{equation}\label{eq:U_vcDBSL}
    \begin{split}
        \qquad  \qquad
\Qcircuit @C=1.5em @R=0.8em {
	\lstick{\brasub{\tfrac{m_1 - m_4}{\sqrt{2}}}{p_{\theta_1}}}           & \bsbal{3}[-->] & \qw            & \bsbal[red]{1} & \qw & \nogate{1} \\
    \lstick{\brasub{m_2}{p_{\theta_2}}}                                   & \qw            & \qw            & \qw            & \qw & \nogate{2} \\
    \lstick{\brasub{m_3}{p_{\theta_3}}}                                   & \qw            & \bsbal{-1}[.>] & \bsbal[red]{1} & \qw & \nogate{3} \\
    \lstick{\brasub{\tfrac{m_1 + m_4}{\sqrt{2}}}{p_{\theta_4 = \theta_1}}}& \qw            & \qw            & \qw            & \qw & \nogate{4}
}
        \raisebox{-3em}{\; ,}
\end{split}
\end{equation}
with the virtual beam splitter indicated as the dashed arrow.

The vcDBSL functions like the cDBSL in the teleportation gadget, with the same two-mode gate under the measurement restrictions above,
\begin{equation}
    \op V_\text{vcDBSL}^{(2)}(\mat \theta) 
    = \op V_\text{cDBSL}^{(2)}(\mat \theta) \vert_{\theta_1 = \theta_4}
\end{equation}
represented in circuit form as
\begin{equation}\label{twomodeGate_DBSL}
    \begin{split}
\Qcircuit @C=1.1em @R=1.25em {
		& \qw & \qw&\gate{V(\theta_1,\theta_4 = \theta_1)}  &  \bsbal{1} &\qw& \\
		& \gate{F^2} &\bsbal{-1}& \gate{V(\theta_2,\theta_3)} & \qw&\qw& \\
}
\end{split}
\; .
\end{equation}

In the larger structure of the cluster state, 
Fig.~\ref{fig:cDBSL} illustrates a choice of restricted measurement angles that gives a highly connected brickwork structure just like that of the QRL (and more densely connected than the BSL, Fig.~\ref{fig:cBSL}).

\subsection{3D DBSL Mikkel-splitter gadget (MSG)} \label{sec:Mikkel}
After their work with the DBSL, Larsen \emph{et al.}~\cite{Larsen2021fault} extended the DBSL into a 3D cluster state that utilises variable beam splitters to choose the location of their two mode gate. In doing so, they identify a teleportation gadget of the same form as Circuit~\eqref{eq:2modegadget}, which we dub the \emph{Mikkel-splitter gadget}~(MSG) after the given name of the first author. The MSG's beam-splitter network $\op{B}_\text{MSG}$ is
\begin{equation}\label{eq:U_DBSL-V}
    \begin{split}
        \raisebox{-3em}{$\op B_\text{MSG} = $ \;}
\Qcircuit @C=1.5em @R=0.8em {
		& \bsbal{3}[.>] & \qw & \bsbal[red]{1} & \qw & \nogate{1} \\
        & \qw & \qw & \qw & \qw & \nogate{2} \\
        & \qw & \qw & \bsbal[red]{1} & \qw & \nogate{3} \\
        & \qw & \qw & \qw & \qw & \nogate{4}
}
\raisebox{-3em}{\; ,}
\end{split}
\end{equation}
with the unitary matrix
\begin{align}
    \BSmat_\text{MSG} 
    =  \BSmat_\text{14} \BSmat_\text{34} \BSmat_\text{12}
    &=\frac{1}{2} \begin{bmatrix}
        1 & -1 & -1 & -1 \\
        \sqrt{2} & \sqrt{2} & 0 & 0 \\
        0 & 0 & \sqrt{2} & -\sqrt{2} \\
        1 & -1 & 1 & 1 
    \end{bmatrix}.
\end{align}

\subsubsection*{Completing the MSG and connecting to the QRL}
The MSG's beam-splitter network can be completed in a similar way to that of the BSL and the DBSL, with the virtual beam splitter instead being inserted between modes 2 and 3. 
Restricting these measurement bases leads to the same completed beam-splitter network and associated orthogonal matrix as the BSL. That is,
$\op B_{\text{cMSG}} \coloneqq \op{B}_{23} \op B_{\text{MSG}} = \op B_\text{cBSL}$, 
\begin{equation}\label{eq:U_cBSL}
    \begin{split}
        \raisebox{-3em}{$\op B_{\text{cMSG}} = $ \;}
\Qcircuit @C=1.5em @R=0.8em {
	& \qw & \bsbal{3}[.>] & \bsbal[red]{1} & \qw & \nogate{1} \\
        & \bsbal{1}[.>] & \qw & \qw & \qw & \nogate{2} \\
        & \qw & \qw & \bsbal[red]{1} & \qw & \nogate{3} \\
        & \qw & \qw & \qw & \qw & \nogate{4}
}
        \raisebox{-3em}{\; ,}
\end{split}
\end{equation}
with associated orthogonal matrix, 
\begin{align}
    \BSmat_{\text{cMSG}} 
    = \BSmat_\text{23} \BSmat_\text{14} \BSmat_\text{34} \BSmat_\text{12}
    &=\frac{1}{2} \begin{bmatrix}
        1 & -1 & -1 & -1 \\
        1 & 1 & -1 & 1 \\
        1 & 1 & 1 & -1 \\
        1 & -1 & 1 & 1 
    \end{bmatrix}.
\end{align}
Consequently, the connection to the QRL is also identical, Eq.~\eqref{eq:U_cBSLTransformed}, and the two-mode gate in the teleportation gadget is the same as the BSL, \eqref{VtwoMode_cBSL},
\begin{align} \label{VtwoMode_cMSG}
    \op{V}_{\text{cMSG}}^{(2)}(\vec \theta) = \op{V}_{\text{cBSL}}^{(2)}(\vec \theta)
    \, ,
\end{align}
with circuit form given in Circuit~\eqref{twomodeGate_cBSL}.

Reference~\cite{Larsen2021fault}'s 3D cluster state reduces down to a form of the DBSL when the second variable beam splitter therein is removed. This connection is clear from their identical completed four-splitter. However, this connection also identifies the useful fact that these four-splitters extend beyond 2D lattices and can be used to define two-mode gates wherever four modes meet together at a lattice site---regardless of dimension of the greater lattice. This will become more important as we generalise these four-splitters in Sec.~\ref{sec:four-splitters}.

\subsubsection*{Virtually completing the MSG}
Now consider the MSG in the two-mode teleportation gadget, Circuit~\eqref{eq:2modegadget}.
We complete the MSG's beam-splitter network by restricting the measurement bases on modes 2 and~3 and to insert a virtual beam splitter between these modes. Doing so gives the circuit for the \emph{virtually completed MSG} (vcMSG),
\begin{equation}\label{eq:U_vcMSG}
    \begin{split}
        \qquad  \qquad
\Qcircuit @C=1.5em @R=0.8em {
    \lstick{\brasub{m_1}{p_{\theta_1}}}                                  & \qw            & \bsbal{3}[.>] & \bsbal[red]{1} & \qw & \nogate{1} \\
    \lstick{\brasub{\tfrac{m_2 - m_3}{\sqrt{2}}}{p_{\theta_2}}}          & \bsbal{1}[-->] & \qw            & \qw            & \qw & \nogate{2} \\
    \lstick{\brasub{\tfrac{m_2 + m_3}{\sqrt{2}}}{p_{\theta_3=\theta_2}}} & \qw            & \qw            & \bsbal[red]{1} & \qw & \nogate{3} \\
    \lstick{\brasub{m_4}{p_{\theta_4}}}                                  & \qw            & \qw            & \qw            & \qw & \nogate{4}
}
        \raisebox{-3em}{\; ,}
\end{split}
\end{equation}
with the virtual beam splitter indicated as the dashed arrow.

In the teleportation gadget, the vcMSG functions like the cMSG with the same two-mode gate under the measurement-angle restrictions,
    \begin{equation}
        \op{V}_{\text{vcMSG}}^{(2)}(\vec \theta)  = \op{V}_{\text{cMSG}}^{(2)}|_{\theta_2=\theta_3}
    \end{equation}
in circuit form
\begin{equation}\label{twomodeGate_cMSG}
    \begin{split}
\Qcircuit @C=1.1em @R=1.25em {
		& \qw & \qw&\gate{V(\theta_1,\theta_2)}  &  \bsbal{1} &\qw& \\
		& \gate{F^2} &\bsbal{-1}& \gate{V(\theta_4,\theta_3=\theta_2)} & \qw&\qw& \\
}
\end{split}
\; .
\end{equation}

\subsection{Modified bi-layer square lattice (MBSL)}\label{sec:MBSL}

\begin{figure}
    \centering
    \includegraphics[width=0.9\columnwidth]{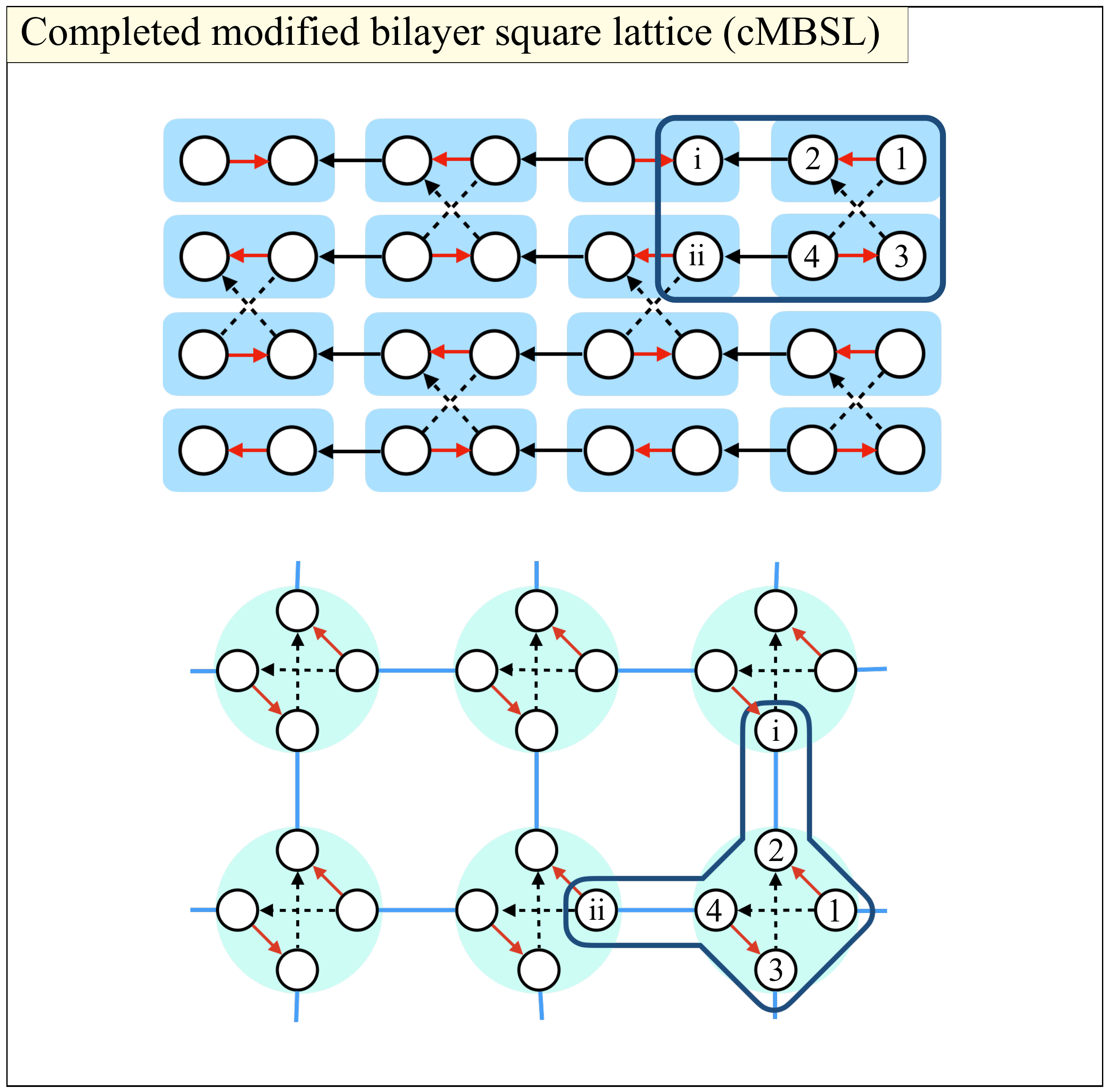}
    \caption{Adding an additional beam splitter to the MBSL produces the completed MBSL (cMBSL). The MBSL cannot be completed virtually by restricting measurement angles. The teleportation gadgets form a dense brickwork structure like that of the QRL.}
    \label{fig:MBSLfixedAngles}
\end{figure}

Asavanant \emph{et al.}~\cite{Asavanant2019detCVCS} introduced and experimentally created the modified bilayer-square-lattice (MBSL) cluster state. It is constructed by first coupling TMSSs into basic four-mode groupings,
\begin{align}\label{eq:MBSLclusterToBS}
    \includegraphics[width=0.4\columnwidth]{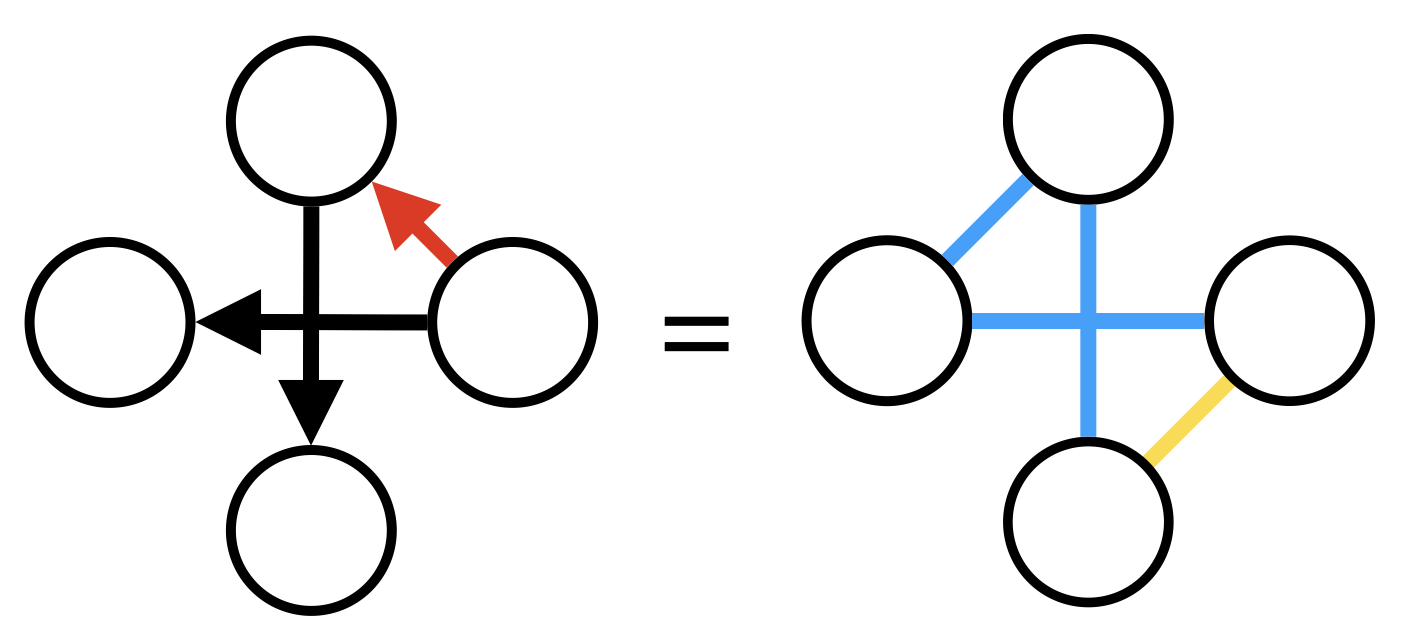}
\end{align}
and then stitching these groupings together using additional beam splitters; this is shown in column 3 of Fig.~\ref{fig:allCLusterStates}. When viewed from the perspective of coupled macronode wires, the MBSL has a different structure than the others considered in this work. That is because, in addition to the beam-splitter network being incomplete---a trait shared by the other cluster states---the macronode wires in the MBSL themselves are also incomplete, which can be seen in column 1 of Fig.~\ref{fig:allCLusterStates} as missing horizontal red beam splitters.. 

The beam-splitter network for the MBSL is
\begin{equation}\label{eq:U_MBSL}
    \begin{split}
        \raisebox{-3em}{$\op B_{\text{MBSL}} = $ \;}
\Qcircuit @C=1.5em @R=0.8em {
		& \bsbal{3}[.>] & \qw & \qw & \qw & \nogate{1} \\
        & \qw & \qw & \qw & \qw & \nogate{2} \\
        & \qw & \bsbal{-1}[.>] & \qw & \qw & \nogate{3} \\
        & \qw & \qw & \bsbal[red]{-1} & \qw & \nogate{4}
}
\raisebox{-3em}{\; ,}
\end{split}
\end{equation}
with the corresponding orthogonal matrix
\begin{align} \label{eq:MBSLmatrix}
    \BSmat_{\text{MBSL}} &=\frac{1}{2} \begin{bmatrix}
        \sqrt{2} & 0 & 1 & -1 \\
        0 & \sqrt{2} & 1 & 1 \\
        0 & -\sqrt{2} & 1 & 1 \\
        \sqrt{2} & 0 & -1 & 1 
    \end{bmatrix}.
\end{align}

The larger MBSL cluster-state construction has a brickwork-like structure, but as mentioned, it is missing critical beam splitters that complete the macronode wires themselves.

\subsubsection*{Completing the MBSL and connecting to the QRL}

Physically inserting the fourth beam splitter gives a new CV cluster state---the completed MBSL (cMBSL). The cMBSL's beam-splitter network is 
$\op B_{\text{cMBSL}} \coloneqq \op B_{\text{MBSL}} \op B_{12}$, given in circuit form as
\begin{equation}\label{eq:U_cMBSL}
    \begin{split}
        \raisebox{-3em}{$\op B_{\text{cMBSL}} = $ \;}
\Qcircuit @C=1.5em @R=0.8em {
	& \bsbal{3}[.>] & \qw & \bsbal[red]{1} & \qw & \nogate{1} \\
        & \qw & \qw & \qw & \qw & \nogate{2} \\
        & \qw & \bsbal{-1}[.>] & \qw & \qw & \nogate{3} \\
        & \qw & \qw & \bsbal[red]{-1} & \qw & \nogate{4}
\raisebox{-3em}{\; ,}
}
\end{split}
\end{equation}
with associated unitary matrix
\begin{align}
    \BSmat_{\text{cMBSL}} 
    =
     \BSmat_\text{14} \BSmat_\text{32}
    \BSmat_\text{43} \BSmat_\text{12}
    &=\frac{1}{2} \begin{bmatrix}
       1  & -1 & 1 & -1 \\
        1 & 1 & 1 & 1 \\
        -1 & -1 & 1 & 1 \\
        1 & -1 & -1 & 1 
    \end{bmatrix}.
\end{align}
The connection to the QRL is given by
\begin{align}
    \op B_\text{cMBSL} = \op F^2_3 \text{SWAP}_{14} \text{SWAP}_{23} \text{SWAP}_{34} \op B_\text{QRL}
\end{align}
with the circuit
\begin{equation}\label{eq:U_cMBSLTransformed}
    \begin{split}
        \raisebox{-3em}{$\op B_\text{cMBSL} = $ \;}
\Qcircuit @C=1.3em @R=0.8em {
      & \qw        & \qswap & \qw& \qw & \qw & \bsbal{2}[.>] & \bsbal[red]{1} & \qw & \qw  & \nogate{1} \\
      & \qw        & \qw \qwx & \qswap   & \qw & \bsbal{2}[.>] & \qw & \qw & \qw & \qw  & \nogate{2} \\
      & \gate{F^2} & \qw \qwx & \qswap \qwx  & \qswap & \qw & \qw & \bsbal[red]{1} & \qw & \qw & \nogate{3} \\
      & \qw        & \qswap \qwx& \qw  & \qswap \qwx & \qw & \qw & \qw & \gate{F^2} & \qw & \nogate{4}
}
        \raisebox{-3em}{\; .}
\end{split}
\end{equation}
Since the fourth beam splitter was physically inserted into the circuit, the completed MBSL has no measurement restrictions in its teleported two-mode gate:
    \begin{equation}
        \op V_\text{cMBSL}^{(2)}(\mat \theta) =  \op{F}_2^2 \op B_{21}[\op V_1(\theta_4, \theta_3) \otimes \op V_2(\theta_1,\theta_2)]\op B_{12} .
    \end{equation}
and can perform any gate that the QRL can via Eq.~\eqref{eq:U_cMBSLTransformed}.     
The larger cluster state structure of the cMBSL is the same as that for the cDBSL shown in Fig.~\ref{fig:cDBSL}. 
This is the same brickwork structure as the QRL, albeit with a different beam-splitter network. Thus, the cMBSL is as versatile as the QRL is for cluster-state computing as the QRL is---it can implement the same gates and can do it as often as the QRL can.

Unlike the other CV cluster states above, the MBSL cannot be completed by restricting measurement bases, because its missing fourth beam splitter is not on the measurement side of the circuit. Rather, it is on the state side of the circuit, meaning that a beam splitter must be physically included in the circuit---see Sec.~\ref{sec:incompletefour-splitters} for more details.

\section{Cluster-state computing with the GKP code} \label{sec:GKPcomputingwithCVCS}

Recent work~\cite{walshe2021streamlined} showed that the QRL is compatible with the Gottesman-Kitav-Preskill (GKP) encoding of a qubit into a mode~\cite{GKP}. %
This compatibility rests on two key features~\cite{walshe2021streamlined}.
First, the ability to perform GKP error correction can be built into the QRL cluster state itself by replacing the usual TMSSs with GKP Bell pairs made from GKP qunaught states (see Eq.~\eqref{eq:GKPqunaught} below) and beam splitters.
As the modes of the cluster state are measured, the potentially noisy input states to each teleportation gadget are automatically projected into the GKP code space, which is followed by an outcome-dependent correction.
Second, choosing appropriate measurement angles in the QRL teleportation gadget gives a complete generating set of GKP Clifford gates. These gates are executed in a single teleportation step and require only four modes in the beam-splitter network. 

Using GKP Bell pairs in the cluster state allows GKP error correction to be performed at the same time as each Clifford gate. 
Walshe \emph{et al.}~\cite{walshe2021streamlined} calculated the logical gate noise 
from simultaneous implementation of a GKP Clifford and GKP error correction 
and found that, given GKP qunaught states of high enough quality, fault tolerance is achievable by concatenation with a larger qubit code. 

The above features are not exclusive to the QRL.
Using the equivalence of the other cluster states to the QRL, established Sec.~\ref{sec:clusterStates}, we show here that
(a)~each cluster state can implement a full generating set of GKP Cliffords while at the same time implementing GKP error correction, and
(b)~the associated GKP Clifford-gate noise is identical to that of the QRL, thus making any of the cluster states suitable for fault-tolerant quantum computation.

\subsection{Constructing cluster states with GKP Bell pairs}\label{subsec:GKPBellpairs}

GKP Bell pairs are created by mixing two 
GKP qunaught states
on a beam splitter~\cite{Walshe2020}.
We consider (unnormalized) energy-constrained qunaught states, also called \emph{noisy qunaught states}, which are defined as
\begin{align} \label{eq:GKPqunaught}
    \ket{\bar \varnothing} = \op N(\beta) \ket{\varnothing} &= \op N(\beta) \int ds \; \Sha_{\sqrt{2\pi}}(s)\qket{s} ,
\end{align}
with energy damping operator
       $ \op N(\beta) \coloneqq  e^{-\beta \op a^\dag \op a}.$ 
The position (and momentum) noisy qunaught wavefunction is a comb of Gaussians with broad Gaussian envelope that damps peaks far from the origin. 
Two noisy GKP qunaughts on a beam splitter produces a noisy GKP Bell pair, which iteself is equivalent to a noisy GKP projector acting on one half of an EPR pair~\cite{Walshe2020},
\begin{align}
    \op{B}_{12} \ket{\bar{\qunaught}} \otimes \ket{\bar{\qunaught}}
    & =
    \op N_1(\beta) \otimes \op N_2(\beta) \ket{\Phi_\text{GKP}} \\
    & = \label{bellpairequivalence}
    \op N_2(\beta) \op \Pi_{\GKP,2} \op N_2(\beta) \ket{\EPR}
\end{align}
where $\ket{\Phi} \coloneqq \frac{1}{\sqrt{2}} (\ket{00} + \ket{11})$ is a two-qubit Bell state, $\ket{\EPR} \coloneqq \int ds \ket{s}_q \otimes \ket{s}_q$ is a two-mode EPR state, and the noisy GKP projector is
    \begin{equation} \label{eq:noisyGKPproj}
    \op{\bar{\Pi}}_{\GKP} \coloneqq \op N(\beta) \op \Pi_{\GKP} \op N(\beta) ,
    \end{equation}
with $\op \Pi_{\GKP}$ being the ideal GKP projector.

Replacing the ancilla states $\ket{\psi}$ and $\ket{\phi}$ in Circuit~\eqref{eq:2modegadget} with noisy qunaught states $\ket{\bar \varnothing}$ teleports the two inputs through two GKP Bell pairs during the implementation of Gaussian gate $\op{V}_\text{network}^{(2)}(\vec{\theta})$. Using the identity in Eq.~\eqref{bellpairequivalence}, the circuit becomes
\begin{equation}\label{eq:2modegadgetpart2}
    \begin{split}
\Qcircuit @C=0.8em @R=1em {
		&&& \lstick{\brasub{m_1}{p_{\theta_1}}}  &  \qw & \multigate{4}{B_\text{network}}  & \qw  & \nogate{\text{(in)}} & \\
		&&& \lstick{\brasub{m_2}{p_{\theta_2}}}  & \qw  &\ghost{B_\text{network}}  &  \qw  & \qw & \EPRdr& \\
	 & &  & &   &  & \nogate{\text{(out)}}  & \gate{\bar{\Pi}_\GKP}& \EPRur& \\
		&&& \lstick{\brasub{m_3}{p_{\theta_3}}} & \qw & \ghost{B_\text{network}}    & \qw  & \nogate{\text{(in)}}& \\
		&&& \lstick{\brasub{m_4}{p_{\theta_4}}}  &\qw& \ghost{B_\text{network}}  & \qw & \qw & \EPRdr& \\
		 &&&&&&\nogate{\text{(out)}}  &\gate{\bar{\Pi}_\GKP} & \EPRur& 
}
\raisebox{-5.5em}{\, ,}
\end{split}
\end{equation}
with the noisy GKP projectors appearing only on the output modes. 
Applying the circuit transformation rules presented in Ref.~\cite{Walshe2020} 
gives the simplified circuit
\begin{equation}\label{eq:2modeKraus}
\begin{split}
    \Qcircuit @C=1.4em @R=2em {
        & \nogate{\text{(out)}} & \gate{\bar{\Pi}_\GKP} & \multigate{1}{V^{(2)}_\text{network}(\bm{\theta})} & \qw & \rstick{\text{(in)}}& \\
        & \nogate{\text{(out)}} & \gate{\bar{\Pi}_\GKP} & \ghost{V^{(2)}_\text{network}(\bm{\theta})} & \qw &  \rstick{\text{(in)}} &
    }
    \end{split}
    \quad ,
\end{equation}
showing that the implemented gate, $\op{V}_\text{network}^{(2)}(\vec{\theta})$, depends only on the details of the beam-splitter network and measurement angles there; the noisy GKP projectors act afterwards. Thus, using GKP Bell pairs in the construction of any of the cluster states enables teleportation-based GKP error correction.
An alternative method is to insert GKP Bell pairs into a CVCS after the fact; a method to do so is given in Appendix~\ref{sec:GKPinsertion}.

\begin{table*}[t!]
    \begin{center}
    \bgroup
    \def\arraystretch{1.25}
    \begin{tabular}{|l|l|l|l|l|l|}
        \hline
    Architecture   & QRL      & vcBSL*       & vcDBSL*         & vcMSG*    \\
    mapping & $\{\theta^\text{Q}_1,\theta^\text{Q}_2,\theta^\text{Q}_3,\theta^\text{Q}_4\}$ &
     $\{\theta^\text{Q}_1,\theta^\text{Q}_2,\theta^\text{Q}_4,\theta^\text{Q}_3\}|_{\theta^\text{Q}_1=\theta^\text{Q}_3}$  & $\{\theta^\text{Q}_1,\theta^\text{Q}_4,\theta^\text{Q}_2,\theta^\text{Q}_3\}|_{\theta^\text{Q}_1=\theta^\text{Q}_3}$   & 
     $\{\theta^\text{Q}_1,\theta^\text{Q}_2,\theta^\text{Q}_4,\theta^\text{Q}_3\}|_{\theta^\text{Q}_2=\theta^\text{Q}_4}$  \\
    \hline \hline
    $\CZ(1)$   & $\{\tfrac{\pi}{2},\tfrac{\pi}{2}\pm \chi,\tfrac{\pi}{2},\tfrac{\pi}{2}\mp \chi\}$ & $\{\tfrac{\pi}{2},\tfrac{\pi}{2}\pm \chi,\tfrac{\pi}{2}\mp \chi,\tfrac{\pi}{2}\}$ & $\{\tfrac{\pi}{2},\tfrac{\pi}{2}\mp \chi,\tfrac{\pi}{2}\pm \chi,\tfrac{\pi}{2}\}$   & $\{-\chi,0,0,\chi \}^\ddagger$   \\
    SWAP  & $\{0,\tfrac{\pi}{2},\tfrac{\pi}{2},0\}$   & ----- $^\mathsection$      & ----- $^\mathsection$     & ----- $^\mathsection$   \\
    $\op I \otimes \op I$   & $\{\tfrac{\pi}{2},0,\tfrac{\pi}{2},0\}$   & $\{\tfrac{\pi}{2},0,0,\tfrac{\pi}{2}\}$ & $\{\tfrac{\pi}{2},0,0,\tfrac{\pi}{2}\}$    & 
    $\{\tfrac{\pi}{2},0,0,\tfrac{\pi}{2}\}$   \\
    $\op F \otimes \op F$   & $\{\tfrac{3\pi}{4},\tfrac{\pi}{4},\tfrac{3\pi}{4},\tfrac{\pi}{4}\}$ & 
    $\{\tfrac{3\pi}{4},\tfrac{\pi}{4},\tfrac{\pi}{4},\tfrac{3\pi}{4}\}$   & $\{\tfrac{3\pi}{4},\tfrac{\pi}{4},\tfrac{\pi}{4},\tfrac{3\pi}{4}\}$    & 
    $\{\tfrac{3\pi}{4},\tfrac{\pi}{4},\tfrac{\pi}{4},\tfrac{3\pi}{4}\}$   \\
    $\op P(1) \otimes \op P(1)$  & $\{\tfrac{\pi}{2},\tfrac{\pi}{2}\mp \chi,\tfrac{\pi}{2},\tfrac{\pi}{2}\mp \chi\}$ & 
    $\{\tfrac{\pi}{2},\tfrac{\pi}{2}\mp \chi,\tfrac{\pi}{2}\mp \chi,\tfrac{\pi}{2}\}$ & 
    $\{\tfrac{\pi}{2},\tfrac{\pi}{2}\mp \chi,\tfrac{\pi}{2}\mp \chi,\tfrac{\pi}{2}\}$    & 
    $\{\tfrac{\pi}{2},\tfrac{\pi}{2}\mp \chi,\tfrac{\pi}{2}\mp \chi,\tfrac{\pi}{2}\}$ \\
    \hline
    \end{tabular}
    \egroup
    \end{center}
    \caption{Dictionary of two-mode GKP Clifford gates. The mapping takes specific sets of measurement angles $\{\theta^\text{Q}_1,\theta^\text{Q}_2,\theta^\text{Q}_3, \theta^\text{Q}_4\}$ that implement gates for the quad-rail lattice (QRL) given in Ref.~\cite{walshe2021streamlined} and shows the reordering of and restrictions on those angles that perform the same gates on the bilayer square lattice (BSL), the double bilayer square lattice (DBSL), 
    and the Mikkel-splitter gadget (MSG) through their virtually completed counterparts vcBSL, vcDBSL, and vcMSG. The constant angle is defined $\chi = \arctan 2$~\cite{walshe2021streamlined}.
    Each mapping indicates a measurement-angle restriction required to complete the beam-splitter network---see Sec.~\ref{sec:clusterStates}. 
    }\label{tab:dictionary}
    \begin{tablenotes}[flushleft]
    \item \footnotesize{
    $^*$ Each gate is accompanied by a parity operator $\op F^2$ on the second mode. This has no effect of the noise properties of the gate operation. }
    \item \footnotesize{$^\ddagger$ These measurement angles, given in Ref.~\cite{Larsen2021fault}, implement $[\op F^\dag \otimes \op F]\CZ(1)$ and are not mapped from the QRL. Incidentally, our results show that this gate is available to the QRL by the inverse mapping.}  
    \item \footnotesize{$^\mathsection$ We found no mappings between SWAP gates for the QRL (isolated or combined with single-mode gates) compatible with the angle restrictions on the vcBSL, vcDBSL, and the vcMSG. }
    \end{tablenotes}
\end{table*}

\subsection{Implementing GKP Clifford gates}

For the square-lattice GKP code, a set of CV unitaries that enables a complete generating set GKP Clifford gates is~\cite{walshe2021streamlined}
\begin{align} \label{eq:GKPCliffords}
    \underbrace{\left\{\op{I}, \op{F}, \op{P}(\pm 1), \op{\mathrm{C}}_Z(\pm 1)\right\}}_{\text {CV unitaries }} \longmapsto \underbrace{\left\{\bar{I}, \bar{H}, \bar{P}, \overline{\mathrm{C}}_Z\right\}}_{\text {GKP Cliffords }}.
\end{align}
where $\op{F} = \op{R}(\tfrac{\pi}{2})$, $\op{P}(\sigma)$ is given in Eq.~\eqref{eq:posshear}, and $\CZ(g) \coloneqq e^{ig \op{q} \otimes \op{q}}$.
The single-qubit GKP Cliffords are implemented in a single teleportation step of a macronode wire~\cite{Walshe2020,walshe2021streamlined} and are thus accessible to any of the cluster states above since we have expressed them as coupled macronode wires. Whether a two-qubit gate required to complete the set of GKP Clifford gates is accessible depends on the beam-splitter network in the teleportation gadget. 
Walshe \emph{et al.}~\cite{walshe2021streamlined}  found that, in a single teleportation step, the QRL can implement $\CZ(1)$, the SWAP gate, and a suite of simultaneous single-mode gates. The measurement angles that do so are given in the first column of Table~\ref{tab:dictionary}. 

The completed cluster states---cBSL, cDBSL, cMSG, and cMBSL---each have access to all of the QRL's gates using the correspondences in the previous section. But what about the virtually completed versions, which are completed by restricting measurement angles? 
In Table~\ref{tab:dictionary}, we leverage the connections between the cluster states above and the QRL to find the measurement angles that implement the same gates on the virtually completed cluster states, up to parity operations (which do not affect GKP codes).%
\footnote{Single-mode GKP codes on any lattice centered at the origin are even parity~\cite{Grimsmo2020}.}
We find that the vcBSL, vcDBSL, and vcMSG can natively implement almost all of the 
Clifford gates that the QRL can---the SWAP is the only one for which we could not find a set of angles. 
In the last column, we give the measurement angles from Ref.~\cite{Larsen2021fault} that generate $\CZ(1)$ (up to local Fourier transforms) for the MSG. This gate was found natively for that cluster state~\cite{Larsen2021fault} and does not result from the QRL connection in this work---highlighting the fact that this list is not exhaustive. Additional measurement angles implementing the same logical-GKP gates may exist for any of the cluster states, including the non-completed versions. As a final note, under the measurement-angle restrictions required to complete each beam-splitter network, we found no available SWAP gate.

\subsubsection*{Logical gate noise}
As a cluster state is measured, unavoidable noise accumulates due to the fact that the cluster states are built from finitely squeezed resources---TMSSs or GKP Bell pairs. The CV noise translates to logical qubit errors when the cluster state is combined with encoded GKP qubits. For high-quality GKP input states, Walshe \emph{et al.}~\cite{walshe2021streamlined} calculated the logical noise accompanying each GKP Clifford gate in Eq.~\eqref{eq:GKPCliffords} and found that a 1\% gate error rate is achievable with 
GKP resources and a CV cluster state both having 
squeezing of at least 11.9 dB. This minimum value is 
set by the the two-mode $\CZ(1)$ gate, which is the poorest performing gate in the set. 

The origin of this noise 
lies in the 
fact that all states used in the protocol 
have finite squeezing; 
none is ideal. This is quantified for the GKP Bell pairs using the damping operator~$\op{N}(\beta)$, Eq.~\eqref{eq:GKPqunaught}. Inspecting Circuit~\eqref{eq:2modeKraus} reveals that the damping operators are entirely contained in the noisy GKP projectors, Eq.~\eqref{eq:noisyGKPproj}, and are independent of the details of the beam-splitter network and measurements that enact the gate~$\op{V}_\text{network}^{(2)}(\vec{\theta})$. This means that if an identical $\op{V}_\text{network}^{(2)}(\vec{\theta})$ is found for two different cluster state constructions, the associated logical-gate noise will be the same. 

All of the completed cluster states---the cBSL, cDBSL, cMSG, and cMBSL---can implement any gate that the QRL can, including the entangling GKP Clifford gate $\CZ(1)$ (up to single-mode Fourier transforms and parity operators).
Table~\ref{tab:dictionary} shows that the vcBSL, vcDBSL, and vcMSG can also implement $\CZ(1)$ natively as well as any subset of the QRL's gates that respect the angle restrictions---SWAP is a notable exception.
This means that the gate noise on $\CZ(1)$ is the same as the QRL's, and each cluster state inherits the logical-gate noise (as a function of squeezing) calculated in Ref.~\cite{walshe2021streamlined} for the gates native to the architecture---all gates for physically completed four-splitters and a subset of the gates for virtually completed ones.

Our conclusion appears to conflict with a prior study of gate noise for the 
BSL and DBSL~\cite{Larsen2020architecture}, which concluded that higher-quality GKP resources were needed there than for the QRL. In that study, deleting certain modes from the BSL and DBSL (using position measurements) was required before computation in order to isolate macronode wires. 
Each deleted mode introduces noise into the remaining ones, increasing the requirements on the GKP resources.
In contrast, the equivalence of the BSL and DBSL with the QRL
in present work arises from utilizing each cluster state efficiently by selecting macronode wires---distinct from those in Ref.~\cite{Larsen2020detgatesCVCS}---that require no deletions before use.

We note that although the GKP logical Clifford gate noise is identical across the different cluster states, they may have different uses due to inherently different connectivity. 
For example, the DBSL has a denser set of teleportation gadgets than the BSL, permitting two-mode gates more often; compare Figs.~\ref{fig:cBSL} and~\ref{fig:cDBSL}. The effects of these differences on explicit calculations of fault-tolerance thresholds is left to future work. Our purpose here is to show that many cluster states perform just as well as the QRL in terms of gate noise and thus inherit the favorable performance reported in Ref.~\cite{walshe2021streamlined}.

\section{Balanced four-splitters}\label{sec:four-splitters}

In the above section, we compared a small set of cluster state constructions that have been introduced in the literature.
The QRL is unique among these constructions in that it stitches together one half of four TMSSs (or GKP Bell pairs~\cite{walshe2021streamlined}) using \emph{four} beam splitters; all the other constructions use three. The key to connecting the cluster states above to the QRL is completing them by inserting a fourth beam splitter---either physically (in the case of the MBSL) or virtually using restricted measurements (all others). 
Here, we show that these completed are members of
a class of four-mode transformations generated by \emph{balanced four-splitters}, a term coined in Ref.~\cite{Alexander2016flexible}.%
\footnote{In that reference, the term referred specifically to~$\op B_\QRL$, Eq.~\eqref{eq:U_QRL}. Here we expand this term to include a larger class of unitary gates, of which $\op B_\QRL$ is one example.}
This extends the gate-noise results---and others that may extend from this work---to apply not only to the small set of cluster states discussed here but to every cluster state that might be conceived of that shares the four-splitter properties described herein. For generality, we extend this definition to $n$~modes.

\begin{definition}[Balanced $n$-splitter]
Let $\op{U}$~be an $n$-mode linear-optical unitary operator with $\op{U} \to \mat U \in \mathrm{U}(n)$ [see Eq.~\eqref{eq:unitaryarrow} for notation]. $\op{U}$~is called a \emph{balanced $n$-splitter} if all elements
of $\mat U$ have the same magnitude. The matrix $\mat U$ is called the \emph{$n$-splitter matrix} corresponding to~$\op U$.
\end{definition}

\begin{definition}[Real balanced $n$-splitter]
Let $\op{U}$ be a balanced $n$-splitter. $\op{U}$~is called a \emph{real balanced $n$-splitter} if its $n$-splitter matrix $\mat U \in \mathrm{O}(n)$ or, equivalently, if all elements of~$\mat U$ are real.
\end{definition}
 
Four-splitters over $n$ modes map every mode operator~$\op a_j$ to an equal-magnitude (\emph{i.e.},~balanced) linear combination of all $n$ mode operators,
namely $\op{U}^\dagger \op{a}_j \op{U} = \sum_{k=1}^n U_{jk} \op{a}_k$ with $|U_{jk}|= \frac{1}{ \sqrt{n} }$. 
Given a balanced $n$-splitter~$\op U \to \mat U$, then $\mat U$ is $n^{-1/2}$ times an order-$n$, complex Hadamard matrix~\cite{colbourne2007handbook,Horadam2007}, by definition of the latter. When $\op U$ is a real balanced $n$-splitter, then the associated Hadamard matrix is real rather than complex.

Our work will largely focus on real balanced $n$-splitters. We need the following important definition:

\begin{definition}[Equivalence of real balanced $n$-splitters]
Two real balanced $n$-splitters~$\op U_1 \to \mat R_1$ and $\op U_2 \to \mat R_2$ are called \emph{equivalent} if $\sqrt n \mat R_1$ and $\sqrt n \mat R_2$ are equivalent as order-$n$ real Hadamard matrices~\cite{hirasaka2016}---i.e.,~if one can be transformed into the other by a sequence of row permutations, column permutations, row negations, and column negations.
\end{definition}

For real balanced $n$-splitters with $n \in \{2, 4, 8, 12\}$, all belong to a single equivalence class for the given~$n$~\cite{hirasaka2016}, but for larger~$n$, multiple equivalence classes exist. We will focus on the cases of $n=2$ and $n=4$ below.

Immediately from the definition above, we see that the beam splitter~$\op B_{jk} \to \mat R_{jk}$ defined in Eq.~\eqref{bsmat} is a real balanced two-splitter. But it is not the only one. For two modes, there are exactly eight possible real balanced two-splitters, corresponding to four cases of exactly one negative sign in $\mat R_{jk}$ and four cases of exactly three negative signs in $\mat R_{jk}$. This is a property of the set of order-2 (real) Hadamard matrices~\cite{hirasaka2016}.

One can straightforwardly verify that for this small set of Hadamard matrices, swapping the rows is equivalent to negating one column, and swapping the columns is equivalent to negating one row. Also, if negation of each row is allowed, then negation need only be allowed on one fixed column (say, the second column) to generate all possible real two-splitters from a given one. Thus, the full set of equivalence operations can be reduced to simply row negation and negation of the second column. This is an example of a smallest generating set, as well, by a simple counting argument: since there are $8=2^3$ real balanced two-splitters, a minimum of three binary operations is necessary to generate all of them from a single one.

The set of real balanced four-splitters displays similar behavior. While the QRL four-splitter~$\op B_\QRL$ [Eq.~\eqref{eq:U_QRL_styled}] defines one particular example of a real balanced four-splitter, many other examples exist. In fact, there are 768 real balanced four-splitters, one for each order-4 real Hadamard matrix~\cite{Yoshikawa2016}.
We are interested in treating this set of matrices
as an equivalence class under a minimal set of matrix operations because each of these corresponds to a linear-optical unitary gate: permutations are SWAP gates, and negating a row or column is a $\op F^2$ gate before or after the four-splitter, respectively. We will use this below to demonstrate the functional equivalence of all real balanced four-splitters in the CV cluster-state setting.

Theorem \ref{equivClass}, proven next, shows that the equivalence class of order-4 Hadamard matrices can be generated with a smaller set of matrix operations than is generally needed---row negations, row permutations, and a single column negation are sufficient. This is indeed the minimal set of operations.

\begin{theorem}\label{equivClass}
Let $\op U_1 \to \mat R_1$ and $\op U_2 \to \mat R_2$ be real balanced four-splitters. Then, $\mat R_1$ can be transformed into $\mat R_2$ by  some sequence of row permutations, row negations, and negation of any single column.
\end{theorem}

\begin{proof}
Since every $\mat R_{1,2} = \tfrac 1 2 \mat H_{1,2}$, where $\mat H_1$ and $\mat H_2$ are order-4 real Hadamard matrices. Our strategy will be to show that the set of all order-4 real Hadamard matrices, $H_4$, is a single equivalence class according to the conditions of equivalence in the theorem. Since $\abs{H_4} = 768$~\cite{hirasaka2016}, we will use the available operations to construct 768 unique matrices, each of which is a member of~$H_4$.

Let
\begin{align}
    \mat H_0 =
    \begin{bmatrix}
        1 & 1 & 1 & 1 \\
        1 & -1 & 1 & -1 \\
        1 & 1 & -1 & -1 \\
        1 & -1 & -1 & 1 
    \end{bmatrix}
    \in H_4
    .
\end{align}
Notice that the row parity (number of minus signs in a given row)
of each row of~$\mat H_0$ is even. 
Orthogonality ensures that the $4!$ possible row permutations and $2^4$ possible row negations generate distinct matrices (since no row can equal any other row, nor can it equal any row's negation). This accounts for $4! \times 2^4 = 384$ of the elements of $H_4$ (including ~$\mat H_0$). All of these have all-even row parity. To access the other half of~$H_4$, we negate the last column of each of the matrices already generated. Each of these new matrices is distinct from any of the others already obtained because row permutations and row negations preserve row parity, but a single column negation flips the parity of all rows. Thus, we have produced another 384 members of~$H_4$ (with odd row parity) that are all distinct from each other and from the other half (with even row parity). This exhausts the 768 members of~$H_4$. The above procedure can be repeated using any of the other columns in place of the last one, showing the choice of column to be arbitrary. The result equally applies after multiplying each of these matrices by~$\tfrac 1 2$. This 
proves the theorem.
\end{proof}

\noindent \textbf{Note:} This subsection has used the precise term \emph{real balanced four-splitter} in order to distinguish it from (more general) balanced four-splitters. In the rest of this work, where the context is clear, the term \emph{four-splitter} is used to mean a real balanced four-splitter, which is the focus of our work.

\subsection{Physical constructions of four-splitters}\label{four-splitterImplementations}

First, we prove the following theorem about beam-splitter networks:
\begin{theorem}
Let $\op{B}_\text{network}$ be a beam splitter network comprising four real balanced beam splitters [Eq.~\eqref{eq:beamsplitter_circuit}]. Then, $\op B_\text{network}$ is a real balanced four-splitter if and only if all of the following conditions are satisfied:
\begin{enumerate}
    \item \label{item:eachtwo} Each mode interacts with exactly two beam splitters.
    \item \label{item:spreadout} All modes interact with a first beam splitter before any mode interacts with a second.
    \item \label{item:nodoubling} No pair of beam splitters interacts with the same two modes.
\end{enumerate}
\end{theorem}
\begin{proof}
Let $\op{B}_\text{network} \to \mat R$. We proved the result by exhaustive search. Each beam splitter can start at any of the four modes and end at any of the three other modes (${4 \times 3}$), and there are four of these choices to be made independently, giving a total of ${(4 \times 3)^4 = 20\,736}$ possible versions of~$\mat R$. We constructed each of these matrices and checked if it was a real balanced four-splitter matrix. Independently, we checked whether each beam-splitter network satisfied each of the three conditions individually and in combination. We found, by comparing these lists, that satisfaction of all three conditions is both necessary and sufficient for $\mat R$ to be a real balanced four-splitter matrix.
\end{proof}

A bit more intuition into this result is possible.
Given conditions~\ref{item:spreadout} and~\ref{item:nodoubling}, without loss of generality, we can choose to partition the four modes, labelled~$(j, k, \ell, m)$, as $\{ \{j, k\}, \{\ell, m\} \}$, such that the first two beam splitters act on modes within the same partition. Condition~\ref{item:eachtwo} now guarantees that the last two beam splitters act across the partition $\{ \{j, \ell\}, \{k, m\} \}$. (If this partitioning is not respected, just swap mode labels $\ell$ and~$m$ in the first step.) This means that $\mat R$ has the form
\begin{align}
    \mat R
&=
    \frac 1 2
    \begin{bmatrix}
        \pm & 0 & \pm & 0 \\
        0 & \pm & 0 & \pm \\
        \pm & 0 & \pm & 0 \\
        0 & \pm & 0 & \pm
    \end{bmatrix}
    \begin{bmatrix}
        \pm & \pm & 0 & 0\\
        \pm & \pm & 0 & 0\\
        0 & 0 & \pm & \pm \\
        0 & 0 & \pm & \pm
    \end{bmatrix}
    ,
\end{align}
where each $\pm$ independently represents either~$+1$ or $-1$ (fixed by the beam splitters in question), and the row and column ordering are $(j, k, \ell, m)$. This block form guarantees that~$\mat R$ has entries with equal magnitude. The fact that it represents a product of real beam-splitter matrices guarantees that $\mat R \in \mathrm{O}(4)$, which proves sufficiency of the conditions. We did not find an elegant way to rigorously prove necessity of the conditions, but the exhaustive search guarantees that
\begin{align}
\label{eq:foursplitterdecomp}
    \mat R
&=
    \FSmat_{km}^{(\tp)}
    \FSmat_{j \ell}^{(\tp)}
    \FSmat_{\ell m}^{(\tp)}
    \FSmat_{jk}^{(\tp)} \eqqcolon  \FSmat_\text{4-split}
\end{align}
is a real balanced four-splitter matrix, where $^{(\tp)}$ indicates that the matrix shown may be transposed or not. This means the form above allows each beam splitter to point in either direction between the modes in the subscript. Note that the first two matrices commute, $[\FSmat^{(\tp)}_{jk}, \FSmat^{(\tp)}_{\ell m}] = 0$, and so do the last two, $[\FSmat^{(\tp)}_{j \ell}, \FSmat_{km}^{(\tp)}] = 0$---the corresponding beam-splitter pairs exhibit the same commutations.

Next, we note that
not all four-splitters can be decomposed into four real beam splitters alone, each of the form in Eq.~\eqref{eq:beamsplitter_circuit}. 
This can be seen with a simple counting argument.
Consider four beamsplitters over four modes arranged according to Eq.~\eqref{eq:foursplitterdecomp}.
There are $4!$ ways to freely assign mode labels and
$2^4$ ways to choose to transpose or not
each beam-splitter matrix.
This gives
$4! \times 2^4 = 384$
beam-splitter networks. Noting the commutativity within the beamsplitter pairs, as described below Eq.~\eqref{eq:foursplitterdecomp}, we can reduce this total by a factor of $2^2$ to account for this possible reordering of the beam splitters within each pair. Accounting for this operational equivalence (since the beamsplitters can act simultaneously) gives $384/4 = 96$ physically distinct beam-splitter networks.

By explicit computation, we found that these 96 physical beam-splitter networks generate only 40 unique four-splitters out of the 768 total. This is due to collisions, where more than one beam-splitter network produces the same four-splitter matrix. Indeed, it was observed previously that pairs of beamsplitters in the QRL commute~\cite{Menicucci2011tempmodeCVCS,Wang2014hypercubic,Alexander2016flexible}, 
 \begin{equation}\label{eq:U_QRL_commute}
    \begin{split}
        \raisebox{-0.94cm}{$\op B_{\text{QRL}} = $ \;}
\Qcircuit @C=0.5cm @R=0.60cm {
		& \qw & \bsbal{2}[.>] & \bsbal[red]{1} & \qw  \\
        & \bsbal{2}[.>] & \qw & \qw & \qw  \\
        & \qw & \qw  & \bsbal[red]{1} & \qw \\
        & \qw & \qw & \qw & \qw 
}
\raisebox{-0.94cm}{\; $ = $ \;}
\Qcircuit @C=0.5cm @R=0.60cm {
        & \bsbal[red]{1} & \qw & \bsbal{2}[.>] & \qw  \\
        & \qw            & \bsbal{2}[.>] & \qw & \qw  \\
        & \bsbal[red]{1} & \qw           & \qw & \qw \\
        & \qw            & \qw           & \qw & \qw 
}
        \raisebox{-0.94cm}{\; ,}
\end{split}
\end{equation}
giving a two-to-one correspondence with a four-splitter matrix.%
\footnote{Unlike the commutation within beam-splitter pairs, this commutation relation still requires a definite ordering of the pairs to be chosen in the experiment. It just so happens that either ordering gives the same result.}
We find that 48 of the 96 beam-splitter networks exhibit the same behavior: one pair of beam splitters commutes with the other pair.
Consequently, this set generates 24 four-splitters. The remaining 48 physical beam-splitter networks have a different property that gives three-to-one correspondences with four-splitters. Each beam-splitter network of this type is equivalent to two other beam-splitter networks. An example is the cBSL,
\begin{align}
 &\op B_\text{cBSL} =
 \\* \nonumber
&\Qcircuit @C=0.5cm @R=0.6cm {
		& \qw & \bsbal{3}[.>] & \bsbal[red]{1} & \qw  \\
        & \bsbal{1}[.>] & \qw & \qw & \qw  \\
        & \qw & \qw & \bsbal[red]{1} & \qw \\
        & \qw & \qw & \qw & \qw
}
 \raisebox{-0.94cm}{\;=\;} 
\Qcircuit @C=0.5cm @R=0.6cm {
		& \bsbal{2}[-->] & \qw & \qw & \bsbal{3}[.>] & \qw  \\
        & \qw & \qw & \bsbal{1}[.>] & \qw & \qw \\
        & \qw & \qw & \qw & \qw & \qw  \\
        & \qw & \bsbal{-2}[-->] & \qw & \qw & \qw
} 
 \raisebox{-0.94cm}{\;=\;} 
\Qcircuit @C=0.5cm @R=0.6cm {
		& \bsbal[red]{1} & \bsbal{2}[-->] & \qw & \qw  \\
        & \qw & \qw & \qw & \qw \\
        & \bsbal[red]{1} & \qw & \qw & \qw \\
        & \qw & \qw & \bsbal{-2}[-->] & \qw
}
\end{align}
The 48 physical beam-splitter networks of this type produce the remaining 16 four-splitters and complete the set of 40 that are possible using beam-splitter networks alone.

The above reveals that we must go beyond beam-splitter networks alone in order to construct the full set of 768 four-splitters. However, starting from any of these 40, Theorem \ref{equivClass} guarantees that we can generate the full equivalence class given three types of matrix operation: row swaps, row negation, and a column negation. These matrix operations are implemented physically with SWAP gates and parity operators, as described in Eq.~\eqref{SWAPmatrix} and Eq.~\eqref{paritymatrix}, respectively. 
Following Theorem~\ref{equivClass}, we include a single $\op F^2$ before the beam-splitter network, up to four $\op F^2$ after it (one for each mode), and final SWAP gates that may 
permute the modes in any way. 

Starting from the 96 physical beam-splitter networks constructed above, 2 choices of parity operator beforehand, $2^4$ afterwards, and $4!$ mode swaps gives $96 \times 2 \times 2^4 \times 4! = 73\,728$ total linear optical networks. In fact, this is exactly 96 times 768, the number of unique four-splitters. Theorem~\ref{equivClass} guarantees that each four-splitter is represented. The question is how many times each one appears. We numerically computed every network's orthogonal matrix 
and
found that
each unique four-splitter is realized by 96 different linear-optical networks.%
\footnote{Note: This is 96 linear-optical networks that construct a \emph{single} unique balanced four-splitter, as opposed to the 96 physical beam-splitter networks, referred to earlier, which can generate 40 different balanced four-splitters between them.}

\subsection{Incomplete four-splitters and how to complete them}\label{sec:incompletefour-splitters}

Most of the cluster states we analyzed in Sec.~\ref{sec:clusterStates} have at their heart beam-splitter networks comprising three beam splitters (rather than four). 
In each case, we went through a procedure to \emph{complete} them with a fourth beam splitter added either virtually by restricting measurement angles or, in the case of the MBSL,
by inserting the fourth beam splitter physically. These illustrate the two types of incomplete beam-splitter network we consider---either the missing fourth beam splitter is (a) on the measurement side of the circuit or (b) on the state side of the circuit. 

The two types of incomplete four-splitter correspond to respectively dropping (a)~$\FSmat_{km}^{(\tp)}$ or (b)~$\FSmat_{jk}^{(\tp)}$ in Eq.~\eqref{eq:foursplitterdecomp}. For type~(a), this gives orthogonal matrices where two of the rows have two 0 entries and two $\pm \sqrt{2}$ entries, such as for the BSL in Eq.~\eqref{eq:BSLmat}. For type~(b), the orthogonal matrices have two columns with two 0 entries and two $\pm \sqrt{2}$ entries, such as for the MBSL in Eq.~\eqref{eq:MBSLmatrix}. Importantly, the operations that define the equivalence class in Theorem \ref{equivClass} do not couple the two types of incomplete four-splitter---doing so would require a transpose-type operation, which is not available.

Restated in terms of physical operations: beam-splitter networks of type~(b) cannot be transformed to type~(a) using the left action of SWAPs and parity operators---\emph{i.e.},~on the measurment side. 
If this were possible, it would allow type (b) to be completed by restricting measurement angles in a teleportation gadget.

One last situation must be addressed. 
Is it possible to complete a type~(b) beam-splitter network by restricting measurements and post-processing the outcomes in a more general way? We show in Appendix~\ref{threesplitterinequivalence} that the answer is ``no,'' which ultimately means that incomplete four-splitters of type~(b) cannot be completed by restricting measurements---the fourth beam splitter must be included in another way.%
\footnote{Of course, \emph{any} real beam-splitter network can be generated virtually by choosing all measurement angles to be the same. This does nothing of value, however, since it also trivially undoes any other beam splitters on the associated modes.}

\section{Conclusion}
By demonstrating an equivalence to the well-studied QRL architecture, we have shown that the BSL, DBSL, MSG, and MBSL architectures are all compatible with the streamlined quantum computing techniques introduced in Ref.~\cite{walshe2021streamlined}, which provide a complete Clifford gate set for the GKP code within a single teleportation gadget, minimizing gate noise arising from noise in the ancillae. 

To show this equivalence, we identify macronode wires that periodically couple to one another throughout each cluster state. 
This identification allows the cluster states to be used without introducing additional noise that arises from deleting modes. 

The equivalence to the QRL requires \emph{completing the four-splitter} within each type of cluster state. For the BSL, DBSL, and MSG, correlating specific measurement angles produces virtual beam-splitters that complete the four-splitters, such that the couplings behave just like the ones in the QRL do, up to SWAP gates and single-mode parity operators. Once completed, each cluster state has access natively to a generating set of 
GKP Clifford gates (inherited from the QRL but limited to those that satisfy the angle restrictions), which have the same intrinsic gate noise as the QRL~\cite{walshe2021streamlined}. These are sufficient to implement any Clifford gate.%
\footnote{Some work remains to discern the effects of these measurement restrictions on any specific architecture, and we leave that to those doing the practical implementations. We have shown, nevertheless, that a generating set of Clifford gates is available natively for each architecture.}

Physically completing the four-splitter (required for the MBSL, optional for the others) eliminates these angle restrictions, allowing all QRL gates to be implemented natively on the new cluster state. This shows that all architectures based on complete four-splitters are entirely equivalent in terms of their available gates and noise properties due to finite squeezing.

The specific cluster states in Sec.~\ref{sec:clusterStates} are a small subset of those available to experimentalists with passive optical components. The analysis in Sec.~\ref{sec:four-splitters} of all balanced four-splitters that are equivalent under trivial passive
operations (SWAP and parity gates) provides new pathways to designing scalable cluster states. In practice, the class of four-splitters (and the cluster states they comprise) that can be completed by correlating measurement angles may be more useful than the QRL-type cluster states due to fewer optical components (one fewer beam splitter). That is because fewer optical components can allow for simpler experimental set ups.

Further, the noise equivalence we show here applies only to gate noise introduced by imperfect GKP states throughout the cluster. Optical components themselves introduce losses and other noise due to imperfections and alignment errors; reducing their number could improve gate noise under more realistic models. 
On the other hand, the connectivity in the context of each cluster state at large may favor a true four-splitter in the teleportation gadget in order to leverage more flexibility for routing quantum information~\cite{Alexander2016flexible} and implementing simultaneous gates. For a given computation, there may be advantages to choosing a specific cluster state to reduce logical gate depth and consequently logical noise.

In this work, we focus on Clifford gates for the square-lattice GKP code. It is not known whether cluster states can implement a full set of Clifford gates for GKP codes on other lattices in a single teleportation step, and, if so, \emph{which} cluster states do.
Our general analysis of four-splitters should be useful in this study: once a particular four-splitter and measurement angles are found that provide the full gate set, equivalent four-splitters (and measurement angles) are immediately accessible via the SWAP and parity operations that define the equivalence class. 
Also, the fact that we have an analytic form for the QRL's teleported gate, Eq.~\eqref{VtwoMode_QRL}, simplified the search for square-lattice GKP Cliffords. Ideally, a similar expression 
for incomplete four-splitters would help expedite the search for other gates; at the moment, none is known.

\section*{Acknowledgements}
We thank Nicholas Funai and Yui Kuramochi for useful discussion. This work is supported by the Australian Research Council Centre of Excellence for Quantum Computation and Communication Technology (Project No. CE170100012). B.Q.B is additionally supported by the Japan Science and Technology Agency through the Ministry of Education, Culture, Sports, Science, and Technology Quantum Leap Flagship Program (MEXT Q-LEAP). T.M.~is supported by JSPS Overseas Research Fellowships.
\appendix
\numberwithin{equation}{section}
\renewcommand{\theequation}{\thesection\arabic{equation}}

\section{Derivation of the single-mode teleported gate $\op V(\theta,\theta')$ }\label{shearVderivation}

We present an alternate derivation of $\op V(\theta,\theta')$ gate, simpler than that found in Ref.~\cite{Walshe2020}.
We will make use of the LDU decomposition of a rotation matrix to decompose the phase-delay operator into shear and squeeze operators,
    \begin{align}\label{rotLDU}
        \op R(\theta)&=\op P^\dag_p(\tan{\theta})\op S(\sec{\theta})\op P(\tan{\theta}) 
\\*
            &=\op P(\tan{\theta})\op S(\cos{\theta})\op P_p^\dagger(\tan{\theta})
    \end{align}
where the position- and momentum-shear gates are given in Eq.~\eqref{eq:posshear}, 
and we use the squeezing operator $\op{S}(\zeta)$ defined in Refs.~\cite{Walshe2020, walshe2021streamlined}, for which $\zeta>0$ squeezes momentum and $\zeta<0$ squeezes position. Using Eq.~\eqref{rotLDU} in Eq.~\eqref{rotatedmeasurement} gives a useful identity for the rotated homodyne measurement outcome,%
\footnote{The factor at the front maintains the norm of the state: $\abs{a} \pinprod {a m} {a m'} = \abs{a} \delta(am - am') = \delta(m-m')$.}
\begin{align}\label{RotToShear}
    \brasub{m}{p_{\theta}}=\pbra{m} \op R(\theta) = \sqrt{\sec \theta} \pbra{m \sec{\theta}} \op P(\tan{\theta})
    \; .
\end{align}	
The fact that the outcome inside the bra is modified means that when using this identity, the modified measurement outcomes that should be used when considering measurement statistics or potential feed-forward operations based on the outcome (as in CV teleportation). Note that these involve classical post-processing of the actual measurement outcome $m$ obtained in the lab.%

Begin with a Bell-type measurement in a CV teleportation circuit with measurement angles $\theta_1$ and $\theta_2$. Using Eq.~\eqref{rotatedmeasurement}, extract phase-delay operators such that the measurement bases are position and momentum, respectively,
	\begin{equation}\label{eq:Belltypemeasurement}
    \resizebox{0.6\columnwidth}{!}{
         \Qcircuit @C=1em @R=2.5em @! 
         {
         	\lstick{\brasub{m_1}{p_{\theta_1}}}	& \bsbal{1} & \qw \\
         \lstick{\brasub{m_2}{p_{\theta_2}}}	& \qw       & \qw
  		  }
\quad  \raisebox{-1.3em}{$=$} \quad \quad \qquad
    \Qcircuit @C=1em @R=1em {
         	\lstick{\brasub{m_1}{q}}	& \gate{R(\theta_1-\tfrac{\pi}{2})}&\bsbal{1} & \qw \\
         \lstick{\brasub{m_2}{p}}	&\gate{R(\theta_2)}& \qw       & \qw
  		  }
    }
    \raisebox{-1em}{\quad .}
	\end{equation}
Decompose the phase delay on the second mode into two, and then commute the common rotation on both modes through the beam splitter,
	\begin{equation}\label{}
    \Qcircuit @C=1em @R=1em {
         	\lstick{\brasub{m_1}{q}}	& \qw&\bsbal{1} & \gate{R(\theta_1-\tfrac{\pi}{2})}&\qw \\
         \lstick{\brasub{m_2}{p}}	&\gate{R(\tfrac{\pi}{2} - \theta_-)}& \qw       & \gate{R(\theta_1-\tfrac{\pi}{2})}&\qw
  		  }
    \raisebox{-1em}{\quad .}
	\end{equation}
where $\theta_- \coloneqq \theta_1 - \theta_2$.
Now, use Eq.~\eqref{RotToShear} to rewrite the rotated-basis measurement on the second mode as a position shear by $\tan (\tfrac{\pi}{2} - \theta_-) = \cot \theta_-$ with modified measurement outcome $m_2' = m_2 \csc \theta_-$.

Then, decompose the beam splitter using the circuit identity in Eq.~\eqref{eq:bsDecomp},
	\begin{equation}\label{}
    \resizebox{0.8\columnwidth}{!}{
    \Qcircuit @C=1em @R=1em {
         &\lstick{\brasub{m_1}{q}}&\qw & \ctrl{1} & \gate{S^\dag({\sqrt{2}})} & \targ & \gate{R(\theta_1-\tfrac{\pi}{2})}&\qw
         \\
         &\lstick{\brasub{m_2'}{p}}& \gate{P (\cot \theta_-)}	& \targ & \gate{S(\sqrt{2})} & \ctrlo{-1} & \gate{R(\theta_1-\tfrac{\pi}{2})}&\qw
  		  }
    }
    \raisebox{-1em}{\quad ,}
	\end{equation}
The left-most controlled-shift gate, $\CX^{12}$, takes the $\op q_1$ eigenvalue $m_1$ to become a shift on mode 2, $\op{X}(m_1)$,
\eqref{RotToShear},
	\begin{equation}\label{}
    \resizebox{0.8\columnwidth}{!}{
    \Qcircuit @C=1em @R=1em {
         &\lstick{\brasub{m_1}{q}} &\qw                      & \qw & \gate{S^\dag({\sqrt{2}})} & \targ & \gate{R(\theta_1-\tfrac{\pi}{2})}&\qw
         \\
         &\lstick{\brasub{m_2'}{p}}& \gate{P(\cot \theta_-) } & \gate{X(m_1)} & \gate{S(\sqrt{2})} & \ctrlo{-1} & \gate{R(\theta_1-\tfrac{\pi}{2})}&\qw
  		  }
    }
    \raisebox{-1em}{\quad ,}
	\end{equation}
Absorbing the resulting shift and squeeze into the outcome gives 
$m_2''=\sqrt{2}(m_2 \csc \theta_- -m_1 \cot \theta_-)$.
The remaining position shear commutes with the controlled gate $\CX^{21}(-1)$, yielding the circuit
	\begin{equation}\label{}
        \resizebox{0.6\columnwidth}{!}{
    \Qcircuit @C=1em @R=1em {
         	&\lstick{\brasub{m'_1}{q}}&\qw & \targ & \qw & \gate{R(\theta_1-\tfrac{\pi}{2})}&\qw\\
         &\lstick{\brasub{m_2''}{p}}	& \qw & \ctrlo{-1} & \gate{P(2\cot \theta_-
)} & \gate{R(\theta_1-\tfrac{\pi}{2})}&\qw
  		  }
         } \raisebox{-1em}{\quad ,}
	\end{equation}

The controlled gate acting on the two quadrature eigenstates is equivalent to (the dual of) a displaced position-correlated EPR pair~\cite{Walshe2020}, $\brasub{s}{q_1}\otimes \brasub{t}{p_2}
e^{i\op p_1 \op q_2} =
\bra{\EPR}
\op D_1^\dag( \tfrac{s-it}{\sqrt{2}})$. With this, the circuit becomes
\begin{equation}\label{}
    \Qcircuit @C=1em @R=1em {
         	&\EPRdl&\gate{D^\dag(\mu_{1,2}')}  & \gate{R(\theta_1-\tfrac{\pi}{2})}&\qw
         	\\
         &\EPRul  & \gate{P(2\cot \theta_-)} & \gate{R(\theta_1-\tfrac{\pi}{2})}&\qw
  		  }
    \raisebox{-1em}{\quad ,}
	\end{equation}
where $\mu_{1,2}'=m_1+i(m_2 \csc \theta_- -m_1 \cot \theta_-)$.
Finally, bounce each operator on mode two to mode one through the EPR state~\cite{Walshe2020} and push the displacement to the left to get the final form of the circuit
	\begin{equation}\label{eq:singlemodeVderive}
        \hspace{1cm}
         \Qcircuit @C=1em @R=2.5em 
         {
         \lstick{\brasub{m_1}{p_{\theta_1}}}	& \bsbal{1} &\qw   \\
         \lstick{\brasub{m_2}{p_{\theta_2}}}	& \qw      & \qw  
  		  }
\quad  \raisebox{-1.3em}{=} \quad 
    \Qcircuit @C=1em @R=1.8em {
         	&\EPRdl &\gate{D(\mu_{1,2})}& \gate{V(\theta_1,\theta_2)}& \qw \\
                &\EPRul & \qw               & \qw                        & \qw
  		  }
    \raisebox{-1em}{\; ,}
	\end{equation}
where $\op V(\theta_1,\theta_2)=\op R(\theta_1-\tfrac{\pi}{2}) \op P(2\cot \theta_-) \op R(\theta_1-\tfrac{\pi}{2})$, 
	and the displacement amplitude is
	\begin{align} \label{eq:amplitude}
	\mu_{1,2} &=-\frac{m_1 e^{i\theta_2}+m_2 e^{i\theta_1}}{\sin(\theta_1-\theta_2)}
	\, .
	\end{align}
For $\theta_- = n \pi$ for integer n, these expressions are not defined. This corresponds to the case where the two modes in Eq.~\eqref{eq:Belltypemeasurement} are measured the same basis (up to parity). This effectively removes the beam splitter, uncoupling the modes and precluding teleportation.

\section{Displacements for two-mode gates}\label{Appendix:Displacements} 

As presented in the main text, the teleported gate in Eq.~\eqref{eq:2modeKraus} lacks a two-mode displacement operator originating from the four measurement outcomes in the teleportation gadget. We give its form here since knowing the displacement is necessary in practice to perform corrections during the teleportation procedure.

Up to a local parity operator but including the single-mode outcome-dependent displacements, the two-mode gate realized by the completed teleportation gadgets in Sec.~\ref{sec:clusterStates} take the form
\begin{align}\label{eq:gen2modeV}
    \op{V}^{(2)}(\vec{\theta},\vec{m}) = \op B_{1,2}(\pm \tfrac{\pi}{4}) \begin{bmatrix}
        \op D_1(\mu_{1,2}) \op V_1(\theta_1,\theta_2) \\ \otimes \\ \op D_2(\mu_{3,4}) \op V_2(\theta_3,\theta_4)
    \end{bmatrix} \op B_{1,2}(\tilde \pm \tfrac{\pi}{4}),
\end{align}
where $\vec{m}$ stands in for the four outcomes, and we use vertical tensor product notation to shorten the expression.
The $\pm$ and $\tilde{\pm}$ are independent signs (indicating the different possible directions of each 50:50 beam splitter%
),  
and each single-mode displacement amplitude is given by Eq.~\eqref{eq:amplitude}.
Pushing the displacement operators to the left through the beam splitter gives the more general form of Eq.~\eqref{eq:2modeKraus},
\begin{align}\label{eq:2modeKrausDisplacements}
    \resizebox*{.85\columnwidth}{!}{
        \Qcircuit @C=1.4em @R=2em {
            & \lstick{\text{(out)}} & \gate{N(\beta)\Pi N(\beta)} & \gate{D(\mat m_1)} & \multigate{1}{V^{(2)}(\bm{\theta})} & \qw & \nogate{\text{(in)}}&\\
            & \lstick{\text{(out)}} & \gate{N(\beta)\Pi N(\beta)} & \gate{D(\mat m_2)} & \ghost{V^{(2)}(\bm{\theta})} & \qw & \nogate{\text{(in)}}&
        }.}
\end{align}
The amplitudes are mixed by the beam splitter to give
\begin{align}
    \vec{m}_1 &= \tfrac{1}{\sqrt{2}}(\mu_{1,2}\pm \mu_{3,4}) \\
    \vec{m}_2 &= \tfrac{1}{\sqrt{2}}(\mu_{3,4}\mp \mu_{1,2}) .
\end{align}
with the sign $\pm$ determined by the sign in
the leftmost beam splitter.

\section{Circuit identity for beam splitters over three modes}\label{appendix:generalEulerAngles}

Consider beam splitters that contain a transmission parameter $\theta$,
\begin{equation}
\begin{split} \label{eq:genbeamsplitter}
	\raisebox{-1.2em}{$\op{B}_{jk}(\theta)  \coloneqq
 e^{ - \theta ( \op{a}_j   \op{a}^\dagger_k - \op{a}^\dagger_j   \op{a}_k )} \, = \, $~~}
         \Qcircuit @C=1.25em @R=2.5em @! 
         {
         	& \varbs{1} & \rstick{j}  \qw \\
         	& \qw       & \rstick{k} \qw
  		  } 
\end{split}		
\end{equation}
with $\theta \in \{ -\pi, \pi \}$ described by orthogonal matrices
\begin{align} \label{bsmatrot}
        \BSmat_{jk}(\theta) = 
        \begin{bmatrix}
        \cos \theta & -\sin \theta  \\
        \sin \theta & \cos \theta
    \end{bmatrix},
    \end{align} 
which are identical in form to two-dimensional rotation matrices. 

When a third mode is included, labeled by $\ell$, the matrix in Eq.~\eqref{bsmatrot} expands into the three-dimensional matrix 
    \begin{align}
        \BSmat_{jk}(\theta) =
        \begin{bmatrix}
        \cos \theta & -\sin \theta & 0 \\
        \sin \theta &  \cos \theta & 0 \\
        0           & 0            & 1
        \end{bmatrix} 
        ,
    \end{align}
which can be interpreted 
as
a rotation in three dimensions around principle $z$ axis by angle $\theta$, $\BSmat_{jk}(\theta) = \BSmat_z(\theta)$. Beam splitters between other pairs of modes describe rotations around the other principle axes: $\BSmat_{k\ell}(\theta) = \BSmat_x(\theta)$ for $x$ rotations, and $\BSmat_{\ell j}(\theta) = \BSmat_y(\theta)$ for $y$ rotations. Note that swapping the subscripts on any matrix takes $\theta \to -\theta$, which reverses the direction of the arrow in the circuit diagram on the right-hand side of Eq.~\eqref{eq:genbeamsplitter}.

This association leads to several important consequences. Any beam-splitter network over three modes containing any number of real beam splitters can be described by a three-dimensional rotation matrix. Decomposing this matrix using Euler angles or Tait-Bryan angles gives its description in terms of three principle-axis rotation matrices, each corresponding to a single beam splitter. Thus, any three-mode beam-splitter network (with real beam splitters) can be reduced to three beam splitters or fewer.

As an example, consider reordering a beam-splitter network of two beam splitters that share one mode, $\op{B}_{k \ell}(\theta_2)\op{B}_{jk}(\theta_1)$. Suppose we seek the specific Tait-Bryan decomposition into fundamental-axis rotations in the order $z-y-x$ (right-to-left). The associated rotation matrices satisfy\footnote{Note that the corresponding rotation is already in an Euler form, $\mat{R}_z(0) \mat{R}_x(\theta_2) \mat{R}_z(\theta_1)$ and a different Tait-Bryan form, $\mat{R}_y(0) \mat{R}_x(\theta_2) \mat{R}_z(\theta_1)$.}
    \begin{align}
        \mat{R}_x(\theta_2) \mat{R}_z(\theta_1) = \mat{R}_z(\gamma) \mat{R}_y(\beta) \mat{R}_x(\alpha)
    \end{align}
corresponding to the circuit equivalence,
\begin{align}
\begin{split} \label{eq:genBScommute1}
    \Qcircuit @C = 0.75cm @R = 0.75cm {
       & \qw & \varbs{1}[\theta_1] & \qw \\
       & \varbs{1}[\theta_2] & \qw & \qw \\
        & \qw & \qw & \qw
    }
    \raisebox{-.75cm}{\;=\;}
    \Qcircuit @C = 0.75cm @R = 0.75cm {
         & \varbs{1}[\gamma]& \qw & \qw & \qw  \\
         & \qw & \qw & \varbs{1}[\alpha] & \qw  \\
          & \qw & \varbs{-2}[\raisebox{1em}{\scriptsize{$\beta$}}] & \qw & \qw 
     }
     \raisebox{-.75cm}{\;,\;}   
\end{split}
\end{align}
where $\alpha$, $\beta$, and $\gamma$ are functions of $\theta_1$ and $\theta_2$. 
For the case of two balanced beam splitters,
\begin{align}
    \Qcircuit @C = 0.75cm @R = 0.75cm {
       & \qw & \bsbal{1} & \qw \\
       & \bsbal{1} & \qw & \qw \\
        & \qw & \qw & \qw
    }
    \raisebox{-.75cm}{\;=\;}
    \Qcircuit @C = 0.75cm @R = 0.75cm {
         & \varbs{1}[\theta]& \qw & \qw & \qw  \\
         & \qw & \qw & \varbs{1}[\theta] & \qw  \\
          & \qw & \varbs{-2}[\raisebox{1em}{\scriptsize{$\theta'$}}] & \qw & \qw 
     }
      \raisebox{-.75cm}{\;},    
\end{align}
where $\theta = \operatorname{arctan}{\frac{1}{\sqrt{2}}}$ and $\theta' = \operatorname{arctan}{-\frac{\sqrt{3}}{3}}$.

\section{Type (b) incomplete four-splitters cannot be completed by restricting measurements} \label{threesplitterinequivalence}

Here, we show that type (b) incomplete four-splitters, 
which lack a fourth beam splitter on the state side of the beam-splitter network, cannot be completed by choosing specific measurement angles and post-processing of the measurement outcomes. 
Consider a beam-splitter network of this type, $\op{B}^{(3)}$, that, when completed, gives a four-splitter $\op{B}_\text{4-split}$. $\op{B}^{(3)}$ is equivalent to the desired four-splitter $\op{B}_\text{4-split}$ with a residual operator on the measurement side,
    \begin{align}
        \op{B}^{(3)} 
        = \op{B}_\text{res} \op{B}_\text{4-split}
        \label{eq:measurement_side_compensation}
    \end{align}
where the residual operator $\op{B}_\text{res} \coloneqq \op{B}^{(3)} \op B^\dagger_\text{4-split}$ is itself another network of real beam splitters. %
$\op{B}_\text{res} \to \mat R_\text{res}$ has an associated orthogonal matrix $\mat{R}_\text{res}$ composed from the individual beam-splitter matrices in the network.

The question is whether there exist measurement angles $\vec{\theta} \coloneqq \{ \theta_1, \theta_2, \theta_3, \theta_4  \}$ and some linear transformation of the measurement outcomes that, together, virtually delete the residual beam-splitter network $\op{B}_\text{res}$---i.e.,
    \begin{align} \label{eq:undoingsomething}
        \brasub{\vec{m}}{p_{\vec{\theta}}}\op{B}_\text{res} =  
        \sqrt{\abss{\det \mat{L}^{-1}} }\,
        \brasub{\mat{L}^{-1}_\text{res} \vec{m}}{p_{\vec{\theta}'}} 
        \, .
    \end{align}
We consider a Heisenberg-picture linear transformation described by $\mat{L}_\text{res}$ more general than that of the orthogonal matrix $\mat{R}_\text{res}$ in order to allow for the possibility that post-processing beyond orthogonal transformation may be required.
For example, scaling the outcome on a single mode can insert a squeezing operator: $\brasub{m}{q} = \sqrt{\zeta} \, \brasub{\zeta m}{q} \op{S}(\zeta) $~\cite{Walshe2020}. Note that doing so introduces a change in measure, which is the origin of the determinant in Eq.~\eqref{eq:undoingsomething}.
Additionally, $\mat{L}_\text{res}$ can only linearly transform outcomes in the chosen measurement bases---we further include the possibility that deletion also requires a new set of angles $\vec{\theta}'$.

Theorem~\ref{noChoice}, proved below, shows that $\op{B}_\text{res}$ cannot be deleted unless all the modes are measured in the same basis, $\theta_j = \theta_k$ for all $j,k$, which trivially disentangles them. For Theorem~\ref{noChoice} to apply here, a further property of $\mat{R}_\text{res}$ is required: it must be neither block diagonal nor row permutations away from block diagonal. 
This can be shown using its decomposition into constituent matrices, Eq.~\eqref{eq:foursplitterdecomp}. Without loss of generality, we choose an incomplete four-splitter missing a beam splitter between modes $j$ and $k$ in Eq.~\eqref{eq:foursplitterdecomp}, i.e.,~$\mat{R}^{(3)}=\mat R_{km} \mat R_{jl} \mat R_{lm}$.%
\footnote{For clarity of presentation, we assume no transposes on the matrices in Eq.~\eqref{eq:foursplitterdecomp}. The result still holds with them restored, however; the notation just becomes cluttered.}
Then, we have
    \begin{align} \label{residualmatrix}
        \mat{R}_\text{res} 
        & = \mat{R}^{(3)} \mat{R}^\tp_\text{4-split} 
        = \FSmat_{km}\FSmat_{j \ell} \FSmat^\tp_{jk} 
       \FSmat^\tp_{j \ell} \FSmat^\tp_{km} .
    \end{align} 
Every type (b) residual beam-splitter network has this form---a pair of commuting beam splitters on the left, a pair of commuting beam splitters on the right, and a fifth beam splitter that commutes with neither of the pairs.

We observe that the matrix in Eq.~\eqref{residualmatrix} has no zero entries.  To show this, let us consider the following permutation $\{m, k, j, \ell \} \rightarrow \{1, 2, 3, 4\} $
and define the matrix $\mat Q$ that realizes this permutation for any $\{m, k, j, \ell \}$.
Then, we have%
\footnote{With the transposes restored in Eq.~\eqref{eq:foursplitterdecomp}, the signs may change in Eq.~\eqref{eq:no_zero_entries}, but the matrix still has no zeroes.}
\begin{align}
    \mat Q & \mat{R}_\text{res}  \mat Q^{-1} = \mat R_{12} \mat R_{34} \mat R_{23} \mat R_{43} \mat R_{21} \\
    &= \frac{1}{2\sqrt{2}}\begin{pmatrix} 1 + \sqrt{2} & 1 -\sqrt{2} & -1 & -1 \\  1-\sqrt{2} &   1 + \sqrt{2} & -1 & -1 \\ 1 & 1 & 1 + \sqrt{2} & 1 -\sqrt{2} \\ 1 & 1 & 1 -\sqrt{2} & 1 +\sqrt{2} \end{pmatrix}, \label{eq:no_zero_entries}
\end{align}
which has no zero entries.%
\footnote{Note that this is the residual matrix associated with the MBSL in Sec.~\ref{sec:MBSL}.}
Since permutation does not change the values of entries, $\mat{R}_\text{res} $ itself has no zero entries.
The next theorem essentially states that the beam-splitter network that corresponds to such a matrix cannot be cancelled 
as in
Eq.~\eqref{eq:undoingsomething} for any nontrivial choices of the measurement angles $\vec{\theta}$. (That is, the only choice that works is all angles equal, which virtually deletes the entire beam-splitter network.)

\begin{theorem}\label{noChoice}
    Let $\vec{m}$ denote a vector of $N$ homodyne measurement outcomes and $\vec{\theta}$ denote a vector of $N$ choices of homodyne-measurement angles with $\theta_j \in(-\frac{\pi}{2},\frac{\pi}{2}]$ for all $j\in\{1,\ldots,N\}$.  Let $\op{B}^{(N)} \to \mat 
    R$
    be a real beam-splitter network%
    , where $\mat R \in \mathrm{O}(N)$ is an $N$-dimensional real orthogonal matrix that is neither block diagonal nor row permutations away from block diagonal. Let $\vec{\theta}'$ be another vector of $N$ choices of homodyne-measurement angles, with $\theta'_i\in(-\frac{\pi}{2},\frac{\pi}{2}]$ for all $i\in\{1,\ldots,N\}$. Finally, let $\mat L \in \reals^{N\times N}$ be an 
    $N$-dimensional real 
    matrix such that there exists a corresponding Gaussian unitary $\op{U}_G$ that satisfies the following:
    \footnote{The determinant is required because of the following unitarity: $ \op{I} = \int d\vec{m} \,  \op{U}^\dagger_G(\mat{L}) \ketsub{\vec{m}}{p}\brasub{\vec{m}}{p} \op{U}_G(\mat{L}) = \int d\vec{m} \, {\abss{\det \mat{L}^{-1}} }\, \ketsub{\mat{L}^{-1} \vec{m}}{p}\brasub{\mat{L}^{-1} \vec{m}}{p} $.}
    \begin{equation}
        \sqrt{\abss{\det \mat{L}^{-1}} }\, \brasub{\mat{L}^{-1} \vec{m}}{p} = \brasub{\vec{m}}{p} \op{U}_G. \label{eq:unitary_corresponding_to_L}
    \end{equation}
    Then, given $\vec{\theta}$ and $\op{B}^{(N)}$, there does not exist a pair $(\vec{\theta}',\mat L)$ that satisfies
    \begin{equation}
        \forall \vec{m}\in\mathbb{R}^N, \qquad \brasub{\vec{m}}{p_{\vec{\theta}}}\op{B}^{(N)} = \sqrt{\abss{\det \mat{L}^{-1}} }\, \brasub{ \mat{L}^{-1} \vec{m}}{p_{\vec{\theta}'}}, \label{eq:rot_BS=Gauss_rot}
    \end{equation}
    unless $\vec{\theta}=\theta\vec{1}_N$, where $\vec{1}_N=(1, 1, \ldots, 1)^{\tp}$.
\end{theorem}
\begin{proof}
    Let $\op{R}^{(N)}(\vec{\theta})$ be the $N$-mode local phase rotation with the rotation angles given by $\vec{\theta}$.  Then, the condition Eq.~\eqref{eq:rot_BS=Gauss_rot} is equivalent to
    \begin{equation}
        \forall \vec{m}\in\mathbb{R}^N, \quad \brasub{\vec{m}}{p}\op{R}^{(N)}(\vec{\theta}) \op{B}^{(N)} = \brasub{\vec{m}}{p} \op{U}_G \op{R}^{(N)}(\vec{\theta}').
    \end{equation}
    Since this holds for all $\vec{m}\in\mathbb{R}^N$, we have
\begin{equation}
    \op{R}^{(N)}(\vec{\theta}) \op{B}^{(N)} = \op{U}_G(\mat{L}) \op{R}^{(N)}(\vec{\theta}').\label{eq:rot_BS=Gauss_rot_operator_level}
\end{equation}
Now the left-hand side, as well as $\op{R}^{(N)}(\vec{\theta}')$, is a linear-optical unitary and thus does not contain squeezing.  Therefore, in order for the above to be equal, $\op{U}_G$ must also be a linear-optical unitary: $\op{U}_G \to \mat U \in \mathrm{U}(N)$.
From Eq.~\eqref{eq:unitary_corresponding_to_L}, $\op{U}_G$ does not mix $\op{p}$s and $\op{q}$s, and~$\mat L$ is the (real) matrix by which the $\op p$s evolve. This restricts~$\mat U$ such that 
$\mat U = \mat L \in \mathrm{O}(N)$,
i.e.,~it is real orthogonal.  This means that $\op{U}_G$ is a real beam-splitter network denoted by $\op{B}'^{(N)} \to \mat R' \in \mathrm{O}(N)$, with $\mat R' = \mat L$.

Eq.~\eqref{eq:rot_BS=Gauss_rot_operator_level} can now be rewritten as
\begin{equation}
    \op{R}^{(N)}(\vec{\theta}) \op{B}^{(N)} = \op{B}'^{(N)} \op{R}^{(N)}(\vec{\theta}').\label{eq:rot_BS=BS_rot}
\end{equation}
Since they are linear-optical unitaries, we consider their representations as unitary matrices acting on the vector of annihilation operators.  It is known that for a beam-splitter network and a local phase rotations, $\op{B}^{(N)} \to \mat R$ and $\op{R}^{(N)}(\vec{\theta}) \to \diag(e^{i\vec{\theta}})$, where $\diag(e^{i\vec{\theta}})$ is a diagonal matrix of phases.
With these in mind, one can see that Eq.~\eqref{eq:rot_BS=BS_rot} is equivalent to
\begin{equation}
    \diag(e^{i\vec{\theta}}) \mat R = \mat R' \diag(e^{i\vec{\theta}'}). \label{eq:diag_ortho=ortho_diag}
\end{equation}
Taking the transpose of the both sides leads to
\begin{equation}
    \mat R^{\tp} \diag(e^{i\vec{\theta}})  = \diag(e^{i\vec{\theta}'}) \mat R'^{\tp}. \label{eq:transpose_D11}
\end{equation}
Combining Eqs.~\eqref{eq:transpose_D11} and~\eqref{eq:diag_ortho=ortho_diag} gives
\begin{equation}
    \mat R^{\tp} \diag(e^{2i\vec{\theta}}) \mat R = \diag(e^{2i\vec{\theta}'}).
\end{equation}
The above implies that the left-hand side is the diagonalization of the right-hand side, which is itself diagonal.  Since the 
spectrum
of a normal matrix is unique up to a permutation of the eigenvalues, we obtain
\begin{equation}
    (\mat P \mat R)^{\tp} \diag(e^{2i\vec{\theta}'})  (\mat P \mat R) = \diag(e^{2i\vec{\theta}'}), \label{eq:commutation_diagonal}
\end{equation}
where $\mat P$ is a permutation matrix that satisfies $\mat P \diag(e^{2i\vec{\theta}}) \mat P^{\tp} = \diag(e^{2i\vec{\theta}'})$.  Equation~\eqref{eq:commutation_diagonal} means that $\mat P \mat R$ and $\diag(e^{2i\vec{\theta}'})$ commute.  However, it is known that any matrix that commutes with a diagonal matrix has a block-diagonal form that acts nontrivially only within the eigenspaces of degenerate eigenvalues.  Given that $2\theta'_j \in(-\pi,\pi]$ for all $j\in\{1,\ldots,N\}$, Eq.~\eqref{eq:commutation_diagonal} contradicts 
the initial assumption that $\mat R$ is neither block diagonal nor a 
row permutation away from block diagonal unless all the elements of $\vec{\theta}'$ are the same.
In this exceptional case, we are restricted to choosing%
\footnote{Technically, they can differ by integer multiples of~$\pi$, which is what the 2 in~$\diag(e^{2i\vec{\theta}'})$ allows. Since these are physically identical (just negate the outcome of the measurement), this freedom of definition provides no benefit, and we do not dwell on it.}
${\vec{\theta}=\vec{\theta}'=\theta \vec{1}_N}$by Eq.~\eqref{eq:commutation_diagonal}, which completes the proof.
\end{proof}

Note that the restriction $\theta_j\in[-\frac{\pi}{2},\frac{\pi}{2})$ for any $j\in\{1,\ldots,N\}$ does not imply the loss of generality because other measurement angles can be simulated by putting negative signs on the measurement outcomes with appropriate measurement angles in this range.
As we see in Eq.~\eqref{eq:no_zero_entries} that the matrix $\mat{R}_\text{res}$ has no zero entries and thus is neither block diagonal nor row permutation away from block diagonal.
This fact combined with the above theorem means that, putting aside the trivial case $\vec{\theta}=\theta \vec{1}_N$, it is impossible to absorb the residual beam splitter $\op{B}_\text{res}$ by the change of the measurement angles as well as the linear post-processing of the measurement outcomes. 
Combining this with Eq.~\eqref{eq:measurement_side_compensation} implies that this (equivalence class of) incomplete four-splitter~$\op{B}^{(3)}$ [\emph{i.e.},~type~(b)] cannot be made equivalent to a four-splitter (followed by homodyne measurement and reinterpretation of its outcomes) with nontrivial choices of measurement angles. %

\section{GKP Bell pair insertion}\label{sec:GKPinsertion}

The noise analysis in Ref.~\cite{walshe2021streamlined} for the QRL assumes that the CV cluster state is entirely constructed from GKP qunaught states. In practice, a deterministic source of GKP qunaughts may not be available, and the cluster state may have locations where squeezed momentum states were swapped in to fill in the gaps~\cite{menicucci2014fault,walshe2021streamlined,Bourassa2021blueprint, Tzitrin2021passive}.
Here, we give a procedure to insert a GKP Bell pair---or any other state made from two single-mode states coupled on a beam splitter---%
into a macronode wire (which may be a part of a larger cluster state) after it has been created.

\begin{figure}
    \centering
    \includegraphics[width=0.8\columnwidth]{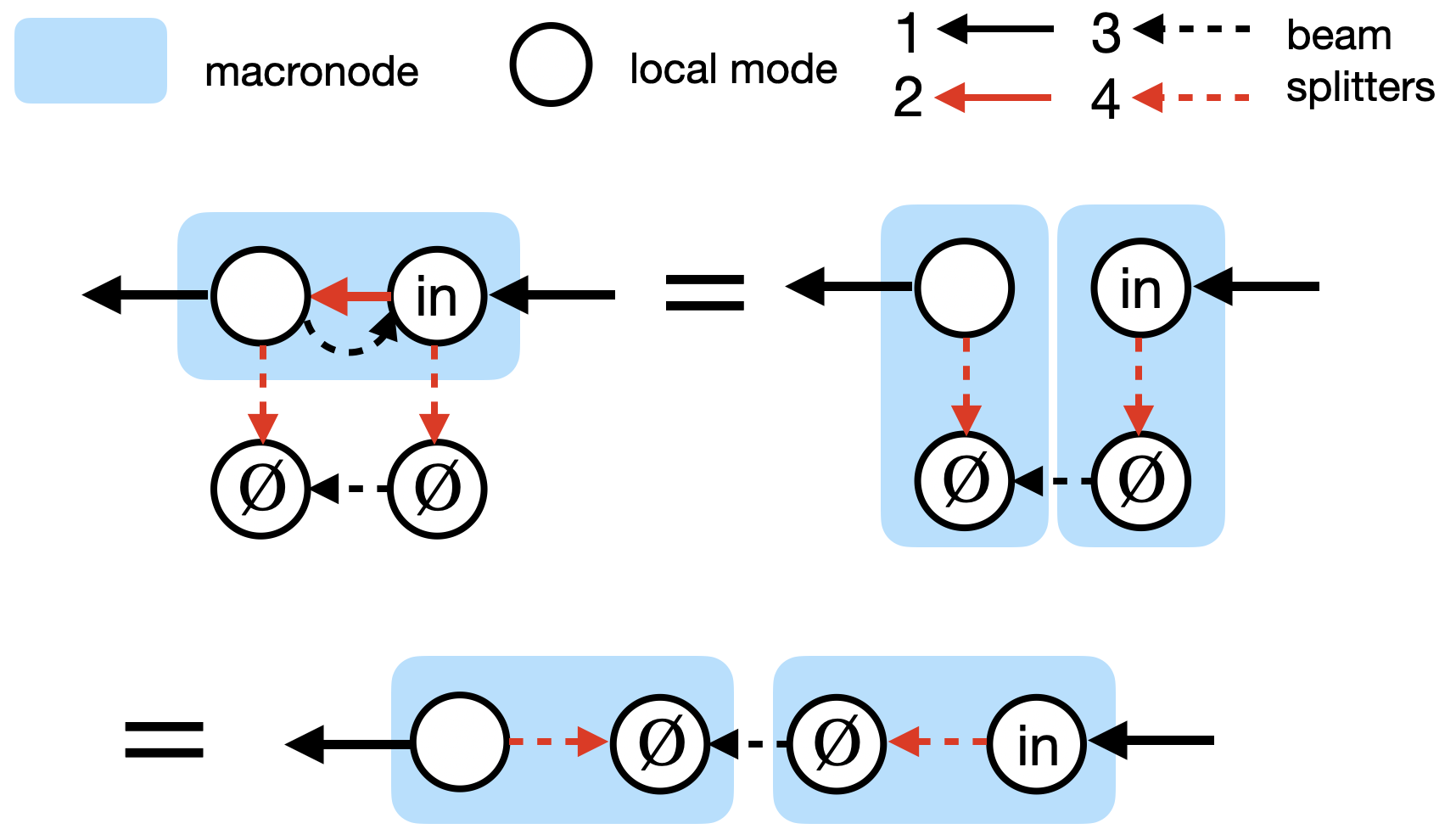} 
    \caption{\label{Fig-BellPairInsertion}
        GKP Bell pair insertion using a balanced four-splitter. Empty circles are local modes in unspecified states, circles filled with $\qunaught$ are local modes prepared in pure GKP qunaught states (noisy versions also work). 
        This process inserts two $\qunaught$ states in-line in a single-mode macronode wire (either isolated or embedded in a two-dimensional cluster state). 
        Described from left to right, two inputs, $\ket{\qunaught}$ are introduced alongside an existing macronode wire and a canonical balanced four-splitter is used to interact the two states with a two-mode macronode. The two beam splitters contained within the macronode cancel out. We then regroup the modes into two new macronodes, each one containing a single $\qunaught$ state. Each of these macronodes, when teleported through, performs error correction on either the position or momentum quadratures~\cite{Walshe2020}. Together, full GKP error correction is possible between these two macronodes.
         Note that one is not required to use qunaught states. This circuit is, more generally, a way to insert a Kraus state $\op{B}_{12} \ket{\psi} \otimes \ket{\phi}$ into a macronode wire.}
    \label{bellPairInsertionGraph}
\end{figure}

By undoing one of the wire's original beam splitters, an additional pair of modes containing a GKP Bell pair can be inserted into a macronode wire; see 
Fig.~\ref{Fig-BellPairInsertion}. 
This procedure is implemented by the circuit%
\footnote{This circuit starts from the red beam splitters in Fig.~\ref{Fig-BellPairInsertion}. The black ones are omitted because they link the top and bottom wires to other wires not shown.}
\begin{equation}\label{GKPMeasurementCircuit}
\begin{split}
\centering
\Qcircuit @C=1em @!R=0.1em {
		& \lstick{\brasub{m_1}{p_{\theta_1}}}& \bsbal[red]{1}[.>]&  \qw&\qw &  \qw &\qw & \qw&\qw&\qw& \bsbal[red]{3} & \qw &  \nogate{1}\\
		&\lstick{\brasub{m_a}{p_{\theta_a}}}& \qw&\bsbal{1}[.>]&\qw& \qw & \qw& \rstick{\hspace{-1.2em}\ket{\qunaught}} &&&&& \nogate{a} \\
		&\lstick{\brasub{m_b}{p_{\theta_b}}}&\bsbal[red]{1}[.>]&\qw& \qw& \qw   &\qw& \rstick{\hspace{-1.2em}\ket{\qunaught}} &&&&& \nogate{b} \\
		& \lstick{\brasub{m_2}{p_{\theta_2}}}&\qw& \qw  &\qw & \bsbal{-3}[.>] &  \qw&\qw &\qw &\qw&\qw& \qw & \nogate{2}  \\
}
\, \raisebox{-3.5em}{.}
\end{split}
\end{equation}
where modes $1$ and $2$, originally coupled by the red beam splitter on the right-hand side, are those from the macronode wire. Modes $a$ and $b$ are the new modes containing GKP Bell qunaught states. 
The dashed beam splitter from the mode $2$ to $1$ eliminates the red one, leaving
\begin{equation}\label{GKPMeasurementCircuitCopy2}
\begin{split}
\centering
\Qcircuit @C=1em @!R=0.1em {
		& \lstick{\brasub{m_1}{p_{\theta_1}}}&  \qw&\bsbal[red]{1}[.>] &  \qw &\qw & \qw&\qw&\qw& \qw & \qw  & \nogate{1} \\
		&\lstick{\brasub{m_a}{p_{\theta_a}}}&\qw&\qw& \bsbal{1}[.>] & \qw& \rstick{\hspace{-1.2em}\ket{\qunaught}} &&&&&  \nogate{a} \\
		&\lstick{\brasub{m_b}{p_{\theta_b}}}&\qw& \bsbal[red]{1}[.>]& \qw   &\qw& \rstick{\hspace{-1.2em}\ket{\qunaught}} & &&&&  \nogate{b}\\
		& \lstick{\brasub{m_2}{p_{\theta_2}}}& \qw  &\qw & \qw &  \qw&\qw &\qw &\qw&\qw& \qw &  \nogate{2} \\
}
 \raisebox{-3.5em}{\; .}
\end{split}
\end{equation}
This circuit is simplified with application of Eq.~\eqref{eq:singlemodeVderive}~\cite{Walshe2020}
to give
\begin{equation}\label{GKPMeasurementCircuitCopy3}
\begin{split}
\centering
\Qcircuit @C=1em @!R=0.1em {
		&\EPRdl &  \gate{D(\mu)}& \gate{V(\theta_1,\theta_a)}& \qw&\\
		 &\EPRul & \gate{\Pi_\GKP} &  \qw &\EPRdr & \\
		&\EPRdl&\gate{D(\mu')} & \gate{V(\theta_b,\theta_2)} & \EPRur&\\
		     &\EPRul&\qw&\qw&\qw \\
}
\, \raisebox{-3.5em}{,}
\end{split}
\end{equation}
and $\mu$ $(\mu^{\prime})$ is a function of $m_1$ and $m_a$ ($m_b$ and $m_2$).
This adds two further measurement degrees of freedom, enabling an additional $\op V$ gate, but only offers GKP error correction after the first one. Inserting GKP Bell pairs throughout the macronode wires is effectively the same as including GKP states at the beginning of cluster state generation on half of the ancillas.

\section{CV circuit identities}\label{appendix:conventionsDict}

Here, we list a collection of useful decompositions and identities for CV circuit elements.

\ \\ \noindent \emph{Beam splitter decomposition:}
\begin{subequations} \label{eq:bsDecomp}
\begin{align}
    \Qcircuit @C=0.5cm @R=1.1cm {
        &\varbs{1}& \qw \\
        & \qw     & \qw 
    } 
    \quad &\raisebox{-.5cm}{=} \quad \,
        \Qcircuit @C=0.5cm @R=0.5cm {
         & \ctrlg{\tan{\theta}}{1} & \gate{S^\dag(\sec{\theta})} & \targ & \qw\\
         & \targ & \gate{S(\sec{\theta})} & \ctrlg{-\tan{\theta}}{-1} & \qw
    } 
    \\
    \quad &\raisebox{-.5cm}{=} \,
       \Qcircuit @C=0.5cm @R=0.5cm {
         & \targ & \gate{S(\sec{\theta})} & \ctrlg{\tan{\theta}}{1}  & \qw\\
         & \ctrlg{-\tan{\theta}}{-1}  & \gate{S^\dag(\sec{\theta})} & \targ & \qw
    }
\end{align}
\end{subequations}
where $\CX^{jk}(g) \coloneqq e^{- i g \op{q}_j \otimes \op{p}_k}$ describes the controlled-shift gates in the circuit, with the control being mode $j$ and target mode $k$. 

\ \\ \noindent \emph{Reversing a beam splitter:}
\begin{equation}
\begin{split} \label{eq:bsdirectionswap}
         \Qcircuit @C=1.2em @R=2.35em  
         {
         	& \qswap      & \varbs{1} & \qswap      & \qw \\
         	& \qswap \qwx & \qw       & \qswap \qwx & \qw
  		  } 
      \raisebox{-1.25em}{\quad $\, = \, $~~}
         \Qcircuit @C=1.2em @R=1.9em 
         {
         	& \gate{F^2} & \varbs{1} & \gate{F^2} & \qw \\
         	& \qw        & \qw       & \qw        & \qw
  		  } 
      \raisebox{-1.25em}{\quad $\, = \, $~~}
         \Qcircuit @C=1.2em @R=2.4em  
         {
         	& \qw        & \qw \\
         	& \varbs{-1} & \qw
  		  } 
\end{split}	
\quad ,
\end{equation}

\ \\ \noindent \emph{$\CZ$ decomposition:}
\begin{align}
    \Qcircuit @C=0.5cm @R=0.4cm {
        &\ctrlg{g}{1}& \qw & \raisebox{-1cm}{=} & & \qw & \gate{P(-g)} & \bsbal{1} & \qw\\
        & \ctrl{-1} & \qw & &  & \bsbal{-1} & \gate{P(g)} & \qw & \qw
    }
\end{align}
where $\CZ^{jk}(g) \coloneqq e^{ig \op{q}_j \otimes \op{q}_k}$ describes the controlled-shift gates on the left-hand side.

\ \\ \noindent \emph{SWAP decompositions:}
\begin{subequations} \label{swapcircuitidentity}
\begin{align} 
    \Qcircuit @C=1.35em @R=2.2em @! 
         {
         	& \qswap      &  \qw \\
         	& \qswap \qwx &  \qw
  		  } 
      &  \raisebox{-1.1em}{\quad $\, = \, $~~} 
         \Qcircuit @C=0.5cm @R=0.45cm {
        &\gate{F^2} & \ctrl{1} & \targ          & \ctrl{1} & \qw \\
        & \qw       & \targ    & \ctrlg{-1}{-1} & \targ    & \qw 
    }
    \\
     &  \raisebox{-1.1em}{\quad $\, = \, $~~}
         \Qcircuit @C=0.1em @R=0.425em @! 
         {
            & \gate{F^2} & \bsbal{1} & \bsbal{1} & \qw         \\
            & \qw        & \qw       & \qw       & \qw  
  		  } \; \raisebox{-1.1em}{,} 
\end{align}
\end{subequations}
Interestingly, this implies that beam splitters as defined in Eq.~\eqref{BSdef} act on CV systems in analogy to the way $\sqrt{\text{SWAP}}$ gates act on qubits or qudits.

\ \\ \noindent \emph{EPR state:}
\begin{align}
    \raisebox{-0.35cm}{$\ket{\EPR} \coloneqq \int ds \ket{s}_q \otimes \ket{s}_q = $} \quad
    \Qcircuit @C=0.5cm @R=0.75cm {
        & \qw & \EPRdr & \\
        & \qw & \EPRur & 
    }
\end{align}

\ \\ \noindent \emph{The teleported gate:}
\begin{align}
    \Qcircuit @C=0.5cm @R=0.4cm {
        &\bsbal{1}& \qw & \nogate{\ket{\psi}} & \raisebox{-1cm}{=} & & \qw & \qw & \EPRdr & \\
        & \qw & \qw & \nogate{\ket{\phi}} & &  & \qw & \gate{\frac{1}{\sqrt{\pi}}A(\psi,\phi)} & \EPRur & 
    }
\end{align}
where
    $\op A(\psi, \phi) = \iint d^2 \alpha \; \tilde \psi(\alpha_I)\phi(\alpha_R)\op D(\alpha)$.

\ \\ \noindent \emph{Basic bounce:}
\begin{align}
    \Qcircuit @C=0.5cm @R=0.5cm {
        &\gate{O} & \EPRdr & \raisebox{-1cm}{=} & & \qw & \EPRdr & \\
        & \qw & \EPRur & &  & \gate{O^\tp}  & \EPRur & 
    }
\end{align}
with the transpose defined in the position basis. 

\ \\ \noindent \emph{Beam-splitter double bounce:}
\begin{align}
    \Qcircuit @C=0.5cm @R=0.5cm {
        & & \qw & \bsbal{2} & \EPRdr && \raisebox{-1.5cm}{=}  & & \qw & \qw & \EPRdr &\\
        &  & \qw & \qw & \EPRur &&&  & \qw &\qw & \EPRur & \\
       &  & \qw &\qw & \EPRdr & & &  & \qw & \qw & \EPRdr & \\
        & & \qw & \qw & \EPRur & & &  & \qw & \bsbal{-2}  & \EPRur & 
    }
\end{align}

\ \\ \noindent \emph{Single-mode teleported gate:}
\begin{subequations}
    \begin{align}
        \op V(\theta_1,\theta_2)&=\op R(\theta_+-\tfrac{\pi}{2})\op S(\tan{\theta_-})\op R(\theta_+) \label{eq:vGateOrig} \\
        &= \op R(\theta_1-\tfrac{\pi}{2}) \op P[2\cot (2 \theta_-)] \op R(\theta_1-\tfrac{\pi}{2})  \label{eq:vGateShear} \\ 
        &= \op R(\theta_1 -\pi) \op P_{\op{p}}[2\cot (2 \theta_-)] \op R(\theta_1)\label{eq:vGatePShear}
     \, ,
    \end{align}
\end{subequations}

\ \\ \noindent \emph{$\CX$ reordering:}
\begin{align}\label{commuteCX}
    \Qcircuit @C=0.5cm @R=0.5cm
        {
        &\qw&\targ&\qw&\qw&\\
        &\targ&\ctrlg{b}{-1}&\qw&\qw&\\
        &\ctrlg{a}{-1}&\qw&\qw&\qw&
        }
        \raisebox{-2em}{=} \quad
        \Qcircuit @C=0.5cm @R=0.5cm
        {
        &\targ&\qw&\targ&\qw&\\
        &\ctrlg{b}{-1}&\targ&\qw&\qw&\\
        &\qw&\ctrlg{a}{-1}&\ctrlg{\raisebox{1em}{\scriptsize{$-ab$}}}{-2}&\qw&
        }
\end{align}

\bibliography{ref}

\begin{thebibliography}{41}%
\makeatletter
\providecommand \@ifxundefined [1]{%
 \@ifx{#1\undefined}
}%
\providecommand \@ifnum [1]{%
 \ifnum #1\expandafter \@firstoftwo
 \else \expandafter \@secondoftwo
 \fi
}%
\providecommand \@ifx [1]{%
 \ifx #1\expandafter \@firstoftwo
 \else \expandafter \@secondoftwo
 \fi
}%
\providecommand \natexlab [1]{#1}%
\providecommand \enquote  [1]{``#1''}%
\providecommand \bibnamefont  [1]{#1}%
\providecommand \bibfnamefont [1]{#1}%
\providecommand \citenamefont [1]{#1}%
\providecommand \href@noop [0]{\@secondoftwo}%
\providecommand \href [0]{\begingroup \@sanitize@url \@href}%
\providecommand \@href[1]{\@@startlink{#1}\@@href}%
\providecommand \@@href[1]{\endgroup#1\@@endlink}%
\providecommand \@sanitize@url [0]{\catcode `\\12\catcode `\$12\catcode
  `\&12\catcode `\#12\catcode `\^12\catcode `\_12\catcode `\%12\relax}%
\providecommand \@@startlink[1]{}%
\providecommand \@@endlink[0]{}%
\providecommand \url  [0]{\begingroup\@sanitize@url \@url }%
\providecommand \@url [1]{\endgroup\@href {#1}{\urlprefix }}%
\providecommand \urlprefix  [0]{URL }%
\providecommand \Eprint [0]{\href }%
\providecommand \doibase [0]{https://doi.org/}%
\providecommand \selectlanguage [0]{\@gobble}%
\providecommand \bibinfo  [0]{\@secondoftwo}%
\providecommand \bibfield  [0]{\@secondoftwo}%
\providecommand \translation [1]{[#1]}%
\providecommand \BibitemOpen [0]{}%
\providecommand \bibitemStop [0]{}%
\providecommand \bibitemNoStop [0]{.\EOS\space}%
\providecommand \EOS [0]{\spacefactor3000\relax}%
\providecommand \BibitemShut  [1]{\csname bibitem#1\endcsname}%
\let\auto@bib@innerbib\@empty
\bibitem [{\citenamefont {Menicucci}\ \emph {et~al.}(2006)\citenamefont
  {Menicucci}, \citenamefont {van Loock}, \citenamefont {Gu}, \citenamefont
  {Weedbrook}, \citenamefont {Ralph},\ and\ \citenamefont
  {Nielsen}}]{Menicucci2006}%
  \BibitemOpen
  \bibfield  {author} {\bibinfo {author} {\bibfnamefont {N.~C.}\ \bibnamefont
  {Menicucci}}, \bibinfo {author} {\bibfnamefont {P.}~\bibnamefont {van
  Loock}}, \bibinfo {author} {\bibfnamefont {M.}~\bibnamefont {Gu}}, \bibinfo
  {author} {\bibfnamefont {C.}~\bibnamefont {Weedbrook}}, \bibinfo {author}
  {\bibfnamefont {T.~C.}\ \bibnamefont {Ralph}},\ and\ \bibinfo {author}
  {\bibfnamefont {M.~A.}\ \bibnamefont {Nielsen}},\ }\bibfield  {title}
  {\bibinfo {title} {Universal quantum computation with continuous-variable
  cluster states},\ }\href {https://doi.org/10.1103/PhysRevLett.97.110501}
  {\bibfield  {journal} {\bibinfo  {journal} {Phys. Rev. Lett.}\ }\textbf
  {\bibinfo {volume} {97}},\ \bibinfo {pages} {110501} (\bibinfo {year}
  {2006})}\BibitemShut {NoStop}%
\bibitem [{\citenamefont {Gu}\ \emph {et~al.}(2009)\citenamefont {Gu},
  \citenamefont {Weedbrook}, \citenamefont {Menicucci}, \citenamefont {Ralph},\
  and\ \citenamefont {van Loock}}]{Gu2009}%
  \BibitemOpen
  \bibfield  {author} {\bibinfo {author} {\bibfnamefont {M.}~\bibnamefont
  {Gu}}, \bibinfo {author} {\bibfnamefont {C.}~\bibnamefont {Weedbrook}},
  \bibinfo {author} {\bibfnamefont {N.~C.}\ \bibnamefont {Menicucci}}, \bibinfo
  {author} {\bibfnamefont {T.~C.}\ \bibnamefont {Ralph}},\ and\ \bibinfo
  {author} {\bibfnamefont {P.}~\bibnamefont {van Loock}},\ }\bibfield  {title}
  {\bibinfo {title} {Quantum computing with continuous-variable clusters},\
  }\href {https://doi.org/10.1103/PhysRevA.79.062318} {\bibfield  {journal}
  {\bibinfo  {journal} {Phys. Rev. A}\ }\textbf {\bibinfo {volume} {79}},\
  \bibinfo {pages} {062318} (\bibinfo {year} {2009})}\BibitemShut {NoStop}%
\bibitem [{\citenamefont {Flammia}\ \emph {et~al.}(2009)\citenamefont
  {Flammia}, \citenamefont {Menicucci},\ and\ \citenamefont
  {Pfister}}]{Flammia2009optical}%
  \BibitemOpen
  \bibfield  {author} {\bibinfo {author} {\bibfnamefont {S.~T.}\ \bibnamefont
  {Flammia}}, \bibinfo {author} {\bibfnamefont {N.~C.}\ \bibnamefont
  {Menicucci}},\ and\ \bibinfo {author} {\bibfnamefont {O.}~\bibnamefont
  {Pfister}},\ }\bibfield  {title} {\bibinfo {title} {The optical frequency
  comb as a one-way quantum computer},\ }\href
  {https://doi.org/10.1088/0953-4075/42/11/114009} {\bibfield  {journal}
  {\bibinfo  {journal} {Journal of Physics B: Atomic, Molecular and Optical
  Physics}\ }\textbf {\bibinfo {volume} {42}},\ \bibinfo {pages} {114009}
  (\bibinfo {year} {2009})}\BibitemShut {NoStop}%
\bibitem [{\citenamefont {Menicucci}(2011)}]{Menicucci2011tempmodeCVCS}%
  \BibitemOpen
  \bibfield  {author} {\bibinfo {author} {\bibfnamefont {N.~C.}\ \bibnamefont
  {Menicucci}},\ }\bibfield  {title} {\bibinfo {title} {{Temporal-mode
  continuous-variable cluster states using linear optics}},\ }\href
  {http://dx.doi.org/10.1103/PhysRevA.83.062314} {\bibfield  {journal}
  {\bibinfo  {journal} {Phys. Rev. A}\ }\textbf {\bibinfo {volume} {83}},\
  \bibinfo {pages} {062314} (\bibinfo {year} {2011})}\BibitemShut {NoStop}%
\bibitem [{\citenamefont {Wang}\ \emph {et~al.}(2014)\citenamefont {Wang},
  \citenamefont {Chen}, \citenamefont {Menicucci},\ and\ \citenamefont
  {Pfister}}]{Wang2014hypercubic}%
  \BibitemOpen
  \bibfield  {author} {\bibinfo {author} {\bibfnamefont {P.}~\bibnamefont
  {Wang}}, \bibinfo {author} {\bibfnamefont {M.}~\bibnamefont {Chen}}, \bibinfo
  {author} {\bibfnamefont {N.~C.}\ \bibnamefont {Menicucci}},\ and\ \bibinfo
  {author} {\bibfnamefont {O.}~\bibnamefont {Pfister}},\ }\bibfield  {title}
  {\bibinfo {title} {{Weaving quantum optical frequency combs into
  continuous-variable hypercubic cluster states}},\ }\href
  {http://dx.doi.org/10.1103/PhysRevA.90.032325} {\bibfield  {journal}
  {\bibinfo  {journal} {Phys. Rev. A}\ }\textbf {\bibinfo {volume} {90}},\
  \bibinfo {pages} {032325} (\bibinfo {year} {2014})}\BibitemShut {NoStop}%
\bibitem [{\citenamefont {Alexander}\ \emph {et~al.}(2016)\citenamefont
  {Alexander}, \citenamefont {Wang}, \citenamefont {Sridhar}, \citenamefont
  {Chen}, \citenamefont {Pfister},\ and\ \citenamefont
  {Menicucci}}]{Alexander2016oneway}%
  \BibitemOpen
  \bibfield  {author} {\bibinfo {author} {\bibfnamefont {R.~N.}\ \bibnamefont
  {Alexander}}, \bibinfo {author} {\bibfnamefont {P.}~\bibnamefont {Wang}},
  \bibinfo {author} {\bibfnamefont {N.}~\bibnamefont {Sridhar}}, \bibinfo
  {author} {\bibfnamefont {M.}~\bibnamefont {Chen}}, \bibinfo {author}
  {\bibfnamefont {O.}~\bibnamefont {Pfister}},\ and\ \bibinfo {author}
  {\bibfnamefont {N.~C.}\ \bibnamefont {Menicucci}},\ }\bibfield  {title}
  {\bibinfo {title} {One-way quantum computing with arbitrarily large
  time-frequency continuous-variable cluster states from a single optical
  parametric oscillator},\ }\href
  {http://dx.doi.org/10.1103/PhysRevA.94.032327} {\bibfield  {journal}
  {\bibinfo  {journal} {Phys. Rev. A}\ }\textbf {\bibinfo {volume} {94}},\
  \bibinfo {pages} {032327} (\bibinfo {year} {2016})}\BibitemShut {NoStop}%
\bibitem [{\citenamefont {Alexander}\ \emph {et~al.}(2018)\citenamefont
  {Alexander}, \citenamefont {Yokoyama}, \citenamefont {Furusawa},\ and\
  \citenamefont {Menicucci}}]{Alexander2018BSL}%
  \BibitemOpen
  \bibfield  {author} {\bibinfo {author} {\bibfnamefont {R.~N.}\ \bibnamefont
  {Alexander}}, \bibinfo {author} {\bibfnamefont {S.}~\bibnamefont {Yokoyama}},
  \bibinfo {author} {\bibfnamefont {A.}~\bibnamefont {Furusawa}},\ and\
  \bibinfo {author} {\bibfnamefont {N.~C.}\ \bibnamefont {Menicucci}},\
  }\bibfield  {title} {\bibinfo {title} {Universal quantum computation with
  temporal-mode bilayer square lattices},\ }\href
  {http://link.aps.org/pdf/10.1103/PhysRevA.97.032302} {\bibfield  {journal}
  {\bibinfo  {journal} {Phys. Rev. A}\ }\textbf {\bibinfo {volume} {97}},\
  \bibinfo {pages} {032302} (\bibinfo {year} {2018})}\BibitemShut {NoStop}%
\bibitem [{\citenamefont {Wu}\ \emph {et~al.}(2020)\citenamefont {Wu},
  \citenamefont {Alexander}, \citenamefont {Liu},\ and\ \citenamefont
  {Zhang}}]{Wu2020scalable}%
  \BibitemOpen
  \bibfield  {author} {\bibinfo {author} {\bibfnamefont {B.-H.}\ \bibnamefont
  {Wu}}, \bibinfo {author} {\bibfnamefont {R.~N.}\ \bibnamefont {Alexander}},
  \bibinfo {author} {\bibfnamefont {S.}~\bibnamefont {Liu}},\ and\ \bibinfo
  {author} {\bibfnamefont {Z.}~\bibnamefont {Zhang}},\ }\bibfield  {title}
  {\bibinfo {title} {Quantum computing with multidimensional
  continuous-variable cluster states in a scalable photonic platform},\ }\href
  {https://doi.org/10.1103/PhysRevResearch.2.023138} {\bibfield  {journal}
  {\bibinfo  {journal} {Phys. Rev. Research}\ }\textbf {\bibinfo {volume}
  {2}},\ \bibinfo {pages} {023138} (\bibinfo {year} {2020})}\BibitemShut
  {NoStop}%
\bibitem [{\citenamefont {Yang}\ \emph {et~al.}(2020)\citenamefont {Yang},
  \citenamefont {Zhang}, \citenamefont {Klich}, \citenamefont
  {Gonz\'alez-Arciniegas},\ and\ \citenamefont
  {Pfister}}]{Yang2020spatiotemporalgraphs}%
  \BibitemOpen
  \bibfield  {author} {\bibinfo {author} {\bibfnamefont {R.}~\bibnamefont
  {Yang}}, \bibinfo {author} {\bibfnamefont {J.}~\bibnamefont {Zhang}},
  \bibinfo {author} {\bibfnamefont {I.}~\bibnamefont {Klich}}, \bibinfo
  {author} {\bibfnamefont {C.}~\bibnamefont {Gonz\'alez-Arciniegas}},\ and\
  \bibinfo {author} {\bibfnamefont {O.}~\bibnamefont {Pfister}},\ }\bibfield
  {title} {\bibinfo {title} {Spatiotemporal graph states from a single optical
  parametric oscillator},\ }\href {https://doi.org/10.1103/PhysRevA.101.043832}
  {\bibfield  {journal} {\bibinfo  {journal} {Phys. Rev. A}\ }\textbf {\bibinfo
  {volume} {101}},\ \bibinfo {pages} {043832} (\bibinfo {year}
  {2020})}\BibitemShut {NoStop}%
\bibitem [{\citenamefont {Chen}\ \emph {et~al.}(2014)\citenamefont {Chen},
  \citenamefont {Menicucci},\ and\ \citenamefont {Pfister}}]{Chen2014}%
  \BibitemOpen
  \bibfield  {author} {\bibinfo {author} {\bibfnamefont {M.}~\bibnamefont
  {Chen}}, \bibinfo {author} {\bibfnamefont {N.~C.}\ \bibnamefont
  {Menicucci}},\ and\ \bibinfo {author} {\bibfnamefont {O.}~\bibnamefont
  {Pfister}},\ }\bibfield  {title} {\bibinfo {title} {Experimental realization
  of multipartite entanglement of 60 modes of a quantum optical frequency
  comb},\ }\href {http://dx.doi.org/10.1103/PhysRevLett.112.120505} {\bibfield
  {journal} {\bibinfo  {journal} {Phys. Rev. Lett.}\ }\textbf {\bibinfo
  {volume} {112}},\ \bibinfo {pages} {120505} (\bibinfo {year}
  {2014})}\BibitemShut {NoStop}%
\bibitem [{\citenamefont {Yokoyama}\ \emph {et~al.}(2013)\citenamefont
  {Yokoyama}, \citenamefont {Ukai}, \citenamefont {Armstrong}, \citenamefont
  {Sornphiphatphong}, \citenamefont {Kaji}, \citenamefont {Suzuki},
  \citenamefont {Yoshikawa}, \citenamefont {Yonezawa}, \citenamefont
  {Menicucci},\ and\ \citenamefont {Furusawa}}]{Yokoyama2013ultra}%
  \BibitemOpen
  \bibfield  {author} {\bibinfo {author} {\bibfnamefont {S.}~\bibnamefont
  {Yokoyama}}, \bibinfo {author} {\bibfnamefont {R.}~\bibnamefont {Ukai}},
  \bibinfo {author} {\bibfnamefont {S.~C.}\ \bibnamefont {Armstrong}}, \bibinfo
  {author} {\bibfnamefont {C.}~\bibnamefont {Sornphiphatphong}}, \bibinfo
  {author} {\bibfnamefont {T.}~\bibnamefont {Kaji}}, \bibinfo {author}
  {\bibfnamefont {S.}~\bibnamefont {Suzuki}}, \bibinfo {author} {\bibfnamefont
  {J.-i.}\ \bibnamefont {Yoshikawa}}, \bibinfo {author} {\bibfnamefont
  {H.}~\bibnamefont {Yonezawa}}, \bibinfo {author} {\bibfnamefont {N.~C.}\
  \bibnamefont {Menicucci}},\ and\ \bibinfo {author} {\bibfnamefont
  {A.}~\bibnamefont {Furusawa}},\ }\bibfield  {title} {\bibinfo {title}
  {Ultra-large-scale continuous-variable cluster states multiplexed in the time
  domain},\ }\href {https://doi.org/10.1038/nphoton.2013.287} {\bibfield
  {journal} {\bibinfo  {journal} {Nature Photonics}\ }\textbf {\bibinfo
  {volume} {7}},\ \bibinfo {pages} {982} (\bibinfo {year} {2013})}\BibitemShut
  {NoStop}%
\bibitem [{\citenamefont {Yoshikawa}\ \emph {et~al.}(2016)\citenamefont
  {Yoshikawa}, \citenamefont {Yokoyama}, \citenamefont {Kaji}, \citenamefont
  {Sornphiphatphong}, \citenamefont {Shiozawa}, \citenamefont {Makino},\ and\
  \citenamefont {Furusawa}}]{Yoshikawa2016}%
  \BibitemOpen
  \bibfield  {author} {\bibinfo {author} {\bibfnamefont {J.-i.}\ \bibnamefont
  {Yoshikawa}}, \bibinfo {author} {\bibfnamefont {S.}~\bibnamefont {Yokoyama}},
  \bibinfo {author} {\bibfnamefont {T.}~\bibnamefont {Kaji}}, \bibinfo {author}
  {\bibfnamefont {C.}~\bibnamefont {Sornphiphatphong}}, \bibinfo {author}
  {\bibfnamefont {Y.}~\bibnamefont {Shiozawa}}, \bibinfo {author}
  {\bibfnamefont {K.}~\bibnamefont {Makino}},\ and\ \bibinfo {author}
  {\bibfnamefont {A.}~\bibnamefont {Furusawa}},\ }\bibfield  {title} {\bibinfo
  {title} {Generation of one-million-mode continuous-variable cluster state by
  unlimited time-domain multiplexing},\ }\href
  {https://doi.org/10.1063/1.4962732} {\bibfield  {journal} {\bibinfo
  {journal} {APL Photonics}\ }\textbf {\bibinfo {volume} {1}},\ \bibinfo
  {pages} {060801} (\bibinfo {year} {2016})}\BibitemShut {NoStop}%
\bibitem [{\citenamefont {Asavanant}\ \emph {et~al.}(2019)\citenamefont
  {Asavanant}, \citenamefont {Shiozawa}, \citenamefont {Yokoyama},
  \citenamefont {Charoensombutamon}, \citenamefont {Emura}, \citenamefont
  {Alexander}, \citenamefont {Takeda}, \citenamefont {Yoshikawa}, \citenamefont
  {Menicucci}, \citenamefont {Yonezawa} \emph {et~al.}}]{Asavanant2019detCVCS}%
  \BibitemOpen
  \bibfield  {author} {\bibinfo {author} {\bibfnamefont {W.}~\bibnamefont
  {Asavanant}}, \bibinfo {author} {\bibfnamefont {Y.}~\bibnamefont {Shiozawa}},
  \bibinfo {author} {\bibfnamefont {S.}~\bibnamefont {Yokoyama}}, \bibinfo
  {author} {\bibfnamefont {B.}~\bibnamefont {Charoensombutamon}}, \bibinfo
  {author} {\bibfnamefont {H.}~\bibnamefont {Emura}}, \bibinfo {author}
  {\bibfnamefont {R.~N.}\ \bibnamefont {Alexander}}, \bibinfo {author}
  {\bibfnamefont {S.}~\bibnamefont {Takeda}}, \bibinfo {author} {\bibfnamefont
  {J.-i.}\ \bibnamefont {Yoshikawa}}, \bibinfo {author} {\bibfnamefont {N.~C.}\
  \bibnamefont {Menicucci}}, \bibinfo {author} {\bibfnamefont {H.}~\bibnamefont
  {Yonezawa}}, \emph {et~al.},\ }\bibfield  {title} {\bibinfo {title}
  {Generation of time-domain-multiplexed two-dimensional cluster state},\
  }\href {https://doi.org/10.1126/science.aay2645} {\bibfield  {journal}
  {\bibinfo  {journal} {Science}\ }\textbf {\bibinfo {volume} {366}},\ \bibinfo
  {pages} {373} (\bibinfo {year} {2019})}\BibitemShut {NoStop}%
\bibitem [{\citenamefont {Larsen}\ \emph {et~al.}(2019)\citenamefont {Larsen},
  \citenamefont {Guo}, \citenamefont {Breum}, \citenamefont
  {Neergaard-Nielsen},\ and\ \citenamefont {Andersen}}]{Larsen2019detCVCS}%
  \BibitemOpen
  \bibfield  {author} {\bibinfo {author} {\bibfnamefont {M.~V.}\ \bibnamefont
  {Larsen}}, \bibinfo {author} {\bibfnamefont {X.}~\bibnamefont {Guo}},
  \bibinfo {author} {\bibfnamefont {C.~R.}\ \bibnamefont {Breum}}, \bibinfo
  {author} {\bibfnamefont {J.~S.}\ \bibnamefont {Neergaard-Nielsen}},\ and\
  \bibinfo {author} {\bibfnamefont {U.~L.}\ \bibnamefont {Andersen}},\
  }\bibfield  {title} {\bibinfo {title} {Deterministic generation of a
  two-dimensional cluster state},\ }\href
  {https://doi.org/10.1126/science.aay4354} {\bibfield  {journal} {\bibinfo
  {journal} {Science}\ }\textbf {\bibinfo {volume} {366}},\ \bibinfo {pages}
  {369–372} (\bibinfo {year} {2019})}\BibitemShut {NoStop}%
\bibitem [{\citenamefont {van Loock}(2007)}]{Loock2007GaussgatesCVCS}%
  \BibitemOpen
  \bibfield  {author} {\bibinfo {author} {\bibfnamefont {P.}~\bibnamefont {van
  Loock}},\ }\bibfield  {title} {\bibinfo {title} {Examples of {Gaussian}
  cluster computation},\ }\href {https://doi.org/10.1364/JOSAB.24.000340}
  {\bibfield  {journal} {\bibinfo  {journal} {J. Opt. Soc. Am. B}\ }\textbf
  {\bibinfo {volume} {24}},\ \bibinfo {pages} {340} (\bibinfo {year}
  {2007})}\BibitemShut {NoStop}%
\bibitem [{\citenamefont {Asavanant}\ \emph {et~al.}(2021)\citenamefont
  {Asavanant}, \citenamefont {Charoensombutamon}, \citenamefont {Yokoyama},
  \citenamefont {Ebihara}, \citenamefont {Nakamura}, \citenamefont {Alexander},
  \citenamefont {Endo}, \citenamefont {Yoshikawa}, \citenamefont {Menicucci},
  \citenamefont {Yonezawa},\ and\ \citenamefont
  {Furusawa}}]{Asavanant2020detgatesCVCS}%
  \BibitemOpen
  \bibfield  {author} {\bibinfo {author} {\bibfnamefont {W.}~\bibnamefont
  {Asavanant}}, \bibinfo {author} {\bibfnamefont {B.}~\bibnamefont
  {Charoensombutamon}}, \bibinfo {author} {\bibfnamefont {S.}~\bibnamefont
  {Yokoyama}}, \bibinfo {author} {\bibfnamefont {T.}~\bibnamefont {Ebihara}},
  \bibinfo {author} {\bibfnamefont {T.}~\bibnamefont {Nakamura}}, \bibinfo
  {author} {\bibfnamefont {R.~N.}\ \bibnamefont {Alexander}}, \bibinfo {author}
  {\bibfnamefont {M.}~\bibnamefont {Endo}}, \bibinfo {author} {\bibfnamefont
  {J.-i.}\ \bibnamefont {Yoshikawa}}, \bibinfo {author} {\bibfnamefont {N.~C.}\
  \bibnamefont {Menicucci}}, \bibinfo {author} {\bibfnamefont {H.}~\bibnamefont
  {Yonezawa}},\ and\ \bibinfo {author} {\bibfnamefont {A.}~\bibnamefont
  {Furusawa}},\ }\bibfield  {title} {\bibinfo {title} {Time-domain-multiplexed
  measurement-based quantum operations with 25-mhz clock frequency},\ }\href
  {https://doi.org/10.1103/PhysRevApplied.16.034005} {\bibfield  {journal}
  {\bibinfo  {journal} {Phys. Rev. Appl.}\ }\textbf {\bibinfo {volume} {16}},\
  \bibinfo {pages} {034005} (\bibinfo {year} {2021})}\BibitemShut {NoStop}%
\bibitem [{\citenamefont {Larsen}\ \emph
  {et~al.}(2021{\natexlab{a}})\citenamefont {Larsen}, \citenamefont {Guo},
  \citenamefont {Breum}, \citenamefont {Neergaard-Nielsen},\ and\ \citenamefont
  {Andersen}}]{Larsen2020detgatesCVCS}%
  \BibitemOpen
  \bibfield  {author} {\bibinfo {author} {\bibfnamefont {M.~V.}\ \bibnamefont
  {Larsen}}, \bibinfo {author} {\bibfnamefont {X.}~\bibnamefont {Guo}},
  \bibinfo {author} {\bibfnamefont {C.~R.}\ \bibnamefont {Breum}}, \bibinfo
  {author} {\bibfnamefont {J.~S.}\ \bibnamefont {Neergaard-Nielsen}},\ and\
  \bibinfo {author} {\bibfnamefont {U.~L.}\ \bibnamefont {Andersen}},\
  }\bibfield  {title} {\bibinfo {title} {Deterministic multi-mode gates on a
  scalable photonic quantum computing platform},\ }\href
  {https://doi.org/10.1038/s41567-021-01296-y} {\bibfield  {journal} {\bibinfo
  {journal} {Nature Physics}\ }\textbf {\bibinfo {volume} {17}},\ \bibinfo
  {pages} {1018} (\bibinfo {year} {2021}{\natexlab{a}})}\BibitemShut {NoStop}%
\bibitem [{\citenamefont {Alexander}\ \emph {et~al.}(2014)\citenamefont
  {Alexander}, \citenamefont {Armstrong}, \citenamefont {Ukai},\ and\
  \citenamefont {Menicucci}}]{Alexander2014noise}%
  \BibitemOpen
  \bibfield  {author} {\bibinfo {author} {\bibfnamefont {R.~N.}\ \bibnamefont
  {Alexander}}, \bibinfo {author} {\bibfnamefont {S.~C.}\ \bibnamefont
  {Armstrong}}, \bibinfo {author} {\bibfnamefont {R.}~\bibnamefont {Ukai}},\
  and\ \bibinfo {author} {\bibfnamefont {N.~C.}\ \bibnamefont {Menicucci}},\
  }\bibfield  {title} {\bibinfo {title} {Noise analysis of single-mode
  {Gaussian} operations using continuous-variable cluster states},\ }\href
  {https://doi.org/10.1103/PhysRevA.90.062324} {\bibfield  {journal} {\bibinfo
  {journal} {Phys. Rev. A}\ }\textbf {\bibinfo {volume} {90}},\ \bibinfo
  {pages} {062324} (\bibinfo {year} {2014})}\BibitemShut {NoStop}%
\bibitem [{\citenamefont {Walshe}\ \emph {et~al.}(2019)\citenamefont {Walshe},
  \citenamefont {Mensen}, \citenamefont {Baragiola},\ and\ \citenamefont
  {Menicucci}}]{Walshe2019}%
  \BibitemOpen
  \bibfield  {author} {\bibinfo {author} {\bibfnamefont {B.~W.}\ \bibnamefont
  {Walshe}}, \bibinfo {author} {\bibfnamefont {L.~J.}\ \bibnamefont {Mensen}},
  \bibinfo {author} {\bibfnamefont {B.~Q.}\ \bibnamefont {Baragiola}},\ and\
  \bibinfo {author} {\bibfnamefont {N.~C.}\ \bibnamefont {Menicucci}},\
  }\bibfield  {title} {\bibinfo {title} {Robust fault tolerance for
  continuous-variable cluster states with excess antisqueezing},\ }\href
  {https://doi.org/10.1103/PhysRevA.100.010301} {\bibfield  {journal} {\bibinfo
   {journal} {Phys. Rev. A}\ }\textbf {\bibinfo {volume} {100}},\ \bibinfo
  {pages} {010301} (\bibinfo {year} {2019})}\BibitemShut {NoStop}%
\bibitem [{\citenamefont {Ohliger}\ \emph {et~al.}(2010)\citenamefont
  {Ohliger}, \citenamefont {Kieling},\ and\ \citenamefont
  {Eisert}}]{ohliger2010limitations}%
  \BibitemOpen
  \bibfield  {author} {\bibinfo {author} {\bibfnamefont {M.}~\bibnamefont
  {Ohliger}}, \bibinfo {author} {\bibfnamefont {K.}~\bibnamefont {Kieling}},\
  and\ \bibinfo {author} {\bibfnamefont {J.}~\bibnamefont {Eisert}},\
  }\bibfield  {title} {\bibinfo {title} {Limitations of quantum computing with
  gaussian cluster states},\ }\href
  {https://doi.org/10.1103/PhysRevA.82.042336} {\bibfield  {journal} {\bibinfo
  {journal} {Phys. Rev. A}\ }\textbf {\bibinfo {volume} {82}},\ \bibinfo
  {pages} {042336} (\bibinfo {year} {2010})}\BibitemShut {NoStop}%
\bibitem [{\citenamefont {Gottesman}\ \emph {et~al.}(2001)\citenamefont
  {Gottesman}, \citenamefont {Kitaev},\ and\ \citenamefont {Preskill}}]{GKP}%
  \BibitemOpen
  \bibfield  {author} {\bibinfo {author} {\bibfnamefont {D.}~\bibnamefont
  {Gottesman}}, \bibinfo {author} {\bibfnamefont {A.}~\bibnamefont {Kitaev}},\
  and\ \bibinfo {author} {\bibfnamefont {J.}~\bibnamefont {Preskill}},\
  }\bibfield  {title} {\bibinfo {title} {Encoding a qubit in an oscillator},\
  }\href {https://doi.org/10.1103/PhysRevA.64.012310} {\bibfield  {journal}
  {\bibinfo  {journal} {Phys. Rev. A}\ }\textbf {\bibinfo {volume} {64}},\
  \bibinfo {pages} {012310} (\bibinfo {year} {2001})}\BibitemShut {NoStop}%
\bibitem [{\citenamefont {Pantaleoni}\ \emph {et~al.}(2021)\citenamefont
  {Pantaleoni}, \citenamefont {Baragiola},\ and\ \citenamefont
  {Menicucci}}]{pantaleoni2021hidden}%
  \BibitemOpen
  \bibfield  {author} {\bibinfo {author} {\bibfnamefont {G.}~\bibnamefont
  {Pantaleoni}}, \bibinfo {author} {\bibfnamefont {B.~Q.}\ \bibnamefont
  {Baragiola}},\ and\ \bibinfo {author} {\bibfnamefont {N.~C.}\ \bibnamefont
  {Menicucci}},\ }\bibfield  {title} {\bibinfo {title} {Hidden qubit cluster
  states},\ }\href {https://doi.org/10.1103/PhysRevA.104.012431} {\bibfield
  {journal} {\bibinfo  {journal} {Phys. Rev. A}\ }\textbf {\bibinfo {volume}
  {104}},\ \bibinfo {pages} {012431} (\bibinfo {year} {2021})}\BibitemShut
  {NoStop}%
\bibitem [{\citenamefont {Baragiola}\ \emph {et~al.}(2019)\citenamefont
  {Baragiola}, \citenamefont {Pantaleoni}, \citenamefont {Alexander},
  \citenamefont {Karanjai},\ and\ \citenamefont {Menicucci}}]{Baragiola2019}%
  \BibitemOpen
  \bibfield  {author} {\bibinfo {author} {\bibfnamefont {B.~Q.}\ \bibnamefont
  {Baragiola}}, \bibinfo {author} {\bibfnamefont {G.}~\bibnamefont
  {Pantaleoni}}, \bibinfo {author} {\bibfnamefont {R.~N.}\ \bibnamefont
  {Alexander}}, \bibinfo {author} {\bibfnamefont {A.}~\bibnamefont
  {Karanjai}},\ and\ \bibinfo {author} {\bibfnamefont {N.~C.}\ \bibnamefont
  {Menicucci}},\ }\bibfield  {title} {\bibinfo {title} {All-{Gaussian}
  universality and fault tolerance with the {Gottesman-Kitaev-Preskill} code},\
  }\href {https://doi.org/10.1103/PhysRevLett.123.200502} {\bibfield  {journal}
  {\bibinfo  {journal} {Phys. Rev. Lett.}\ }\textbf {\bibinfo {volume} {123}},\
  \bibinfo {pages} {200502} (\bibinfo {year} {2019})}\BibitemShut {NoStop}%
\bibitem [{\citenamefont {{Noh}}\ \emph {et~al.}(2019)\citenamefont {{Noh}},
  \citenamefont {{Albert}},\ and\ \citenamefont {{Jiang}}}]{Noh2019channel}%
  \BibitemOpen
  \bibfield  {author} {\bibinfo {author} {\bibfnamefont {K.}~\bibnamefont
  {{Noh}}}, \bibinfo {author} {\bibfnamefont {V.~V.}\ \bibnamefont
  {{Albert}}},\ and\ \bibinfo {author} {\bibfnamefont {L.}~\bibnamefont
  {{Jiang}}},\ }\bibfield  {title} {\bibinfo {title} {Quantum capacity bounds
  of {Gaussian} thermal loss channels and achievable rates with
  {Gottesman-Kitaev-Preskill} codes},\ }\href
  {https://doi.org/10.1109/TIT.2018.2873764} {\bibfield  {journal} {\bibinfo
  {journal} {IEEE Transactions on Information Theory}\ }\textbf {\bibinfo
  {volume} {65}},\ \bibinfo {pages} {2563} (\bibinfo {year}
  {2019})}\BibitemShut {NoStop}%
\bibitem [{\citenamefont {Wu}\ and\ \citenamefont
  {Zhuang}(2021)}]{Wu2021GKPgengauss}%
  \BibitemOpen
  \bibfield  {author} {\bibinfo {author} {\bibfnamefont {J.}~\bibnamefont
  {Wu}}\ and\ \bibinfo {author} {\bibfnamefont {Q.}~\bibnamefont {Zhuang}},\
  }\bibfield  {title} {\bibinfo {title} {Continuous-variable error correction
  for general gaussian noises},\ }\href
  {https://doi.org/10.1103/PhysRevApplied.15.034073} {\bibfield  {journal}
  {\bibinfo  {journal} {Phys. Rev. Appl.}\ }\textbf {\bibinfo {volume} {15}},\
  \bibinfo {pages} {034073} (\bibinfo {year} {2021})}\BibitemShut {NoStop}%
\bibitem [{\citenamefont {Menicucci}(2014)}]{menicucci2014fault}%
  \BibitemOpen
  \bibfield  {author} {\bibinfo {author} {\bibfnamefont {N.~C.}\ \bibnamefont
  {Menicucci}},\ }\bibfield  {title} {\bibinfo {title} {Fault-tolerant
  measurement-based quantum computing with continuous-variable cluster
  states},\ }\href {https://doi.org/10.1103/PhysRevLett.112.120504} {\bibfield
  {journal} {\bibinfo  {journal} {Phys. Rev. Lett.}\ }\textbf {\bibinfo
  {volume} {112}},\ \bibinfo {pages} {120504} (\bibinfo {year}
  {2014})}\BibitemShut {NoStop}%
\bibitem [{\citenamefont {Fukui}\ \emph {et~al.}(2018)\citenamefont {Fukui},
  \citenamefont {Tomita}, \citenamefont {Okamoto},\ and\ \citenamefont
  {Fujii}}]{fukui2018high}%
  \BibitemOpen
  \bibfield  {author} {\bibinfo {author} {\bibfnamefont {K.}~\bibnamefont
  {Fukui}}, \bibinfo {author} {\bibfnamefont {A.}~\bibnamefont {Tomita}},
  \bibinfo {author} {\bibfnamefont {A.}~\bibnamefont {Okamoto}},\ and\ \bibinfo
  {author} {\bibfnamefont {K.}~\bibnamefont {Fujii}},\ }\bibfield  {title}
  {\bibinfo {title} {High-threshold fault-tolerant quantum computation with
  analog quantum error correction},\ }\href
  {https://doi.org/10.1103/PhysRevX.8.021054} {\bibfield  {journal} {\bibinfo
  {journal} {Phys. Rev. X}\ }\textbf {\bibinfo {volume} {8}},\ \bibinfo {pages}
  {021054} (\bibinfo {year} {2018})}\BibitemShut {NoStop}%
\bibitem [{\citenamefont {Larsen}\ \emph
  {et~al.}(2021{\natexlab{b}})\citenamefont {Larsen}, \citenamefont
  {Chamberland}, \citenamefont {Noh}, \citenamefont {Neergaard-Nielsen},\ and\
  \citenamefont {Andersen}}]{Larsen2021fault}%
  \BibitemOpen
  \bibfield  {author} {\bibinfo {author} {\bibfnamefont {M.~V.}\ \bibnamefont
  {Larsen}}, \bibinfo {author} {\bibfnamefont {C.}~\bibnamefont {Chamberland}},
  \bibinfo {author} {\bibfnamefont {K.}~\bibnamefont {Noh}}, \bibinfo {author}
  {\bibfnamefont {J.~S.}\ \bibnamefont {Neergaard-Nielsen}},\ and\ \bibinfo
  {author} {\bibfnamefont {U.~L.}\ \bibnamefont {Andersen}},\ }\bibfield
  {title} {\bibinfo {title} {Fault-tolerant continuous-variable
  measurement-based quantum computation architecture},\ }\href
  {https://doi.org/10.1103/PRXQuantum.2.030325} {\bibfield  {journal} {\bibinfo
   {journal} {PRX Quantum}\ }\textbf {\bibinfo {volume} {2}},\ \bibinfo {pages}
  {030325} (\bibinfo {year} {2021}{\natexlab{b}})}\BibitemShut {NoStop}%
\bibitem [{\citenamefont {Bourassa}\ \emph {et~al.}(2021)\citenamefont
  {Bourassa}, \citenamefont {Alexander}, \citenamefont {Vasmer}, \citenamefont
  {Patil}, \citenamefont {Tzitrin}, \citenamefont {Matsuura}, \citenamefont
  {Su}, \citenamefont {Baragiola}, \citenamefont {Guha}, \citenamefont
  {Dauphinais}, \citenamefont {Sabapathy}, \citenamefont {Menicucci},\ and\
  \citenamefont {Dhand}}]{Bourassa2021blueprint}%
  \BibitemOpen
  \bibfield  {author} {\bibinfo {author} {\bibfnamefont {J.~E.}\ \bibnamefont
  {Bourassa}}, \bibinfo {author} {\bibfnamefont {R.~N.}\ \bibnamefont
  {Alexander}}, \bibinfo {author} {\bibfnamefont {M.}~\bibnamefont {Vasmer}},
  \bibinfo {author} {\bibfnamefont {A.}~\bibnamefont {Patil}}, \bibinfo
  {author} {\bibfnamefont {I.}~\bibnamefont {Tzitrin}}, \bibinfo {author}
  {\bibfnamefont {T.}~\bibnamefont {Matsuura}}, \bibinfo {author}
  {\bibfnamefont {D.}~\bibnamefont {Su}}, \bibinfo {author} {\bibfnamefont
  {B.~Q.}\ \bibnamefont {Baragiola}}, \bibinfo {author} {\bibfnamefont
  {S.}~\bibnamefont {Guha}}, \bibinfo {author} {\bibfnamefont {G.}~\bibnamefont
  {Dauphinais}}, \bibinfo {author} {\bibfnamefont {K.~K.}\ \bibnamefont
  {Sabapathy}}, \bibinfo {author} {\bibfnamefont {N.~C.}\ \bibnamefont
  {Menicucci}},\ and\ \bibinfo {author} {\bibfnamefont {I.}~\bibnamefont
  {Dhand}},\ }\bibfield  {title} {\bibinfo {title} {Blueprint for a {S}calable
  {P}hotonic {F}ault-{T}olerant {Q}uantum {C}omputer},\ }\href
  {https://doi.org/10.22331/q-2021-02-04-392} {\bibfield  {journal} {\bibinfo
  {journal} {{Quantum}}\ }\textbf {\bibinfo {volume} {5}},\ \bibinfo {pages}
  {392} (\bibinfo {year} {2021})}\BibitemShut {NoStop}%
\bibitem [{\citenamefont {Tzitrin}\ \emph {et~al.}(2021)\citenamefont
  {Tzitrin}, \citenamefont {Matsuura}, \citenamefont {Alexander}, \citenamefont
  {Dauphinais}, \citenamefont {Bourassa}, \citenamefont {Sabapathy},
  \citenamefont {Menicucci},\ and\ \citenamefont {Dhand}}]{Tzitrin2021passive}%
  \BibitemOpen
  \bibfield  {author} {\bibinfo {author} {\bibfnamefont {I.}~\bibnamefont
  {Tzitrin}}, \bibinfo {author} {\bibfnamefont {T.}~\bibnamefont {Matsuura}},
  \bibinfo {author} {\bibfnamefont {R.~N.}\ \bibnamefont {Alexander}}, \bibinfo
  {author} {\bibfnamefont {G.}~\bibnamefont {Dauphinais}}, \bibinfo {author}
  {\bibfnamefont {J.~E.}\ \bibnamefont {Bourassa}}, \bibinfo {author}
  {\bibfnamefont {K.~K.}\ \bibnamefont {Sabapathy}}, \bibinfo {author}
  {\bibfnamefont {N.~C.}\ \bibnamefont {Menicucci}},\ and\ \bibinfo {author}
  {\bibfnamefont {I.}~\bibnamefont {Dhand}},\ }\bibfield  {title} {\bibinfo
  {title} {Fault-tolerant quantum computation with static linear optics},\
  }\href {https://doi.org/10.1103/PRXQuantum.2.040353} {\bibfield  {journal}
  {\bibinfo  {journal} {PRX Quantum}\ }\textbf {\bibinfo {volume} {2}},\
  \bibinfo {pages} {040353} (\bibinfo {year} {2021})}\BibitemShut {NoStop}%
\bibitem [{\citenamefont {Larsen}\ \emph {et~al.}(2020)\citenamefont {Larsen},
  \citenamefont {Neergaard-Nielsen},\ and\ \citenamefont
  {Andersen}}]{Larsen2020architecture}%
  \BibitemOpen
  \bibfield  {author} {\bibinfo {author} {\bibfnamefont {M.~V.}\ \bibnamefont
  {Larsen}}, \bibinfo {author} {\bibfnamefont {J.~S.}\ \bibnamefont
  {Neergaard-Nielsen}},\ and\ \bibinfo {author} {\bibfnamefont {U.~L.}\
  \bibnamefont {Andersen}},\ }\bibfield  {title} {\bibinfo {title}
  {Architecture and noise analysis of continuous-variable quantum gates using
  two-dimensional cluster states},\ }\href
  {https://doi.org/10.1103/PhysRevA.102.042608} {\bibfield  {journal} {\bibinfo
   {journal} {Phys. Rev. A}\ }\textbf {\bibinfo {volume} {102}},\ \bibinfo
  {pages} {042608} (\bibinfo {year} {2020})}\BibitemShut {NoStop}%
\bibitem [{\citenamefont {Walshe}\ \emph {et~al.}(2021)\citenamefont {Walshe},
  \citenamefont {Alexander}, \citenamefont {Menicucci},\ and\ \citenamefont
  {Baragiola}}]{walshe2021streamlined}%
  \BibitemOpen
  \bibfield  {author} {\bibinfo {author} {\bibfnamefont {B.~W.}\ \bibnamefont
  {Walshe}}, \bibinfo {author} {\bibfnamefont {R.~N.}\ \bibnamefont
  {Alexander}}, \bibinfo {author} {\bibfnamefont {N.~C.}\ \bibnamefont
  {Menicucci}},\ and\ \bibinfo {author} {\bibfnamefont {B.~Q.}\ \bibnamefont
  {Baragiola}},\ }\bibfield  {title} {\bibinfo {title} {Streamlined quantum
  computing with macronode cluster states},\ }\href
  {https://doi.org/10.1103/PhysRevA.104.062427} {\bibfield  {journal} {\bibinfo
   {journal} {Phys. Rev. A}\ }\textbf {\bibinfo {volume} {104}},\ \bibinfo
  {pages} {062427} (\bibinfo {year} {2021})}\BibitemShut {NoStop}%
\bibitem [{\citenamefont {Alexander}\ and\ \citenamefont
  {Menicucci}(2016)}]{Alexander2016flexible}%
  \BibitemOpen
  \bibfield  {author} {\bibinfo {author} {\bibfnamefont {R.~N.}\ \bibnamefont
  {Alexander}}\ and\ \bibinfo {author} {\bibfnamefont {N.~C.}\ \bibnamefont
  {Menicucci}},\ }\bibfield  {title} {\bibinfo {title} {Flexible quantum
  circuits using scalable continuous-variable cluster states},\ }\href
  {https://doi.org/10.1103/PhysRevA.93.062326} {\bibfield  {journal} {\bibinfo
  {journal} {Phys. Rev. A}\ }\textbf {\bibinfo {volume} {93}},\ \bibinfo
  {pages} {062326} (\bibinfo {year} {2016})}\BibitemShut {NoStop}%
\bibitem [{\citenamefont {Walshe}\ \emph {et~al.}(2020)\citenamefont {Walshe},
  \citenamefont {Baragiola}, \citenamefont {Alexander},\ and\ \citenamefont
  {Menicucci}}]{Walshe2020}%
  \BibitemOpen
  \bibfield  {author} {\bibinfo {author} {\bibfnamefont {B.~W.}\ \bibnamefont
  {Walshe}}, \bibinfo {author} {\bibfnamefont {B.~Q.}\ \bibnamefont
  {Baragiola}}, \bibinfo {author} {\bibfnamefont {R.~N.}\ \bibnamefont
  {Alexander}},\ and\ \bibinfo {author} {\bibfnamefont {N.~C.}\ \bibnamefont
  {Menicucci}},\ }\bibfield  {title} {\bibinfo {title} {Continuous-variable
  gate teleportation and bosonic-code error correction},\ }\href
  {https://doi.org/10.1103/PhysRevA.102.062411} {\bibfield  {journal} {\bibinfo
   {journal} {Phys. Rev. A}\ }\textbf {\bibinfo {volume} {102}},\ \bibinfo
  {pages} {062411} (\bibinfo {year} {2020})}\BibitemShut {NoStop}%
\bibitem [{\citenamefont {Duivenvoorden}\ \emph {et~al.}(2017)\citenamefont
  {Duivenvoorden}, \citenamefont {Terhal},\ and\ \citenamefont
  {Weigand}}]{Duivenvoorden2017}%
  \BibitemOpen
  \bibfield  {author} {\bibinfo {author} {\bibfnamefont {K.}~\bibnamefont
  {Duivenvoorden}}, \bibinfo {author} {\bibfnamefont {B.~M.}\ \bibnamefont
  {Terhal}},\ and\ \bibinfo {author} {\bibfnamefont {D.}~\bibnamefont
  {Weigand}},\ }\bibfield  {title} {\bibinfo {title} {Single-mode displacement
  sensor},\ }\href {https://doi.org/10.1103/PhysRevA.95.012305} {\bibfield
  {journal} {\bibinfo  {journal} {Phys. Rev. A}\ }\textbf {\bibinfo {volume}
  {95}},\ \bibinfo {pages} {012305} (\bibinfo {year} {2017})}\BibitemShut
  {NoStop}%
\bibitem [{\citenamefont {Volkoff}\ and\ \citenamefont
  {Suba{\c{s}}ı}(2022)}]{Volkoff2022}%
  \BibitemOpen
  \bibfield  {author} {\bibinfo {author} {\bibfnamefont {T.~J.}\ \bibnamefont
  {Volkoff}}\ and\ \bibinfo {author} {\bibfnamefont {Y.}~\bibnamefont
  {Suba{\c{s}}ı}},\ }\bibfield  {title} {\bibinfo {title} {Ancilla-free
  continuous-variable {SWAP} test},\ }\href
  {https://doi.org/10.22331/q-2022-09-08-800} {\bibfield  {journal} {\bibinfo
  {journal} {{Quantum}}\ }\textbf {\bibinfo {volume} {6}},\ \bibinfo {pages}
  {800} (\bibinfo {year} {2022})}\BibitemShut {NoStop}%
\bibitem [{\citenamefont {Menicucci}\ \emph {et~al.}(2011)\citenamefont
  {Menicucci}, \citenamefont {Flammia},\ and\ \citenamefont {van
  Loock}}]{Menicucci2011}%
  \BibitemOpen
  \bibfield  {author} {\bibinfo {author} {\bibfnamefont {N.~C.}\ \bibnamefont
  {Menicucci}}, \bibinfo {author} {\bibfnamefont {S.~T.}\ \bibnamefont
  {Flammia}},\ and\ \bibinfo {author} {\bibfnamefont {P.}~\bibnamefont {van
  Loock}},\ }\bibfield  {title} {\bibinfo {title} {{Graphical calculus for
  {Gaussian} pure states}},\ }\href
  {http://dx.doi.org/10.1103/PhysRevA.83.042335} {\bibfield  {journal}
  {\bibinfo  {journal} {Phys. Rev. A}\ }\textbf {\bibinfo {volume} {83}},\
  \bibinfo {pages} {042335} (\bibinfo {year} {2011})}\BibitemShut {NoStop}%
\bibitem [{\citenamefont {Grimsmo}\ \emph {et~al.}(2020)\citenamefont
  {Grimsmo}, \citenamefont {Combes},\ and\ \citenamefont
  {Baragiola}}]{Grimsmo2020}%
  \BibitemOpen
  \bibfield  {author} {\bibinfo {author} {\bibfnamefont {A.~L.}\ \bibnamefont
  {Grimsmo}}, \bibinfo {author} {\bibfnamefont {J.}~\bibnamefont {Combes}},\
  and\ \bibinfo {author} {\bibfnamefont {B.~Q.}\ \bibnamefont {Baragiola}},\
  }\bibfield  {title} {\bibinfo {title} {Quantum computing with
  rotation-symmetric bosonic codes},\ }\href
  {https://doi.org/10.1103/PhysRevX.10.011058} {\bibfield  {journal} {\bibinfo
  {journal} {Phys. Rev. X}\ }\textbf {\bibinfo {volume} {10}},\ \bibinfo
  {pages} {011058} (\bibinfo {year} {2020})}\BibitemShut {NoStop}%
\bibitem [{\citenamefont {Colbourne}\ and\ \citenamefont
  {Dinitz}(2007)}]{colbourne2007handbook}%
  \BibitemOpen
  \bibfield  {author} {\bibinfo {author} {\bibfnamefont {C.}~\bibnamefont
  {Colbourne}}\ and\ \bibinfo {author} {\bibfnamefont {J.}~\bibnamefont
  {Dinitz}},\ }\href
  {https://www.scopus.com/record/display.uri?eid=2-s2.0-85058097239&origin=inward&txGid=90d5e71c4fb892ba5eaf197191e9b7b3}
  {\emph {\bibinfo {title} {Handbook of combinatorial designs}}}\ (\bibinfo
  {publisher} {CRC press Boca Raton, FL},\ \bibinfo {year} {2007})\BibitemShut
  {NoStop}%
\bibitem [{\citenamefont {Horadam}(2007)}]{Horadam2007}%
  \BibitemOpen
  \bibfield  {author} {\bibinfo {author} {\bibfnamefont {K.~J.}\ \bibnamefont
  {Horadam}},\ }\href {https://doi.org/10.1515/9781400842902} {\emph {\bibinfo
  {title} {Hadamard Matrices and Their Applications}}}\ (\bibinfo  {publisher}
  {Princeton University Press},\ \bibinfo {address} {Princeton},\ \bibinfo
  {year} {2007})\ pp.\ \bibinfo {pages} {1--263}\BibitemShut {NoStop}%
\bibitem [{\citenamefont {Hirasaka}\ \emph {et~al.}(2016)\citenamefont
  {Hirasaka}, \citenamefont {Kim},\ and\ \citenamefont {Yu}}]{hirasaka2016}%
  \BibitemOpen
  \bibfield  {author} {\bibinfo {author} {\bibfnamefont {M.}~\bibnamefont
  {Hirasaka}}, \bibinfo {author} {\bibfnamefont {K.}~\bibnamefont {Kim}},\ and\
  \bibinfo {author} {\bibfnamefont {H.}~\bibnamefont {Yu}},\ }\bibfield
  {title} {\bibinfo {title} {Isomorphism classes of association schemes induced
  by hadamard matrices},\ }\href
  {https://doi.org/https://doi.org/10.1016/j.ejc.2015.04.010} {\bibfield
  {journal} {\bibinfo  {journal} {European Journal of Combinatorics}\ }\textbf
  {\bibinfo {volume} {51}},\ \bibinfo {pages} {37} (\bibinfo {year}
  {2016})}\BibitemShut {NoStop}%
\end{thebibliography}%
\end{document}